%% file: main.tex
\documentclass[a4paper,USenglish,cleveref, autoref, thm-restate]{lipics-v2021}
\usepackage{graphicx} 
\usepackage{caption} 
\usepackage{algorithm}
\usepackage{algorithmic}

\usepackage{amsmath,amsthm,amssymb}
\usepackage{mathtools}
\usepackage{dsfont}
\usepackage{xcolor}

\newtheorem{problem}{Problem}
\usepackage{thmtools}
\usepackage{thm-restate}

\usepackage{hyperref}

\usepackage{cleveref}

\usepackage{csquotes}
\MakeOuterQuote{"}

\usepackage[colorinlistoftodos,prependcaption]{todonotes}
\usepackage{xargs}                      
\usepackage[pdftex,dvipsnames]{xcolor}  
\newcommandx{\issue}[2][1=]{\todo[linecolor=red,backgroundcolor=red!25,bordercolor=red,#1]{#2}}

\nolinenumbers

\title{Maximizing Diversity in (near-)Median String Selection}

\author{Diptarka Chakraborty}{National University of Singapore}{diptarka@comp.nus.edu.sg}{}{}

\author{Rudrayan Kundu\footnote{A part of the work was done when the author was doing an internship at the National University of Singapore}}{Indian Statistical Institute Kolkata}{rajasreekundu@gmail.com}{}{}

\author{Nidhi Purohit\footnote{A major part of the work was done when the author was a research fellow at the National University of Singapore}}{Institute of Mathematical Sciences, Chennai, India}{nidhipurohit95@gmail.com}{}{}

\author{Aravinda Kanchana Ruwanpathirana\footnote{The work was done when the author was a research fellow at the National University of Singapore}}{Nanyang Technological University}{kanchana.ruwanpathirana@gmail.com}{}{}

\authorrunning{D.~Chakraborty, R.~Kundu, N.~Purohit, and A.K.~Ruwanpathirana} 

\Copyright{Diptarka Chakraborty, Rudrayan Kundu, Nidhi Purohit, and Aravinda Kanchana Ruwanpathirana} 



\category{} 

\relatedversion{} 

\funding{This work was supported by an MoE AcRF Tier 1 grant (T1 251RES2303).}

\begin{document}
\include{header}
\maketitle

\begin{abstract}
Given a set of strings over a specified alphabet, identifying a median or consensus string that minimizes the total distance to all input strings is a fundamental data aggregation problem. When the Hamming distance is considered as the underlying metric, this problem has extensive applications, ranging from bioinformatics to pattern recognition. However, modern applications often require the generation of multiple (near-)optimal yet diverse median strings to enhance flexibility and robustness in decision-making.

In this study, we address this need by focusing on two prominent diversity measures: \emph{sum dispersion} and \emph{min dispersion}. We first introduce an exact algorithm for the \emph{diameter} variant of the problem, which identifies pairs of near-optimal medians that are maximally diverse. Subsequently, we propose a $(1-\epsilon)$-approximation algorithm (for any $\epsilon >0$) for sum dispersion, as well as a bi-criteria approximation algorithm for the more challenging min dispersion case, allowing the generation of multiple (more than two) diverse near-optimal Hamming medians. Our approach primarily leverages structural insights into the Hamming median space and also draws on techniques from error-correcting code construction to establish these results.
\end{abstract}


\newpage
\pagenumbering{arabic}

\section{Introduction}
\input{intro}

\section{Preliminaries}\label{sec:prelim}
\input{prelim}
\section{Exact Algorithms for Diameter Maximization}\label{sec:diameter}
\input{diameter}

\section{Maximizing the Sum Dispersion}\label{sec:sum}
\input{sum-dispersion}
\section{Maximizing the Minimum Dispersion}\label{sec:min}
\input{min-dispersion}

\section{Discussion and Future Work}
This paper initiates the study of computing a diverse set of medians in the Hamming metric using two classical dispersion objectives: sum dispersion and minimum dispersion. First, we present an exact algorithm for the diameter variant, which outputs two near-medians with maximum diversity. Second, we address the task of producing multiple (near-)medians and give a PTAS for maximizing sum dispersion. Third, we develop a bi-criteria approximation algorithm for maximizing minimum dispersion.

For the minimum-dispersion objective with $k$ approximate medians, there remains a gap in the regime $\omega(1) \le D^* \le O(\log k)$, where $D^*$ denotes the optimal diameter. In this range, we do not have a polynomial-time approximation algorithm; our $1/2$-approximation runs in quasipolynomial time (more specifically, $(d/D^*)^{O(\log k)} poly(ndk)$ time) instead. Designing a polynomial-time algorithm with a comparable approximation guarantee in this regime is an immediate open problem. Another open direction is to obtain an approximation factor solely on the dispersion objective (instead of bi-criteria trade-offs). Finally, extending diverse median computation to other metric spaces -- such as Euclidean, edit, Jaccard, and Kendall–tau -- is an interesting avenue for future work.

\bibliography{ref}

\newpage

\appendix

\input{appendix}

\end{document}

%% file: header.tex
\newcommand{\freq}[2]{f_{#1}^{#2}}
\newcommand{\ind}[1]{\mathds{1}(#1)}
\newcommand{\mcc}{\textnormal{\textsc{mfc}}}
\newcommand{\mdiff}{\textnormal{\textsc{Min-Diff Partition}}}
\newcommand{\sym}{\mathcal{S}}
\newcommand{\ral}{\mathbb{R}}
\newcommand{\exv}{\mathbb{E}}
\renewcommand{\epsilon}{\varepsilon}
\newcommand{\opt}{\textnormal{\textsc{opt}}}
\newcommand{\costind}{\textnormal{\textsc{index-cost}}}
\newcommand{\sdahm}{\textnormal{\textsc{Sum-Dispersion-Approx-Median}}}
\newcommand{\sdalg}{\textnormal{\textsc{Sum-Dispersion-k-Strings}}}
\newcommand{\sdelg}{\textnormal{\textsc{Sum-Dispersion-Exact}}}
\newcommand{\mdehm}{\textnormal{\textsc{Min-Dispersion-Median}}}
\newcommand{\dmhm}{\textnormal{\textsc{Diameter-Maximizing-Median}}}
\newcommand{\cgrd}{\textnormal{\textsc{Cost-Greedy}}}
\newcommand{\ilpApp}{\textnormal{\textsc{ILP-Min-Dispersion}}}

\newcommand{\sDp}{\textnormal{\textsc{sumDp}}}
\newcommand{\mDp}{\textnormal{\textsc{minDp}}}

%% file: intro.tex
In classical optimization problems, the goal is to find an optimal or nearly optimal solution for a given input instance. However, these solutions may not align with the preferences of certain users due to subjective factors like economic considerations, political views, environmental concerns, aesthetics, and more, which the algorithm might not account for. For some users, their personal preferences may outweigh the importance of achieving the optimal solution, making certain approximate solutions more desirable. For instance, a user might favor a cost-effective near-optimal solution over an expensive optimal one due to financial limitations, or an energy-efficient near-optimal solution over an optimal option because of environmental concerns. To address this, it is beneficial to offer users a range of optimal or near-optimal solutions, allowing them to select based on their specific requirements, which the algorithm may not initially know. However, if the solutions provided are too similar, they fail to offer genuine alternatives, undermining the purpose of presenting multiple options. This necessitates the study of returning diverse solutions -- specifically, multiple solutions that are far or significantly dissimilar with respect to certain measures in the solution space.

In recent years, there has been an increasing interest in exploring various optimization problems through the perspective of generating diverse solutions~\cite{petit2015finding, vadlamudi2016combinatorial, petit2019enriching, ingmar2020modelling, baste2022diversity, gao2022obtaining,hanguir2025optimizing}. Diverse solutions are crucial in scenarios where decision-making flexibility, robustness against uncertainty, and the ability to encompass multiple viewpoints are vital. For instance, in bioinformatics, particularly in areas like gene motif identification, generating diverse solutions allows for the consideration of multiple motifs, thereby facilitating the investigation of several hypotheses. A range of problems have been studied with the aim of producing multiple diverse solutions, including satisfiability~\cite{nadel2011generating, misra2024parameterized, austrin2025algorithms}, constraint programming~\cite{hebrard2005finding}, hitting set~\cite{baste2019fpt}, longest common susequence~\cite{shida2024diverseLCS}, matching~\cite{fominpetr2024diverse, fomin2024diverse}, shortest paths~\cite{hanaka2022computing}, minimum cut~\cite{de2023finding}, feedback vertex set~\cite{baste2019fpt}, rank aggregation~\cite{arrighi2021diversity}, and spanning tree~\cite{hanaka2021finding}.

Computing a representative of a given data set is one of the most fundamental computational data summarization tasks. In a widely recognized variant of this problem, given a set $S$ of data points coming from an underlying metric space $\mathcal{X}$, the objective is to find a point (not necessarily from $S$) that minimizes the sum of distances to the points in $S$. The problem is referred to as \emph{median} (or \emph{geometric median}) problem. The complexity of the problem varies with the underlying metric space. In this paper, we consider the median problem over the well-known \emph{Hamming metric}. Hamming distance, which counts the number of coordinatewise dissimilarities between a pair of strings, is perhaps the most primitive distance measure defined over strings. In other words, Hamming distance measures the minimum number of character substitutions required to convert one string into another, which is the same as the $\ell_1$ distance over binary strings. The problem of computing median string under Hamming finds a wide range of applications, ranging from bioinformatics in applications such as gene motif classification~\cite{Kaysar20Gene, pevzner2000computational}, classification tasks in pattern recognition~\cite{Juan00Use}, and coding theory~\cite{frances1997covering}. It is folklore to compute the Hamming median exactly in linear time. In many applications of Hamming median, we often need to produce a diverse set of solutions; for instance, in selecting a small and diverse set of prototype strings/vectors, in designing diverse consensus sequences capturing different clades/subtypes (e.g., pathogen panels), in choosing query/test inputs that are both representative and diverse across binary feature regions, etc. Despite being a fundamental problem of significant importance, the Hamming median problem has not yet been examined with a focus on diversity. In this paper, we initiate the systematic study of generating diverse (approximate) medians.

Among measures used to quantify diversity, two of the most prominent ones are \emph{sum-dispersion} -- the sum of all pairwise distances, and \emph{min-dispersion} -- the minimum pairwise distance, where the notion of distance depends on the underlying metric space. In diversity maximization, the objective is to maximize either the sum or the min dispersion. In the classical dispersion problem, given a set of points in a metric space, the task is to choose a subset of a specific (input-specified) size that maximizes the sum/min dispersion. The problem is already known to be NP-hard for general metric~\cite{Ravi1994Heuristic, abbar2013diverse}. For the min dispersion under a general metric, $1/2$-approximation is known~\cite{Ravi1994Heuristic, HASSIN1997133}, which is also tight. On the other hand, for the sum dispersion, a $1/2$-approximation is known for general metric~\cite{HASSIN1997133, Birnbaum2006amproved}, and it is also known to be tight under the Exponential Time Hypothesis~\cite{gao2022obtaining}. The min dispersion problem remains NP-hard even when the underlying metric space is the Hamming metric~\cite{shida2024diverseLCS}, which is the space considered in this work. 

The problem becomes much more challenging when the candidate points are not explicitly given, as in diversity variants of many optimization problems, and becomes especially difficult when the solution space is exponential. In the Hamming median problem, the solution set is implicit; moreover, while the optimal median may be unique, the set of approximate medians can be exponential in the length of the strings, making the dispersion problem over this search space computationally more difficult. In many practical scenarios, it is sufficient to output a diverse collection of approximate solutions when the optimal solution is unique, or there are only a few distinct optimal solutions, necessitating the study of generating diverse near-optimal solutions.

\paragraph*{Our Contribution} 
In this paper, we initiate the study of finding diverse (approximate) median strings under the Hamming metric. Let $\Gamma$ be the alphabet set. Given a set of $n$ strings over alphabet $\Gamma$, each of length $d$, the goal is to compute a set of diverse strings -- measured with respect to both sum dispersion and minimum dispersion -- that are $(1+\epsilon)$-approximate medians (for $\epsilon \ge 0$) of the input dataset.

\subparagraph*{Maximizing diameter.} We start by considering the problem of finding just two (approximate) medians that are as diverse as possible. In the literature, the maximum Hamming distance between two candidate solutions is referred to as the \emph{diameter}. When focusing on exact medians (i.e., where the solutions are required to be exact medians), it is relatively straightforward to optimally solve the diameter variant. Specifically, we can efficiently find two Hamming medians that are maximally diverse (see~\autoref{thm:diverse-fin-exact} in~\autoref{app:diverse-fin-exact}) -- thanks to the special structure of the space of Hamming medians. However, this same structure often leads to the Hamming median being unique, which precludes the possibility of constructing a diverse set of exact Hamming medians, even of size two.

In many practical settings, it is sufficient to work with near-medians (i.e., $(1+\epsilon)$-approximate medians for some small $\epsilon >0$). Consequently, when the aim is to generate diverse solutions and diversity is prioritized over optimality of the underlying solutions, a natural question arises: can we find two approximate medians that are as diverse as possible? We answer this question in the affirmative by presenting an efficient algorithm that optimally solves the diameter problem for approximate Hamming medians.

\begin{restatable}{theorem}{diversefin}
\label{thm:diverse-fin}
Consider an alphabet $\Gamma$. There exists an algorithm that, given any $X \subseteq \Gamma^d$ of size $n$ and $\epsilon > 0$, outputs two $(1+\epsilon)$-approximate Hamming medians of $X$ with maximum diameter, and runs in time $O((1+\epsilon)nd + d\log d)$.
\end{restatable}

\subparagraph*{Maximizing sum dispersion.}
Next, we turn our attention to the task of finding multiple -- potentially more than two -- (approximate) medians that maximize specific diversity measures. As mentioned earlier, one widely used diversity measure is the \emph{sum dispersion}, where the objective is to maximize the sum of pairwise Hamming distances among the selected solutions. When generating multiple Hamming (exact) medians, it is still feasible -- though somewhat more intricate than the diameter variant -- to achieve the maximum sum dispersion (see~\autoref{thm:sum-dispersion-k-exact}).

The challenge increases when the goal is to compute a diverse set of approximate medians. For approximate medians, we provide a PTAS for the sum dispersion objective.

\begin{restatable}{theorem}{sumdispersionk}
\label{thm:sum-dispersion-k}
Consider an alphabet $\Gamma$. There exists a polynomial-time algorithm that, given any $X \subseteq \Gamma^d$ of size $n$, a non-negative integer $k$ and $\epsilon,\delta > 0$, returns a set $S$ of $k$ $(1+\epsilon)$-approximate Hamming medians, such that their sum dispersion $\sDp(S) \ge (1-\delta)v^*$, where $v^*$ is the maximum sum dispersion of $k$ many $(1+\epsilon)$-approximate Hamming medians.
\end{restatable}

\subparagraph*{Maximizing min dispersion.}
The problem becomes more challenging when it comes to \emph{min dispersion}. The goal here is to produce $k$ (approximate) Hamming median strings that maximize the minimum pairwise distance. In contrast to diameter and sum dispersion, this task is computationally more challenging even when seeking exact medians.


In this paper, we provide an efficient algorithm to generate $k$ Hamming medians while approximating the maximum min dispersion. First, we note that when $k = O(1)$, the problem can be solved optimally using standard dynamic programming in polynomial time. So from now, we assume $k = \Omega(1)$.
We show the following result, a formal statement of which appears in~\autoref{sec:min}.

\begin{theorem}[Informal Statement]
    Given a $X \subseteq \Gamma^d$, a non-negative integer $k$, and a $\delta > 0$,
    \begin{itemize}
        \item If the optimal diameter of $X$ is $D^* \ge \Omega\left( \frac{1}{\delta^2} \log k\right)$, a set of $k$ Hamming medians can be generated in polynomial time, which gives a $(1-\delta)$-approximation to the maximum min dispersion with high probability;
        \item If $D^* \le O\left( \frac{1}{\delta^2} \log k\right)$, a set of $k$ Hamming medians can be generated in polynomial time, which gives a $1/2$-approximation to the maximum min dispersion.
    \end{itemize}
\end{theorem}

Next, we consider the min dispersion problem over approximate medians. In this case, we provide a bi-criteria algorithm. In particular, we present the following result, a formal statement of which appears in~\autoref{sec:min}.

\begin{theorem}[Informal Statement]
    Given a $X \subseteq \Gamma^d$, a non-negative integer $k$, and $\epsilon, \delta > 0$,
    \begin{itemize}
        \item If $D^* \le O(1)$, a set of $k$ many $(1+\epsilon)$-approximate Hamming medians can be generated in polynomial time, with min dispersion at least $t^*/2$;
        \item If $D^* \ge \Omega\left( \frac{1}{\delta^2} \log k\right)$, a set of $k$ many $(1+2\epsilon)$-approximate Hamming medians can be generated in polynomial time, with min dispersion at least $(1/2 - \delta)t^*$ with high probability,
    \end{itemize}
    where $D^*$ denotes the optimal diameter for $(1+\epsilon)$-approximate Hamming medians of $X$ and $t^*$ denotes the maximum min dispersion of a set of $k$ many $(1+\epsilon)$-approximate Hamming medians of $X$.
\end{theorem}

Note that only in the second point of the above result, we attain a bi-criteria approximation. We would also like to highlight that if $t^* \ge \Omega\left( \frac{1}{\delta} \sqrt{d} \log k\right)$, we can even get a better bi-criteria bound; more specifically, we output a set of $(1+\epsilon +\delta)$-approximate medians with $(1/2-\delta)$-approximation to the min dispersion objective, for any $\delta >0$.

\paragraph*{Related Works}
The max-min and max-sum dispersion problems have been studied in the classical setting, where the goal is to select a subset of points from a given finite metric space. For max-min dispersion, constant-factor approximations are known, including a $1/2$-approximation~\cite{Ravi1994Heuristic, HASSIN1997133}, which is known to be tight. For max-sum dispersion, known approximation algorithms achieve a $1/2$-approximation in general metrics~\cite{HASSIN1997133, Birnbaum2006amproved}, with the ratio improving to $0.634 - \epsilon$ in the 2D Euclidean space~\cite{Ravi1994Heuristic}. Moreover, under the Exponential Time Hypothesis, no polynomial-time algorithm can approximate $k$-sum dispersion in general metric spaces within a factor better than $1/2$~\cite{gao2022obtaining}. The problem remains NP-hard in the Hamming metric, the space we consider in this paper. Better approximation results exist in more structured domains, such as a PTAS for negative-type metrics with matroid constraints~\cite{cevallos2015maxsum} and for bounded doubling-dimension metrics~\cite{cevallos2018diversity}. In addition, frameworks have been developed to approximate diversity~\cite{hanaka2023framework} or to simultaneously guarantee approximate optimality and diversity via bi-criteria reductions to certain budget-constrained problems~\cite{gao2022obtaining}.

Recent work has extended dispersion to broader combinatorial structures, including NP-hard problems such as knapsack, vertex cover, and independent set, with provable approximation guarantees while ensuring a high level of diversity under symmetric-difference as the diversity measure~\cite{galvez2025framework}. In graphs and matroids, diverse bases, independent sets, and matchings have been studied, yielding NP-hardness results and fixed-parameter tractable algorithms~\cite{fomin2024diverse}. Dispersion has also been examined for the Longest Common Subsequence problem, allowing polynomial-time exact algorithms when the number of diverse subsequences is bounded and a PTAS for the max-sum variant~\cite{shida2024diverseLCS}. In satisfiability and related NP-complete problems, the diverse-$k$-SAT problem has been studied, yielding improved exponential-time algorithms and randomized approximations for both min-dispersion and sum-dispersion objectives~\cite{austrin2025algorithms}.

\subparagraph*{Pareto Optimality.} When optimization involves multiple criteria, there may not be a single solution that is optimal for all objectives. In this context, Pareto-optimal solutions have been studied: these are solutions where no objective can be improved without worsening another. Research has focused on computing approximate Pareto fronts. In particular,~\cite{Papadimitriou2000On} shows that for any multicriteria optimization problem, there exists a polynomial-size set of Pareto-optimal solutions such that each objective is satisfied up to a factor of $(1+\epsilon)$, and this set can be computed in polynomial time provided a gap version of the problem can be solved. Further,~\cite{Hezel2021One,Hezel2021Approx} show that for a class of problems in which a dual exists for a restricted version of a budget-constrained optimization problem, if the dual can be solved in polynomial time, it is possible to obtain $(1+\epsilon)$-approximations for all objectives except one, which can be optimized exactly.

\paragraph*{Technical Overview}
We start by addressing the challenge of identifying a set of diverse Hamming (exact) medians for a given dataset. We first demonstrate that by exploiting the special structure inherent in the space of all Hamming medians, it is possible to efficiently find diverse Hamming medians, either optimally or through approximation. Subsequently, we discuss methods for obtaining a collection of diverse approximate medians.

\noindent \textbf{Warm-up with Exact medians: Maximizing diameter and dispersions.}
Suppose we are given a dataset $X \subseteq \Gamma^d$, and let $w$ be a Hamming median of $X$. By a simple observation (\autoref{lem:MCC-opt}), for every index $i$, the character $w_i$ at that position in a Hamming median must be one of the most frequent characters at index $i$ in the dataset $X$. Therefore, if there are two distinct Hamming medians, they can only differ at those positions where multiple characters are tied for the highest frequency. This observation facilitates the construction of the two most diverse Hamming medians: for each such index, select two different most frequent characters and assign them to these positions in each median, thereby creating two distinct medians.

The scenario becomes less straightforward when the goal is to generate $k$ Hamming medians that maximize the sum of all pairwise Hamming distances (i.e., maximize sum dispersion). Still, the problem remains quite manageable. Since the objective is to maximize the sum, we can focus on optimizing it coordinate-wise. As before, the only indices that contribute to the sum dispersion among Hamming medians are those with multiple most frequent characters. More specifically, for any such index $i$, suppose there are $r$ most frequent characters, say $a_1,a_2,\cdots,a_r$. Then, the sum dispersion at that index is maximized if each character $a_j$ appears exactly $k/r$ times among the $k$ median strings (assuming $r$ divides $k$; otherwise, the frequencies should be distributed as evenly as possible). It is again easy to generate $k$ Hamming medians that respect the above property, leading to a maximum sum dispersion. For completeness, we provide the details in~\autoref{app:sum-exact}.

The problem becomes much more difficult when it comes to maximizing the min dispersion, i.e., we want to generate $k$ Hamming medians that maximize the minimum pairwise Hamming distances. The computational hardness arises from its innate connection with the minimum distance problem, one of the fundamental questions in error-correcting codes. We begin by giving a dynamic programming algorithm to solve the problem exactly when $k = O(1)$ (\autoref{lem:min-dispersion-DP} in~\autoref{app:min-dispersion-exact}). Thus, from now on, we assume that $k \ge \Omega(1)$. Next, we consider the following two cases separately: (I) diameter is $\Omega(\log k)$, and (II) diameter is $O(\log k)$. 

For Case (I), we achieve a $(1-\delta)$-approximation to the min dispersion (note, for brevity, we hide the dependency on $1/\delta^2$ in the above $\Omega(\cdot)$ and $O(\cdot)$ notation). Recall that Hamming median requires each index to take only one of the most frequent characters at that index in the input set. Thus, we first compute the set of most frequent characters for each index, denoted as $\Gamma_i$ for index $i$ ($\Gamma_i$ may contain a single element). Next, we construct a set of $k$ Hamming medians as follows: For each index $i$, we select a character from $\Gamma_i$ independently and uniformly at random. We repeat this process to generate $k$ strings. To argue that this yields a $(1-\delta)$-approximation to the min dispersion with high probability, we first apply standard concentration bounds to get a lower bound on the min dispersion of the output set. Then we establish a generalized version of the Plotkin bound for codes with potentially different alphabets in each index (see~\autoref{sec:plotkin}). Finally, by combining the lower bound on the min dispersion of the output set and the generalized Plotkin bound, we get our desired approximation guarantee. For Case (II), we observe that the set of all candidate Hamming medians is polynomially bounded, and we then apply the classical greedy algorithm for min dispersion from~\cite{Ravi1994Heuristic}, albeit paying a $1/2$-approximation to the min dispersion objective. We refer to~\autoref{app:min-dispersion-exact} for the details.

\noindent \textbf{Approximate medians: Maximizing diameter.}
Due to its special structure, the Hamming median is often unique for an input set, particularly when there is a unique majority character at each position. This uniqueness limits the possibility of having multiple diverse medians. However, this restriction does not generally apply to approximate medians, even if we consider an approximation factor of $(1+\epsilon)$, for any small $\epsilon >0$. For instance, consider an input set $X$ containing $n$ binary strings each of length $d$ such that for every index, the most frequent character is 1 and it occurs in $n/2 + 1$ strings (i.e., frequency is $n/2+1$). Here, the Hamming median is unique (the all-one string). However, it is easy to observe that even if we consider the all-zero string ($0^d$), it is still an $(1+\epsilon)$-approximate median, for $\epsilon \approx 8/n$, and thus the maximum distance between two $(1+\epsilon)$-approximate medians (referred to as diameter) could be as large as $d$.

In this paper, we present an exact algorithm for this diameter variant. Specifically, we design an algorithm that constructs two $(1+\epsilon)$-approximate medians that are maximally distant from each other. Our approach begins with the string formed by taking the most frequent (majority) character at each position (breaking ties arbitrarily). For each coordinate, we then consider the second most frequent character and assign a weight corresponding to the increase in the median objective if it were chosen instead; this weight is the frequency difference between the most and second most frequent characters. Next, we use a greedy strategy to select a maximal subset $T$ of positions so that the sum of these weights does not exceed $2\epsilon \opt$, where $\opt$ denotes the minimum Hamming median objective. Next, we partition this set $T$ in a balanced manner -- as evenly as possible into two subsets, $T_1$ and $T_2$, so that the difference in their total weights is minimized. We then output two strings: $z$, which uses the second most frequent characters at indices in $T_1$ and the most frequent elsewhere; and $y$, which uses the second most frequent characters at indices in $T_2$ and the most frequent elsewhere. Intuitively, both $y$ and $z$ have a median objective cost of at most $(1+\epsilon)\opt$ (due to the balanced partitioning), and since $ T_1$ and $ T_2$ are disjoint, they are maximally apart. To ensure both $y$ and $z$ are $(1+\epsilon)$-approximate medians and they realize the maximum diameter (not even off by a small factor), we introduce some additional refinements in the above selection process, leading to a more intricate analysis, which is detailed in~\autoref{app:diam-approx}.

\noindent \textbf{Approximate medians: Maximizing sum dispersion.}
Next, we turn our attention to generating $k$ (more than two) $(1+\epsilon)$-approximate medians with the aim of maximizing the sum dispersion measure, denoted as $\sDp$. This introduces new challenges, particularly in adapting and extending the previous approach used for the diameter variant in the case of approximate medians and for maximizing sum dispersion with exact medians. In the case of the diameter, since only two strings are produced, it is sufficient to focus on the two most frequent characters at each position. However, for generating $k$ approximate medians, we may need to consider more than two characters for each index -- potentially all characters in $\Gamma$ -- each associated with a different weight. Recall that the weight of a character at a particular index corresponds to the increase in the median objective if it is selected over the most frequent character.

Furthermore, it is no longer enough to simply identify a set of index positions where replacing the most frequent symbol with another does not increase the median cost by more than $k \epsilon \opt$, and then "distribute equitably" these positions among the $k$ candidate median strings. First, achieving a balanced partition into $k$ groups is hard, especially since $k$ can be arbitrarily large. Second, and perhaps more importantly, even if such a balanced partitioning is efficiently done and we then make changes to the assigned index set for each of the $k$ candidate medians, this does not guarantee maximization of $\sDp$. For exact medians, we have already observed that maximizing $\sDp$ coordinate-wise requires using as many distinct symbols as possible, distributed as evenly as possible across the candidate median strings. Now, since each symbol at a given position has potentially different weight, we face a trade-off: whether to increase the number of indices where candidates deviate from the most frequent symbol (thus different from the exact median) or to maximize diversity at certain index positions. In essence, the problem now involves meeting $k$ separate "hard" budget constraints (each having a budget of at most $\epsilon \opt$), while still striving to maximize the overall $\sDp$ objective.

To address the challenges outlined above, we first present a $(1-4/D^*)$-approximation algorithm, where $D^*$ represents the optimal diameter for $(1+\epsilon)$-approximate medians in the given input set. For any $\delta > 0$, our algorithm directly yields a $(1-\delta)$-approximation whenever  $D^* > 4/\delta$. Otherwise, it is not hard to observe that there are only a polynomial number of possible candidates for the $(1+\epsilon)$-approximate median strings. We argue that the Hamming metric is of "negative-type" and that the constraint of selecting $k$ strings reduces to a matroid constraint. Then, by applying the result from~\cite{cevallos2015maxsum} (which essentially involves rounding a quadratic program), we obtain a PTAS. Next, we briefly outline the main ideas behind the $(1-4/D^*)$-approximation algorithm.

We begin with a set $S$ of $k$ identical exact median strings, all equal to $w$; we call $S$ the set of candidate medians. We then consider a collection of modification operations. Each operation is specified by: an index position $i$, a character $a \in \Gamma$, an integer $r$ -- the target frequency of $a$ at position $i$, and an integer $\ell$ -- the frequency of the character $w_i$ at position $i$ among the candidate medians in $S$. We assign each operation a \emph{density}, which informally measures the ratio between its increase in the $\sDp$ objective and its increase in the median objective if applied.

Next, we sort all plausible operations in nondecreasing order of density and attempt to apply the longest prefix of this sorted list that yields a feasible solution. A prefix is feasible if the following holds: we initially allocate a budget of $\epsilon \opt$ to each candidate median in $S$. Then, for each index $i$ and character $a$, we perform all the corresponding operations in the chosen prefix, one by one, on candidate medians that still have positive remaining budget. After each operation, we deduct the appropriate amount from the budget of the candidate median to whom it was applied. If every operation in the prefix can be completed without any candidate median exceeding its budget, the prefix is deemed feasible. In the algorithm, we execute the selected operations in a carefully chosen order and select which candidate medians to modify so as to guarantee the desired approximation.

The approximation guarantee proceeds as follows. Let $v^*$ denote the optimal $\sDp$ achievable by $k$ many $(1+\epsilon)$-approximate medians. First, we note that $v^* \gtrsim k^2 D^*/4$. Next, we derive a lower bound on the $\sDp$ value $v$ attained by our algorithm’s output (the final set of candidate medians). Let $U$ be the total remaining budget across all candidate medians in $S$ at the end of the algorithm, and let $\rho$ be the maximum density among the operations that were not performed. The core of the argument is to establish the two statements: $U \cdot \rho \lesssim k^2$, and $v^* - v\le U \cdot \rho$. Informally, $U \cdot \rho$ upper-bounds the additional $\sDp$ one could gain by executing all leftover operations without violating any budget constraints. The main technical hurdle lies in proving these two bounds (see the proof of~\autoref{lem:opt-v-bound-k}). This immediately yields an approximation factor of $(1-4/D^*)$. Full details appear in~\autoref{app:sum-approx}.

\noindent \textbf{Approximate medians: Maximizing min dispersion.}
For the problem of selecting $k$ approximate Hamming medians while maximizing the min dispersion, just like the case of exact medians, we can get a dynamic programming algorithm in polynomial time for $k=O(1)$. Thus, from now on, we focus on $k \ge \Omega(1)$. Then we split into cases depending on the value of the optimal diameter $D^*$. When $D^* \le O(1)$, the number of candidate Hamming medians is polynomially bounded; in this case, applying the greedy min-dispersion heuristic~\cite{Ravi1994Heuristic} yields a $1/2$-approximation. For $D^* \ge \Omega(\log k)$, we first compute two $(1+\epsilon)$-approximate Hamming medians $y,z$ of distance equal to the diameter (using~\autoref{thm:diverse-fin}). Then for each index $i$, we create an alphabet set $\Gamma_i=\{y_i,z_i\}$. Next, we generate $k$ candidate approximate median strings via the following randomized procedure: For each index $i$, we select a character from $\Gamma_i$ independently and uniformly at random, and repeat this process to form $k$ strings. We first argue that each resulting string is an $(1+2\epsilon)$-approximate median. Using concentration bounds, we establish a lower bound on the min dispersion $\mDp$ for these strings, and combining this with the fact that the optimal min dispersion can at most be the diameter, we obtain $\mDp \gtrsim t^*/2$, where $t^*$ denotes the optimal $\mDp$ achievable by a set of $k$ many $(1+\epsilon)$-approximate medians. This yields a bi-criteria approximation. We can further improve this approximation factor for the "higher regime" of $t^*$. We formulate the problem using an integer linear program and consider its LP relaxation. Then, using a \emph{dependent rounding} framework, we show that we can generate $k$ many $(1+\epsilon +\delta)$-approximate medians (for any $\delta >0$) with $\mDp \ge (1/2 - \delta) t^*$. We detailed the arguments in~\autoref{app:min-dispersion-approx}.

%% file: prelim.tex
\noindent \textbf{Notations.} Let $\Gamma$ denote an alphabet set. For a string $s \in \Gamma^d$, we use $s_i$ to refer to the character at the index $i$ of $s$. Similarly, for any array (or ordered set) $W$, we use $W_i$ to denote the element at index $i$. We use $\ind{.}$ to indicate an identity function where for a logical predicate $t$, $\ind{t}=1$ if $t$ is true; and $0$ otherwise. For any $X \subseteq \Gamma^d$ and a string $s \in \Gamma^d$, let $\freq{i}{s}(X)$ indicate the number of times the character $s_i$ appears at the $i$-th index of strings in $X$, i.e., 
$\freq{i}{s}(X) := \left|\{x\in X : x_i = s_i\}\right| = \sum_{x \in X} \ind{x_i = s_i}.$
For brevity, when clear from the context, we drop $X$ from the above notation and simply use $f_i^s$.

For any $x,y \in \Gamma^d$, their \emph{Hamming distance} is defined as,  $H(x,y) := \sum_{i=1}^d \ind{x_i \neq y_i}.$

\noindent \textbf{Hamming Median.}
Given a set $X \subseteq \Gamma^d$, the \emph{Hamming median} problem asks to find a string $y^* \in \Gamma^d$ that minimizes the sum of distance to the strings in $X$, i.e.,
$y^* = \arg\min_{y \in \Gamma^d} \sum_{x \in X}H(x,y).$
We use $\opt(X)$ (or simply $\opt$ when $X$ is clear from the context) to denote $\sum_{x \in X}H(x,y^*)$. We call a string $y \in \Gamma^d$ an \emph{$\alpha$-approximate median} (for any $\alpha \ge 1$) iff $\sum_{x \in X}H(x,y) \le \alpha \cdot \opt$.

Finding an (exact) median string under the Hamming distance is folklore. Consider the following string: For any $X \subseteq \Gamma^d$, the \emph{most frequent character string}, denoted by $\mcc(X)$, is a string $w$ where $w_i$ is set to be the most frequently (breaking ties arbitrarily) occurred character at the $i$-th index in the strings in $X$, i.e., $w_i = \arg \max_{e \in \Gamma} |\{x\in X : x_i = e\}|$, and in the rest of the paper we use $w$ to refer to $\mcc(X)$.

It is straightforward to see that $w=\mcc(X)$ is a median for the set $X$ under the Hamming distance, as stated in the following result (we provide the proof in~\autoref{sec:folklore-proof}).
\begin{lemma}[Folklore]\label{lem:MCC-opt}
For any $X \subseteq \Gamma^d$, $w=\mcc(X)$ is an optimal median of $X$, i.e. $\sum_{x \in X}H(x,w) = \opt$. Furthermore, for any optimal solution $w^*$, $\freq{i}{w}=\freq{i}{w^*}$, for all $i \in [d]$.
\end{lemma}

Next, we show how the cost of any string can be related to the optimal, the proof of which is deferred to~\autoref{sec:folklore-proof}.
\begin{lemma}\label{lem:opt-offset}
For any $X \subseteq \Gamma^d$, let $w=\mcc(X)$. Then, for any $s \in \Gamma^d$, we can express its objective cost using $w$ and $\opt$ as,
$\sum_{x \in X} H(x,s)=\opt+\sum_{i:s_i\neq w_i} (\freq{i}{w}-\freq{i}{s}).$
\end{lemma}

\noindent \textbf{Dispersion Measures.} Dispersion is the notion of computing $k \ge 2$ diverse solutions to the Hamming median problem. There are multiple ways we could define the dispersion of a set of strings. In this work, we consider two common forms of dispersion, minimum Hamming distance (\emph{min dispersion}) and sum of pairwise Hamming distances (\emph{sum dispersion}). We formally define the \emph{min dispersion} and \emph{sum dispersion} as follows:

\begin{definition}[Min Dispersion]
    Given a set of strings $S=\{s_1,s_2,\dots,s_m\}$, the \emph{min dispersion} of $S$ is defined as,
$\textsc{minDp}(S) := \min_{s_i,s_j\in S} H(s_i,s_j).$
\end{definition}

\begin{definition}[Sum Dispersion]
    Given a set of strings $S=\{s_1,s_2,\dots,s_m\}$, the \emph{sum dispersion} of $S$ is defined as,
$\textsc{sumDp}(S) := \sum_{s_i,s_j\in S|i<j} H(s_i,s_j).$
\end{definition}

\noindent \textbf{Diverse Hamming (Approximate) Median.}
In this paper, we explore three key problems. We first introduce the $\dmhm$ problem, where the goal is to find two (approximate) median strings such that the Hamming distance between them is maximized.

\begin{problem}[Diameter Maximization ($\dmhm$)]\label{prob:diameter-max}
Given a set of strings $X \subseteq \Gamma^d$ and $\epsilon \ge 0$, the $\dmhm$ problem asks to find two $(1+\epsilon)$-approximate Hamming medians $s_1,s_2 \in \Gamma^d$ of $X$ such that $H(s_1,s_2)$ is maximized.
\end{problem}

Next, the $\sdahm$ problem is to find a $k$ set of $(1+\epsilon)$-approximate Hamming medians that maximize the sum dispersion.

\begin{problem}[Sum Dispersion Approximate Medians ($\sdahm$)]\label{prob:Hamming-sum}
Given a set of strings $X \subseteq \Gamma^d$, a non-negative integer $k$, and $\epsilon \ge 0$, the \emph{Sum Dispersion Approximate Hamming medians} problem asks to find a set of (cardinality $k$) strings, $S=\{s_1,s_2,\dots,s_k\} \in \Gamma^d$ such that for all $i \in [k]$, $s_i$ is a $(1+\epsilon)$-approximate Hamming median of $X$, and the sum dispersion $\sDp(S)$ is maximized.   
\end{problem}

Finally, we explore the $\mdehm$ problem, where the goal is to find $k$ (approximate) Hamming medians that maximize the minimum dispersion.

\begin{problem}[Min Dispersion Hamming Medians ($\mdehm$)]\label{prob:Hamming-exact-min}
Given a set of strings $X \subseteq \Gamma^d$, a non-negative integer $k$, and $\epsilon \ge 0$, the \emph{Min Dispersion Hamming Medians} problem asks to find a set of (cardinality $k$) strings, $S=\{s_1,s_2,\dots,s_k\} \in \Gamma^d$ such that for all $i \in [k]$, $s_i$ is a $(1+\epsilon)$-approximate Hamming median of $X$, and the min dispersion $\mDp(S)$ is maximized. 
\end{problem}

%% file: diameter.tex
In this section, we develop efficient algorithms for the $\dmhm$ problem. We first claim that when $\epsilon = 0$, a straightforward construction yields two exact medians that maximize the diameter, which we defer to the~\autoref{app:diverse-fin-exact}. We then extend this result to the case $\epsilon > 0$, demonstrating that there is still an efficient algorithm that produces two $(1+\epsilon)$-approximate medians achieving maximum diameter.

\diversefin*

\subparagraph*{Algorithm Description.} Suppose we are given $X\subseteq \Gamma^d$ as input, and $n=|X|$. Next, we define an auxiliary string $\hat{w}$ that consists of the second-most frequent character (if it exists) in each position, more specifically: set $\hat{w}_i = \arg\max_{e \in \Gamma \setminus \{w_i\}} |\{x \in X:x_i=e\}|$ (breaking ties arbitrarily). 

Before proceeding with the detailed description of our algorithm, let us introduce the following problem, an optimal solution of which is pivotal in our algorithm. Let us consider the {\mdiff} problem: Given an $n$-length array $M$ and a set $T \subseteq [n]$, the goal is to partition $T$ into two (disjoint) sets $T_1, T_2$ (where $T=T_1 \cup T_2$) such that $\left |\sum_{i \in T_1} M_i - \sum_{i \in T_2} M_i \right|$ is minimized. It is not hard to see that this problem can be solved using a dynamic programming algorithm (\autoref{alg:min-sum-diff}), which we detailed in~\autoref{sec:smallest-sum-diff}.

Let us now describe our algorithm. Our algorithm first sorts the indices of $[d]$ in the non-decreasing order of the value of $\freq{i}{w}-\freq{i}{\hat{w}}$. The algorithm then greedily selects a maximal set $S$ of indices in the sorted order such that $\sum_{i \in S} \left(\freq{i}{w}-\freq{i}{\hat{w}} \right)\le \epsilon \opt$. Then, it again selects another maximal set $R$ greedily starting from the index $|S|+1$ of the sorted order, such that $\sum_{i \in R} \left(\freq{i}{w}-\freq{i}{\hat{w}}\right) \le \epsilon \opt$. We output the strings $z,y$ where the characters of $z$ are the same as $w$ except for the indices in $S$ (where they become the corresponding character in $\hat{w}$) and the characters of $y$ are the same as $w$ except for the indices in $R$ (where they become the corresponding character in $\hat{w}$). We also consider the set $T$, which consists of the first $|S|+|R|+1$ sorted indices. If $\sum_{i \in T} \left(\freq{i}{w}-\freq{i}{\hat{w}} \right)\le 2\epsilon \opt$, we use~\autoref{alg:min-sum-diff} (setting $M$ to be $\{f_i^w-f_i^{\hat{w}}|i \in [n]\}$) to find two partitions of $T$ such that the sum difference between the two partitions is minimized. If $T_1$ and $T_2$ are two partitions such that the cost is $\le (1+\epsilon)\opt$, we use the partition $T_1, T_2$. We output the strings $z,y$ where the characters of $z$ are the same as $w$ except for the indices in $S$ (or $T_1$ when there is a valid partition) and the characters of $y$ are the same as $w$ except for the indices in $R$ (or $T_1$ when there is a valid partition). 

We provide the pseudocode (\autoref{alg:diverse-finite}) along with a detailed analysis in~\autoref{app:diam-approx}.

%% file: sum-dispersion.tex
In this section, we present an approximation algorithm for the $\sdahm$ problem. We first show that when $\epsilon = 0$, a simple construction yields $k$ exact medians that maximize the sum dispersion. 

\begin{restatable}{theorem}{sumdispersionkexact}
\label{thm:sum-dispersion-k-exact}
Consider an alphabet $\Gamma$. There exists an algorithm that, given any $X \subseteq \Gamma^d$ of size $n$ and a non-negative integer $k$, returns $k$ Hamming medians maximizing the sum dispersion, in $O\left(nd\log \left(\min\left\{n,|\Gamma|\right\}\right)+kd\right)$ time.
\end{restatable}

We design the $\sdelg$ Algorithm (\autoref{alg:sum-exact}) that ends up giving $k$ strings that maximize sum dispersion. The main idea behind the algorithm is first to find the set of majority (most frequent) characters at each index, and then distribute them "evenly" over $k$ candidate medians (see~\autoref{fig:sum-exact-index-main}). This will ensure that the sum of pairwise distances between them is maximized. We present the algorithm and detailed analysis in~\autoref{app:sum-exact}.

\begin{figure}[h]
    \centering
    \includegraphics[width=\linewidth]{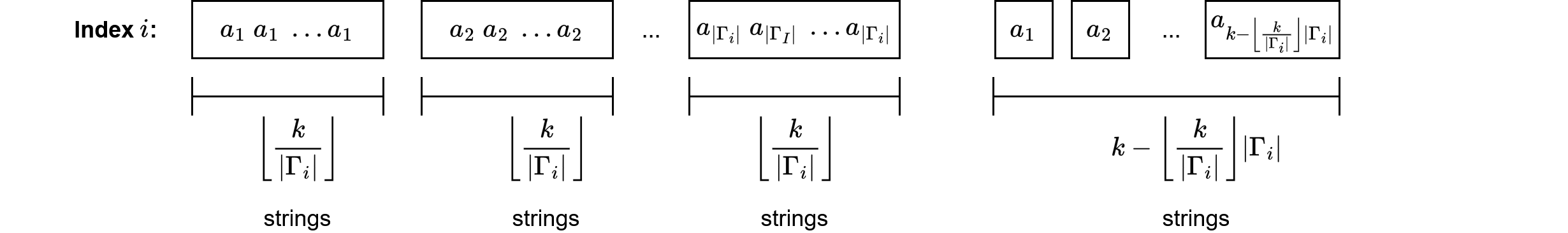}
    \caption{Let $\Gamma_i =\{a_1,\cdots,a_r\}$ be set of most frequent characters at the index $i$. Overview of the characters at index $i$ after the modifications by $\sdelg$ Algorithm (\autoref{alg:sum-exact})}
    \label{fig:sum-exact-index-main}
\end{figure}

We then generalize to the case $\epsilon > 0$, proving that an efficient algorithm can still be obtained to produce $k$ $(1+\epsilon)$-approximate medians that approximately maximize the sum dispersion. Let us now recall~\autoref{thm:sum-dispersion-k}.

\sumdispersionk*

The proof of the above theorem proceeds as follows. Given $X \subseteq \Gamma^d$, let $v^*$ be the maximum possible sum dispersion for any set of $k$ $(1+\epsilon)$-approximate medians, and $D^*$ be the diameter for the $(1+\epsilon)$-approximate medians. We first establish the following result, which is the key contribution towards attaining our approximation result for the sum dispersion. The proof is deferred to~\autoref{app:sum-approx}. 

\begin{restatable}{theorem}{sumdislargeD}
\label{thm:sum-dispersion-largeD}
Consider an alphabet $\Gamma$. There exists an algorithm that, given any $X \subseteq \Gamma^d$ of size $n$, a non-negative integer $k$ and $\epsilon > 0$, returns a set $S$ of $k$ $(1+\epsilon)$-approximate Hamming medians, such that their sum dispersion $\sDp(S) \ge \left(1 - \frac{4}{D^*}\right)v^*$, where $D^*$ is the maximum diameter between two $(1+\epsilon)$-approximate Hamming medians in $X$ and $v^*$ is the maximum sum dispersion of $k$ $(1+\epsilon)$-approximate Hamming medians. Moreover, the algorithm runs in time $O(nd|\Gamma| + d^2k^4|\Gamma|^2)$.
\end{restatable}

We then argue that if $D^*$ is sufficiently large, this already leads to a PTAS. On the other hand, if $D^*$ is small, a PTAS can be obtained by arguing that the Hamming metric is of "negative-type" and that the constraint of selecting $k$ strings reduces to a matroid constraint, and then using the earlier work of~\cite{cevallos2015maxsum}. Together, it completes the proof of~\autoref{thm:sum-dispersion-k}. We defer the details to~\autoref{app:sum-approx}.

%% file: min-dispersion.tex
 In this section, we study the $\mdehm$ problem. We first show that in the special case of exact medians ($\epsilon = 0$), the problem admits a PTAS, at least for a constant-sized alphabet. More specifically, we show the following result. 

\begin{restatable}{theorem}{exactkon}
\label{thm:exact-kon}
Given a set of strings $X$, a parameter $k$, and two parameters $\delta, \eta$, there exists an algorithm such that:
\begin{enumerate}[(I).]
\item If $k \le \frac{1}{\delta}$, the algorithm outputs $k$ exact medians with min dispersion at least $t^*$, in $O(nd\log \min\left(n,|\Gamma|\right)+|\Gamma|^{\frac{1}{\delta}} d^{\frac{1}{2\delta^2}})$ time, and,
\item If $D^* \ge \frac{4}{\delta^2}(2\log k + 1)$ and $k>\frac{1}{\delta}$, then with probability at least $1-\eta$, the algorithm outputs $k$ exact medians with min dispersion at least $(1-2\delta)t^*$, in $O(nd\log \min\left(n,|\Gamma|\right)+(kd|\Gamma|+k^2d)\log \frac{1}{\eta})$ time, and,
\item If $D^* < \frac{4}{\delta^2}(2\log k + 1)$ and $k>\frac{1}{\delta}$, the algorithm outputs $k$ exact median strings with min dispersion at least $\frac{1}{2}t^*$, in $O\left(nd\log \min\left(n,|\Gamma|\right)+|\Gamma|^{\frac{4}{\delta^2}}\cdot k^{2+\frac{8}{\delta^2} \log |\Gamma|}\right)$ time.
\end{enumerate}
Here $D^*$ is the optimal diameter between two exact medians in $X$, and $t^*$ is the optimal min dispersion of $k$ exact medians in $X$.
\end{restatable}

We now provide a high-level proof idea for the above theorem. We first derive a dynamic programming–based algorithm that exactly solves the minimum dispersion problem (\autoref{lem:min-dispersion-DP}), and that establishes Item (I). Next, we consider the case where $D^*$ is large enough, and show that one can obtain a $(1-\delta)$-approximation to the minimum dispersion. In this case, we first find the set $\Gamma_i$ of all majority (most frequent) characters per index $i$, and then generate $k$ candidate medians by drawing a character for an index $i$ uniformly at random from that set $\Gamma_i$. A lower bound on the min dispersion achieved by this randomized process follows from a standard concentration inequality (\autoref{lem:uniform-exact}). The main crux of the argument lies in establishing a near-tight upper bound on the optimum min dispersion objective, which we derive by proving a generalized Plotkin bound (\autoref{lem:plotkin-thm} in~\autoref{sec:plotkin}), and that in turn implies Item (II). Finally, we consider the scenario in which $D^*$ is small, and demonstrate that a solution achieving a $1/2$-approximation to the minimum dispersion can be obtained via a greedy algorithm over a polynomial-sized solution (search) space (see~\autoref{lem:exact-small}), establishing Item (III). We defer all the details to~\autoref{app:min-dispersion-exact}.

For the more general case where $\epsilon > 0$, meaning the objective is to compute $(1+\epsilon)$-approximate medians, we present a bi-criteria approximation algorithm.

\begin{restatable}{theorem}{mindisapproxk}
\label{thm:mindis-approx-k}
Given a set of strings $X$, a parameter $k$, and two parameters $\delta, \eta$, there exists an algorithm such that:
\begin{enumerate}[(I).]
\item If $k \le \frac{1}{\delta}$, then the algorithm outputs $k$ $(1+\epsilon)$-approximate medians with min dispersion at least $t^*$, in $O((1+\epsilon)^{\frac{1}{\delta}}|\Gamma|^\frac{1}{\delta} n^{\frac{1}{\delta}} d^{\frac{2}{\delta^2}})$ time, and,
\item If $D^* \le \frac{4}{\delta^2}$ and $k >\frac{1}{\delta}$, then the algorithm outputs $k$ $(1+\epsilon)$-approximate medians with min dispersion at least $\frac{1}{2}t^*$, in $O \left(k^2 |\Gamma|^{\frac{4}{\delta^2}}d^{\frac{4}{\delta^2}}+nd \cdot |\Gamma|^{\frac{4}{\delta^2}}d^{\frac{4}{\delta^2}}\right)$ time.
\item If $D^* \ge \frac{4}{\delta^2}(2\log k + 1)$ and $k >\frac{1}{\delta}$, then with probability at least $1-\eta$, the algorithm outputs $k$ $(1+2\epsilon)$-approximate medians with min dispersion at least $\frac{1-\delta}{2}t^*$, in $O ((1+\epsilon)nd+d\log d+k^2d\log \frac{1}{\eta})$ time.
\item If $t^* \ge \frac{8+4\delta}{\delta}\sqrt{d}(2\log k + 2)$ and $k >\frac{1}{\delta}$, then with probability at least $1-\eta$, the algorithm outputs $k$ distinct $(1+\epsilon+\delta)$-medians with min dispersion at least $\frac{1-\delta}{2}t^*$, in $O(nd\log \min\left(n,|\Gamma|\right) +(nkd+ k^9 d^3) \log \frac{1}{\eta})$ time.
\end{enumerate}
Here $D^*$ is the optimal diameter between two $(1+\epsilon)$-approximate medians in $X$, and $t^*$ is the optimal min dispersion of $k$ many $(1+\epsilon)$-approximate medians in $X$.
\end{restatable}

The proofs of Item (I) and (II) are similar to the argument used for the corresponding cases in~\autoref{thm:exact-kon} (see~\autoref{lem:min-dispersion-approx-DP},~\autoref{lem:approx-small}). For Item (III), we start with two $(1+\epsilon)$-approximate medians realizing the diameter (obtained from~\autoref{thm:diverse-fin}), and generate $k$ candidate approximate medians via a randomized process by selecting characters randomly from these two initial approximate medians (see~\autoref{lem:uniform-approx}). We argue that all these candidates are also $(1+2\epsilon)$-approximate medians. Then, using a standard concentration bound together with the fact that the min dispersion can at most be the diameter, we derive Item (III). We further improve the bi-criteria approximation in Item (IV) for a large regime by using LP relaxation together with \emph{dependent rounding} framework (\autoref{thm:diverse-fin-k}). We provide all the details in~\autoref{app:min-dispersion-approx}.

%% file: appendix.tex
\section{Missing Proofs from Preliminaries}
\label{sec:folklore-proof}

\begin{proof}[Proof of~\autoref{lem:MCC-opt}]
First, observe the following about the Hamming distance objective of the median string problem. Let $s$ be a string in $\Gamma^d$. Then,
\begin{align*}
\sum_{x \in X}H(x,s) &= \sum_{x \in X}\sum_{i=1}^d \ind{x_i \neq s_i} \\
&= \sum_{i=1}^d \sum_{x \in X} \ind{x_i \neq s_i} \\
&= \sum_{i=1}^d \left( n-\sum_{x \in X} \ind{x_i = s_i}\right)\\
&= \sum_{i=1}^d \left( n-\freq{i}{s}\right)
\end{align*}

Now consider the string $w=\mcc(X)$. For the sake of contradiction, assume $w$ is not an optimal solution. Let $y^*$ be an (arbitrary) optimal solution. Then, since $w$ is not optimal, we get that,
\begin{align*}
&\sum_{x \in X}H(x,w) > \sum_{x \in X}H(x,y^*)\\ \iff & \sum_{i=1}^d \left( n-\freq{i}{w}\right) > \sum_{i=1}^d \left( n-\freq{i}{y^*}\right)\\
\iff & \sum_{i=1}^d \freq{i}{y^*} > \sum_{i=1}^d \freq{i}{w}.\\
\end{align*}
However, by definition, we know that for all $i \in [d]$, $\freq{i}{w} \ge \freq{i}{y^*}$. Therefore, we can see that $\sum_{i=1}^d \freq{i}{y^*} > \sum_{i=1}^d \freq{i}{w}$ is a contradiction. Therefore, the assumption that $\sum_{x \in X}H(x,w) > \sum_{x \in X}H(x,y^*)$ is wrong. Since $\sum_{x \in X}H(x,y^*) = \opt$ and $w$ is a valid solution, we get that $\opt \le \sum_{x \in X}H(x,w) \le \opt$ which implies $\sum_{x \in X}H(x,w) = \opt$. Since $w$ is an optimal solution, we get,
\begin{align*}
0&=\sum_{x \in X}H(x,w)-\sum_{x \in X}H(x,y^*) \\
&=\sum_{i=1}^d \left( n-\freq{i}{w}\right)-\sum_{i=1}^d \left( n-\freq{i}{y^*}\right)\\
&=\sum_{i=1}^d \left( \freq{i}{y^*}-\freq{i}{w}\right).
\end{align*}
Now assume $\exists\; i\in[d]$ such that $\freq{i}{y^*} \neq \freq{i}{w}$. Since $\freq{i}{w} \ge \freq{i}{y^*}$, we get $\freq{i}{w} > \freq{i}{y^*}$. Since $\freq{j}{y^*}-\freq{j}{w} \le 0$ for all $j \in [d]$, we get that $\sum_{j=1}^d \left( \freq{j}{y^*}-\freq{j}{w}\right) = \sum_{j \neq i} \left( \freq{j}{y^*}-\freq{j}{w}\right) + \left( \freq{i}{y^*}-\freq{i}{w}\right) <0$. However, this contradicts with $\sum_{j=1}^d \left( \freq{j}{y^*}-\freq{j}{w}\right) = 0$ which is from our assumption of optimality of $w$. Therefore, the assumption that $\exists\; i\in[d]$ such that $\freq{i}{y^*} \neq \freq{i}{w}$ is incorrect. Therefore, $\forall\; i \in [d]$, $\freq{i}{y^*} = \freq{i}{w}$. 
\end{proof}

\begin{proof}[Proof of~\autoref{lem:opt-offset}]
The proof follows from the direct calculation given below,
\begin{align*}
\sum_{x \in X} H(x,s)&=\sum_{i=1}^d (n-\freq{i}{s}) \\
&=\sum_{i=1}^d (n-\freq{i}{w})+\sum_{i=1}^d (\freq{i}{w}-\freq{i}{s}) \\
&=\opt+\sum_{i:s_i\neq w_i} (\freq{i}{w}-\freq{i}{s})+ \sum_{i:s_i=w_i} (\freq{i}{w}-\freq{i}{s}) \\
&=\opt+\sum_{i:s_i\neq w_i} (\freq{i}{w}-\freq{i}{s}).
\end{align*}
\end{proof}

\section{Exact Algorithm for Diameter Maximization: Median Strings}
\label{app:diverse-fin-exact}
\begin{restatable}{theorem}{diversefinexact}
\label{thm:diverse-fin-exact}
Consider an alphabet $\Gamma$. There exists an algorithm that, given any $X \subseteq \Gamma^d$ of size $n$, outputs two Hamming medians with maximum diameter, and runs in time $O(nd)$.
\end{restatable}

\begin{proof}
For the Hamming metric, two exact medians such that the distance between them is maximized are straightforward to obtain. Let $T$ be the set of all indices $i \in [d]$ where $\exists \; e\neq w_i \in \Gamma$ such that $|\{x \in X : x_i=e\}|=\freq{i}{w}$. Now, consider the following string $\hat{w}$: For each $i \in T$, set $\hat{w}_i = e$, where $|\{x \in X : x_i=e\}|=\freq{i}{w}$, and for each $i \not \in T$, set $\hat{w}_i = w_i$. Since $\freq{i}{w}=\freq{i}{\hat{w}}$ for all $i \in [d]$, $\hat{w}$ is also an optimal median (by~\autoref{lem:MCC-opt}). Now, for the contradiction's sake, suppose there is an optimal median $s$ such that $H(w,s) > H(w,\hat{w})$. By~\autoref{lem:MCC-opt}, for each $i$, $\freq{i}{s} = \freq{i}{w}$, and thus each $j \in \{i:w_i \neq s_i\}$, $j \in T$ (by the construction). Hence, $H(w,s) \le H(w, \hat{w})$, leading to a contradiction. Thus, we conclude that $w, \hat{w}$ are two optimal medians with maximum diameter.

We can see that using an $O(nd)$ memory to store the number of occurrences for each character (assuming constant read and write), we can calculate the $\hat{w}$ in $O(nd)$ time.
\end{proof}

\section{Exact Algorithm for Diameter Maximization: \texorpdfstring{$(1+\epsilon)$}{(1+ϵ)}-Approximate Medians}
\label{app:diam-approx}

\begin{algorithm}
\caption{Diverse $(1+\epsilon)$-Approximate Medians Algorithm}\label{alg:diverse-finite}
\textbf{Input} A set of strings $X \subseteq \Gamma^d$, and an $\epsilon \ge 0$\\
\textbf{Output} Two strings $y,z \in \Gamma^d$\\
\begin{algorithmic}[1]
\STATE Compute $w=\mcc(X)$, and using that compute the value of $\opt$
\STATE Let $S,R=\emptyset$
\STATE Compute $\hat{w}$ as follows: If $\freq{i}{w} < n$, then set $\hat{w}_i = \arg\max_{e \in \Gamma \setminus \{w_i\}} |\{x \in X:x_i=e\}|$ (breaking ties arbitrarily); otherwise, set $\hat{w}_i=e$ for an arbitrary character $e \in \Gamma \setminus \{w_i\}$\label{lne:g-fin}
\STATE Let array $M=\{(\freq{i}{w}-\freq{i}{\hat{w}},i): i\in [d]\}$. Sort $M$ in non-decreasing order of $(\freq{i}{w}-\freq{i}{\hat{w}})$.  \label{lne:M-sort-fin}
\FOR{$i = 1 \text{ to } d$}\label{lne:for-st-fin}
    \STATE Consider the tuple $M_i$ (in the sorted $M$) and consider the index $k$ such that, $M_i=(\freq{k}{w}-\freq{k}{\hat{w}},k)$.
    \IF{$\sum_{j \in S} (\freq{j}{w}-\freq{j}{\hat{w}}) + (\freq{k}{w}-\freq{k}{\hat{w}}) \le \epsilon \opt$} 
        \STATE Add $k$ to $S$\label{lne:S-lne-fin}
    \ELSE
        \STATE Break
    \ENDIF
\ENDFOR
\FOR{$i = |S|+1 \text{ to } d$}
    \STATE Consider the tuple $M_i$ (in the sorted $M$) and consider the index $k$ such that, $M_i=(\freq{k}{w}-\freq{k}{\hat{w}},k)$.
    \IF{$\sum_{j \in R} (\freq{j}{w}-\freq{j}{\hat{w}}) + (\freq{k}{w}-\freq{k}{\hat{w}}) \le \epsilon \opt$} 
        \STATE Add $k$ to $R$\label{lne:R-lne-fin}
    \ELSE
        \STATE Break
    \ENDIF
\ENDFOR \label{lne:for-en-fin}
\STATE Construct $s$ by setting $s_i = \hat{w}_i$ for all $i \in S$, and $s_i=w_i$ for all $i \not \in S$.\label{lne:S-string-fin} 
\STATE Construct $r$ by setting $r_i = \hat{w}_i$ for all $i \in R$, and $r_i=w_i$ for all $i \not \in R$. \label{lne:R-string-fin}
\STATE \label{lne:T-set-fin}Set $T = S \cup R \cup \{k\}$ where $k$ is the index such that $M_{(|S|+|R|+1)}=(\freq{k}{w}-\freq{k}{\hat{w}},k)$ (where $M$ is the sorted array).
\STATE Take two empty strings $y,z$, and set $y \leftarrow s$, $z \leftarrow r$
\IF{$\sum_{i \in T} (\freq{i}{w}-\freq{i}{\hat{w}}) \le 2\epsilon \opt$}
    \STATE Use {\mdiff} algorithm (Algorithm~\ref{alg:min-sum-diff} in~\autoref{sec:smallest-sum-diff}) with inputs $T$ and $\hat{M}=\{\freq{i}{w}-\freq{i}{\hat{w}}: i\in [d]\}$, to find partitions of $T$, $T_1,T_2$ such that the $|\sum_{i \in T_1} (\freq{i}{w}-\freq{i}{\hat{w}})-\sum_{i \in T_2} (\freq{i}{w}-\freq{i}{\hat{w}})|$ is minimized. \label{lne:fin-sum-diff}
    \IF{$\sum_{i \in T_1} (\freq{i}{w}-\freq{i}{\hat{w}}) \le \epsilon \opt$ and $\sum_{i \in T_2} (\freq{i}{w}-\freq{i}{\hat{w}}) \le \epsilon \opt$}\label{lne:if-condition}
        \STATE Set $z$ to be the string where $z_i = \hat{w}_i$ for $i \in T_1$ and $z_i=w_i$ for $i \not \in T_1$ \label{lne:T1-string-fin}
        \STATE Set $y$ to be the string where $y_i = \hat{w}_i$ for $i \in T_2$ and $y_i=w_i$ for $i \not \in T_2$.\label{lne:T2-string-fin}
    \ENDIF
\ENDIF
\STATE Return $y,z$
\end{algorithmic}
\end{algorithm}

\noindent \textbf{Analysis of the algorithm.} Let $y^*,z^*$ be (arbitrary) optimal solutions to $\dmhm$ Problem (Problem~\ref{prob:diameter-max}). In other words, $y^*,z^*$ are two $(1+\epsilon)$-approximate medians with maximum (diversity) $H(y^*,z^*)$.

Let us first argue that there always exists a pair of "structured" strings with the same optimality guarantees; however, the sets of indices on which these two strings deviate from $w$ are disjoint. More specifically,

\begin{lemma}\label{lem:optimal-string}
There exists a pair of $(1+\epsilon)$-approximate medians $\hat{y},\hat{z} \in \Gamma^d$ such that,
\begin{enumerate}[(I).]
    \item For all $ i\in [d]$, either $\hat{y}_i=w_i$ or $\hat{z}_i=w_i$,
    \item $H(\hat{y},\hat{z})=D^*$.
\end{enumerate}
Furthermore, for $\hat{Y} := \{i: \hat{y}_i \neq w_i\}$ and $\hat{Z} := \{i: \hat{z}_i \neq w_i\}$, $\hat{Y}, \hat{Z}$ are disjoint and $|\hat{Y}|+|\hat{Z}| = D^*$.
\end{lemma}
To prove~\autoref{lem:optimal-string} we use the strings $y^*,z^*$ and directly construct two strings $\hat{y},\hat{z}$ from $y^*,z^*$ such that  $\hat{y},\hat{z}$ satisfies the lemma. 

\begin{proof}
We provide a constructive proof of the lemma. Let us construct two strings as follows:
\[
\hat{y}= y^*, 
\]
and $\hat{z}$ by setting for each $i \in [d]$,
\begin{align*}
    \hat{z}_i = \begin{cases}
    z^*_i \text{ if }y^*_i=w_i\\
    w_i \text{ otherwise}
    \end{cases}
\end{align*}
Let $\hat{Y} := \{i: \hat{y}_i \neq w_i\}$ and $\hat{Z} := \{i: \hat{z}_i \neq w_i\}$. It is not hard to observe that by the above construction, $\hat{Y}, \hat{Z}$ are disjoint.

We first argue that $\hat{y},\hat{z}$ are $(1+\epsilon)$-approximate medians. Clearly, $\hat{y}=y^*$ is an $(1+\epsilon)$-approximate median. Now consider $\hat{z}$. By construction, observe that for any $i \in [d]$, $\hat{z}_i \neq w_i$ only if $z^*_i \neq w_i$, and thus 
\begin{equation}
    \label{eq:Z-subset}
    \hat{Z} \subseteq \{i:z^*_i \neq w_i\}.
\end{equation}
Moreover, observe that
\begin{equation}
\label{eq:Z-equal}
    \forall i\in \hat{Z},\; \freq{i}{\hat{z}}=\freq{i}{z^*}.
\end{equation}
    
Now, by~\autoref{lem:opt-offset} and the definition of $\hat{Z}$, we get that
\begin{align*}
\sum_{x \in X} H(x,\hat{z}) &= \opt+\sum_{i \in \hat{Z}} (\freq{i}{w}-\freq{i}{\hat{z}})\\
&= \opt+\sum_{i \in \hat{Z}} (\freq{i}{w}-\freq{i}{z^*})\quad\text{(By Eq.~\ref{eq:Z-equal})}\\
&\le \opt+\sum_{i: z^*_i \neq w_i} (\freq{i}{w}-\freq{i}{z^*})\quad\text{(By Eq.~\ref{eq:Z-subset})}\\
&=\sum_{x \in X} H(x,z^*) \quad\text{(By Lem.~\ref{lem:opt-offset})}\\
&\le (1+\epsilon) \opt.
\end{align*}

Next, we argue that $H(\hat{y},\hat{z}) = H(y^*,z^*)$ (recall, $D^*=H(y^*,z^*)$). In doing so, our first step would be to show that $H(y^*,z^*) = |\hat{Y}|+|\hat{Z}|$.

We now claim the following:
\begin{enumerate}
    \item \label{itm:Y} For any $i \in \hat{Y}$, $y^*_i \ne z^*_i$.
    \item \label{itm:Z} For any $i \in \hat{Z}$, $y^*_i \ne z^*_i$.
    \item \label{itm:Y-Z} For any $i \in [d] \setminus \left(\hat{Y}\cup \hat{Z}\right)$, $y^*_i = z^*_i$.
\end{enumerate}

Let us start by reasoning about the last two items since they are relatively easier to observe. To see Item~\ref{itm:Y-Z}, over the universe $[d]$, $[d] \setminus \left(\hat{Y}\cup \hat{Z}\right) = \overline{\hat{Y}} \cap \overline{\hat{Z}}$, and thus by construction
\begin{equation}
    \label{eq:Y-Z}
    y^*_i = \hat{y}_i = w_i = \hat{z}_i = z^*_i.
\end{equation}

Then, to see Item~\ref{itm:Z}, for any $i \in \hat{Z}$, by construction, we must have that $y^*_i=w_i$, and thus by Equation~\ref{eq:Z-subset}, $y^*_i \ne z^*_i$.

Now, we reason about Item~\ref{itm:Y}. Note that for any $i \in \hat{Y}$, by construction of $\hat{y}$, $y^*_i =\hat{y}_i\neq w_i$. So,
\begin{itemize}
    \item If $z^*_i = w_i$, clearly, $y^*_i \ne z^*_i$.
    \item Otherwise (i.e., $z^*_i \ne w_i$), we claim that if there exists an $i \in \hat{Y}$ such that $y^*_i = z^*_i$, we get another string $z'$ such that $z'$ is also an $(1+\epsilon)$-approximate median and $H(y^*,z') > H(y^*,z^*)$, contradicting the fact that $y^*,z^*$ maximizes $H(y^*,z^*)$ (among all pair of $(1+\epsilon)$-approximate medians). To see this, assuming the existence of an $i \in \hat{Y}$ such that $y^*_i = z^*_i$, let us construct the following string $z'$ as follows: Set $z'_i = w_i$, and for all $j \ne i$, set $z'_j = z^*_j$. Then, clearly, $H(y^*,z') = H(y^*,z^*) + 1$. Further, by~\autoref{lem:opt-offset},
\begin{align*}
\sum_{x \in X} H(x,z') &= \opt+\sum_{z'_j \neq w_j} (\freq{j}{w}-\freq{j}{z'})\\
&= \opt+\sum_{z'_j \neq w_j,j\neq i} (\freq{j}{w}-\freq{j}{z^*})+ (\freq{i}{w}-\freq{i}{z'})\\
&\le \opt+\sum_{z^*_j \neq w_j} (\freq{j}{w}-\freq{j}{z^*}) \;\;\;\;\text{(Since $z'_i=w_i$)}\\
&=\sum_{x \in X} H(x,z^*) \le (1+\epsilon) \opt
\end{align*}
    showing $z'$ is also an $(1+\epsilon)$-approximate median.
\end{itemize}
This establishes Item~\ref{itm:Y}.

Now, by definition,
\begin{align} \label{eq:opt-diverse}
H(y^*,z^*) &= \sum_{i=1}^d \ind{y^*_i \neq z^*_i} \nonumber\\
&= \sum_{i\in\hat{Y}} \ind{y^*_i \neq z^*_i} +\sum_{i\in\hat{Z}} \ind{y^*_i \neq z^*_i} +\sum_{i\in[d]\setminus (\hat{Y}\cup \hat{Z})} \ind{y^*_i \neq z^*_i}\quad \text{(Note, $\hat{Y}\cap \hat{Z}=\emptyset$)}\nonumber\\
&=|\hat{Y}|+|\hat{Z}| \quad\text{(By Items~\ref{itm:Y},~\ref{itm:Z},~\ref{itm:Y-Z})}.
\end{align}

Next, we show that $H(\hat{y},\hat{z}) = |\hat{Y}|+|\hat{Z}|$.
\begin{align}\label{eq:construct-diverse}
H(\hat{y},\hat{z}) &= \sum_{i=1}^d \ind{\hat{y}_i \neq \hat{z}_i}\nonumber\\
&= \sum_{i\in\hat{Y}} \ind{\hat{y}_i \neq \hat{z}_i}+\sum_{i\in\hat{Z}} \ind{\hat{y}_i \neq \hat{z}_i}+\sum_{i\in[d]\setminus (\hat{Y}\cup \hat{Z})} \ind{\hat{y}_i \neq \hat{z}_i}\quad\text{(Note, $\hat{Y}\cap \hat{Z}=\emptyset$)}\nonumber\\
&= \sum_{i\in\hat{Y}} \ind{\hat{y}_i \neq \hat{z}_i}+\sum_{i\in\hat{Z}} \ind{\hat{y}_i \neq \hat{z}_i}\quad\text{(By Eq.~\ref{eq:Y-Z})}\nonumber\\
&=|\hat{Y}|+|\hat{Z}|
\end{align}
where the last equality follows since for any $i \in  \hat{Y}$, by construction, $\hat{z}_i = w_i$, and for any $i\in  \hat{Z}$, again by construction, $\hat{y}_i = y^*_i = w_i$. Hence, by Equation~\ref{eq:opt-diverse} and Equation~\ref{eq:construct-diverse}, we deduce that, 
\[
H(\hat{y},\hat{z}) = H(y^*,z^*)=D^*.
\]
\end{proof}

Next, we consider the strings $s,r$ constructed at lines~\ref{lne:S-string-fin}, \ref{lne:R-string-fin} of Algorithm~\ref{alg:diverse-finite}. We show that they are approximate medians with diversity ``close'' to the maximum. We then use this result to argue that our algorithm indeed returns two approximate medians with the maximum diversity.

\begin{lemma}\label{lem:diverse-fin-diff}
Consider $s,r$ constructed in Algorithm~\ref{alg:diverse-finite}.
\begin{enumerate}[(I).]
\item Both $s,r$ are $(1+\epsilon)$-approximate medians, and
\item $H(s,r) \ge D^* - 1$. 
\end{enumerate}
\end{lemma}
To prove~\autoref{lem:diverse-fin-diff} we use the strings $\hat{y},\hat{z}$ from~\autoref{lem:optimal-string} and the sets $S,R$ lines~\ref{lne:S-string-fin}, \ref{lne:R-string-fin}, and show that by construction, $s,r$ are  $(1+\epsilon)$-approximate medians and due to the maximality of $S$ and $R$ (and since $S \cup R$ is the set of indices corresponding to smallest $\freq{i}{w}-\freq{i}{\hat{w}}$ values), $H(s,r) \ge D^* - 1$. 

\begin{proof}
Let $\hat{y},\hat{z},\hat{Y},\hat{Z}$ be strings and sets referred to in~\autoref{lem:optimal-string}. Further, consider the sets $S, R$ constructed in Algorithm~\ref{alg:diverse-finite}. It follows immediately from the construction of the string $\hat{w}$ and sets $S,R$ that
\begin{align}
    &S=\{i: s_i = \hat{w}_i\} = \{i: s_i \neq w_i\} \label{eq:S-def}\\
    &R=\{i: r_i = \hat{w}_i\} = \{i: r_i \neq w_i\}. \label{eq:R-def}
\end{align}

First, by~\autoref{lem:opt-offset},
\begin{align*}
\sum_{x \in X} H(x,s) &= \opt + \sum_{i:s_i \neq w_i} (\freq{i}{w}-\freq{i}{s})\\
&= \opt + \sum_{i\in S} (\freq{i}{w}-\freq{i}{s})\\
&= \opt + \sum_{i\in S} (\freq{i}{w}-\freq{i}{\hat{w}}) \\
&\le \opt+\epsilon \opt = (1+\epsilon)\opt
\end{align*}
where the inequality follows since $\sum_{i\in S} (\freq{i}{w}-\freq{i}{\hat{w}}) \le \epsilon \opt$ by construction. Similarly, since $\sum_{i\in R} (\freq{i}{w}-\freq{i}{\hat{w}}) \le \epsilon \opt$,
\begin{align*}
\sum_{x \in X} H(x,r) &\le (1+\epsilon)\opt.
\end{align*}
Thus, both $s,r$ are $(1+\epsilon)$-approximate medians.

We now argue that $H(s,r) \ge H(y^*,z^*) - 1$ (recall, $H(y^*,z^*) = D^*$). For the sake of contradiction, assume that $H(s,r) \le H(y^*,z^*) - 2$. Note, by construction, $S \cap R = \emptyset$. Thus,
\begin{align}\label{eq:Ham-sr}
H(s,r) &= \sum_{i=1}^d \ind{s_i \neq r_i}\nonumber\\
&=\sum_{i \in S} \ind{s_i\neq r_i}+\sum_{i \in R} \ind{s_i\neq r_i}+\sum_{i \in [d]\setminus (S\cup R)} \ind{s_i\neq r_i}\nonumber \\
&=\sum_{i \in S} \ind{s_i\neq r_i}+\sum_{i \in R} \ind{s_i\neq r_i}\nonumber\\
&=|S| + |R| \; \; \text{(By Eq.~\ref{eq:S-def},~\ref{eq:R-def}, and $S\cap R=\emptyset$)}
\end{align}
where the second last equality follows since for all $i \in [d]\setminus (S\cup R)$, by construction, $s_i=r_i=w_i$. Now, by Equation~\ref{eq:Ham-sr} and~\autoref{lem:optimal-string}, our assumption ($H(s,r) \le H(y^*,z^*) - 2$) implies
\begin{equation}
    \label{eq:contradict}
    |\hat{Y}|+|\hat{Z}|\ge |S|+|R|+2.
\end{equation}

Recall, $S,R$ are disjoint (thus $|S\cup R|=|S|+|R|$) and by~\autoref{lem:optimal-string}, $\hat{Y},\hat{Z}$ are also disjoint (thus $|\hat{Y} \cup \hat{Z}|=|\hat{Y}|+|\hat{Z}|$). Let $P=\left(\hat{Y}\cup\hat{Z} \right) \setminus \left( S \cup R \right)$, and $Q=\left( S \cup R \right) \setminus \left(\hat{Y}\cup\hat{Z} \right)  $. Now, Equation~\ref{eq:contradict} immediately implies
\begin{equation}
    \label{eq:contradict-modified}
    |P|\ge |Q|+2.
\end{equation}

Let us now consider $(1+\epsilon)$-approximate medians $\hat{y},\hat{z}$ (from~\autoref{lem:optimal-string}). By applying~\autoref{lem:opt-offset}, we derive that
\begin{equation}
    \label{eq:sum-bound}
    \sum_{i:\hat{y}_i\neq w_i} (\freq{i}{w}-\freq{i}{\hat{y}})+\sum_{i:\hat{z}_i\neq w_i} (\freq{i}{w}-\freq{i}{\hat{z}}) \le 2\epsilon \opt.
\end{equation}

Next, we use Equation~\ref{eq:contradict-modified} to reach a contradiction to the above equation, refuting our initial assumption $H(s,r) \le H(y^*,z^*) - 2$.

Note that $S \cup R$ consists of the first $|S|+|R|$ indices in $M$ sorted by the non-decreasing order of $\freq{i}{w}-\freq{i}{\hat{w}}$. Thus
\begin{equation}
    \label{eq:sorted}
    \forall i\in P, \forall j\in S\cup R,; \freq{i}{w}-\freq{i}{\hat{w}} \ge \freq{j}{w}-\freq{j}{\hat{w}}.
\end{equation}
Let $\ell$ be an index in $P$ with the minimum value of $\freq{i}{w}-\freq{i}{\hat{w}}$, i.e.,
\begin{equation}
    \label{eq:l-def}
    \ell:= \arg \min_{i \in P}\freq{i}{w}-\freq{i}{\hat{w}}.
\end{equation}

Now, we get

\begin{align*}
&\sum_{i:\hat{y}_i\neq w_i} (\freq{i}{w}-\freq{i}{\hat{y}})+\sum_{i:\hat{z}_i\neq w_i} (\freq{i}{w}-\freq{i}{\hat{z}})\\
&\ge \sum_{i:\hat{y}_i\neq w_i} (\freq{i}{w}-\freq{i}{\hat{w}})+\sum_{i:\hat{z}_i\neq w_i} (\freq{i}{w}-\freq{i}{\hat{w}})\quad\text{(By definition of $\hat{w}$)}\\
&= \sum_{i\in \hat{Y}} (\freq{i}{w}-\freq{i}{\hat{w}})+\sum_{i\in\hat{Z}} (\freq{i}{w}-\freq{i}{\hat{w}})\\
&= \sum_{i\in P} (\freq{i}{w}-\freq{i}{\hat{w}}) + \sum_{i\in ({\hat{Y}} \cup {\hat{Z}}) \cap (S\cup R)} (\freq{i}{w}-\freq{i}{\hat{w}})\\
&\ge \left(\sum_{i\in Q} (\freq{i}{w}-\freq{i}{\hat{w}})+2(\freq{\ell}{w}-\freq{\ell}{\hat{w}}) \right)\quad\text{(By Eq.~\ref{eq:contradict-modified},~\ref{eq:sorted},~\ref{eq:l-def})}\\
&+\sum_{i\in (\hat{Y} \cup \hat{Z}) \cap (S\cup R)} (\freq{i}{w}-\freq{i}{\hat{w}})\\
&= \sum_{i\in S\cup R} (\freq{i}{w}-\freq{i}{\hat{w}})+2(\freq{\ell}{w}-\freq{\ell}{\hat{w}})\\
&= \left(\sum_{i\in S} (\freq{i}{w}-\freq{i}{\hat{w}})+(\freq{\ell}{w}-\freq{\ell}{\hat{w}}) \right)+\left(\sum_{i\in R} (\freq{i}{w}-\freq{i}{\hat{w}})+(\freq{\ell}{w}-\freq{\ell}{\hat{w}})\right)\\
&>2 \epsilon \opt
\end{align*}
leading to a contradiction to Equation~\ref{eq:sum-bound}. The last inequality follows due to the following reason: Note, $\ell \not \in S \cup R$. By the construction of $S$, since $S$ is a maximal sized subset such that $\sum_{i\in S} (\freq{i}{w}-\freq{i}{\hat{w}}) \le \epsilon \opt$, 
\[
\sum_{i\in S} (\freq{i}{w}-\freq{i}{\hat{w}})+(\freq{\ell}{c}-\freq{\ell}{g})>\epsilon\opt.
\]
Similarly, $\sum_{i\in R} (\freq{i}{w}-\freq{i}{\hat{w}})+(\freq{\ell}{c}-\freq{\ell}{g})>\epsilon\opt$. 

This, in turn, provides a contradiction to our initial assumption that $H(s,r) \le H(y^*,z^*) - 2$. This completes the proof of~\autoref{lem:diverse-fin-diff}.
\end{proof}

~\autoref{lem:diverse-fin-diff} essentially establishes that it suffices to consider only two cases in showing the correctness of Algorithm~\ref{alg:diverse-finite}. More specifically, next, we argue that when $H(s,r)=D^*$, Algorithm~\ref{alg:diverse-finite} indeed returns $s,r$; otherwise, it returns $y,z$ such that $H(y,z)=H(s,r)+1=D^*$.

\begin{lemma}
    \label{lem:correct-case1}
    Consider $s,r$ constructed in Algorithm~\ref{alg:diverse-finite}. If $H(s,r)=D^*$, Algorithm~\ref{alg:diverse-finite} returns $s,r$.
\end{lemma}
\begin{proof}
Assume otherwise. Then there exists two partitions of $T$, $T_1,T_2$ such that, $\sum_{i\in T_1} (\freq{i}{w}-\freq{i}{\hat{w}}) \le \epsilon \opt$ and $\sum_{i\in T_2} (\freq{i}{w}-\freq{i}{\hat{w}}) \le \epsilon \opt$. Let us construct two strings as follows:
for each $i\in [d]$,
\begin{align*}
\hat{s}_i = 
\begin{cases}
\hat{w}_i, & i \in T_1\\
w_i, & \text{otherwise}
\end{cases}
\end{align*}
and
\begin{align*}
\hat{r}_i = 
\begin{cases}
\hat{w}_i, & i \in T_2\\
w_i, & \text{otherwise}
\end{cases}
\end{align*}
It follows from the construction, 
\begin{align}\label{eq:T1-def}
T_1 = \{i: \hat{s}_i = \hat{w}_i\}=\{i: \hat{s}_i \neq w_i\}
\end{align}
\begin{align}\label{eq:T2-def}
T_2 = \{i: \hat{r}_i = \hat{w}_i\}=\{i: \hat{r}_i \neq w_i\}
\end{align}
By~\autoref{lem:opt-offset},
\begin{align*}
\sum_{x \in X} H(x,\hat{s}) &= \opt + \sum_{i:\hat{s}_i \neq w_i} (\freq{i}{w}-\freq{i}{\hat{s}})\\
&= \opt + \sum_{i\in T_1} (\freq{i}{w}-\freq{i}{\hat{s}})\\
&= \opt + \sum_{i\in T_1} (\freq{i}{w}-\freq{i}{\hat{w}}) \\
&\le \opt+\epsilon \opt = (1+\epsilon)\opt
\end{align*}
where the inequality follows since $\sum_{i\in T_1} (\freq{i}{w}-\freq{i}{\hat{w}}) \le \epsilon \opt$ by construction. Similarly, since $\sum_{i\in T_2} (\freq{i}{w}-\freq{i}{\hat{w}}) \le \epsilon \opt$,
\begin{align*}
\sum_{x \in X} H(x,\hat{r}) &\le (1+\epsilon)\opt.
\end{align*}
Thus, both $\hat{s},\hat{r}$ are $(1+\epsilon)$-approximate medians.

We now argue that $H(\hat{s},\hat{r}) > D^*$, which contradicts the maximality of $D^*$. Therefore, we can show that there are no such strings $\hat{s},\hat{r}$. Note, by construction, $T_1 \cap T_2 = \emptyset$. Thus,
\begin{align}
H(\hat{s},\hat{r}) &= \sum_{i=1}^d \ind{\hat{s}_i \neq \hat{r}_i}\nonumber\\
&=\sum_{i \in T_1} \ind{\hat{s}_i \neq \hat{r}_i}+\sum_{i \in T_2} \ind{\hat{s}_i \neq \hat{r}_i}+\sum_{i \in [d]\setminus (T_1\cup T_2)} \ind{\hat{s}_i \neq \hat{r}_i}\nonumber \\
&=\sum_{i \in T_1} \ind{\hat{s}_i \neq \hat{r}_i}+\sum_{i \in T_2} \ind{s_i\neq r_i}\nonumber\\
&=|T_1| + |T_2| \; \; \text{(By Eq.~\ref{eq:T1-def},~\ref{eq:T2-def}, and $T_1\cap T_2=\emptyset$)}
\end{align}
where the second last equality follows since for all $i \in [d]\setminus (T_1\cup T_2)$, by construction, $\hat{s}_i=\hat{r}_i=w_i$. 

By Equation~\ref{eq:Ham-sr}, $H(s,r)=|S|+|R|$. Therefore,
\begin{align*}
H(\hat{s},\hat{r}) &=|T_1| + |T_2| \\
&=|T| \;\;\text{(Since $T_1 \cup T_2 = T$ and $T_1 \cap T_2 = \emptyset$)}\\
&> |S|+|R| = H(s,r) =D^*
\end{align*}
leading to a contradiction, and therefore, we conclude that the Algorithm~\ref{alg:diverse-finite} returns $s,r$.
\end{proof}

The~\autoref{lem:correct-case1} shows that when $H(s,r)=D^*$, Algorithm~\ref{alg:diverse-finite} indeed returns $s,r$. We next argue that if $H(s,r)=D^*-1$ the Algorithm~\ref{alg:diverse-finite} returns $y,z$ such that $H(y,z)=D^*$.
\begin{lemma}
    \label{lem:correct-case2}
    Consider $s,r$ constructed in Algorithm~\ref{alg:diverse-finite}. If $H(s,r)=D^*-1$, Algorithm~\ref{alg:diverse-finite} returns $y,z$ such that
    \begin{enumerate}[(I).]
    \item Both $y,z$ are $(1+\epsilon)$-approximate medians, and
    \item $H(y,z)=D^*$.
    \end{enumerate}
\end{lemma}
To prove~\autoref{lem:correct-case2}, we show that if $H(s,r)=D^*-1$, then the set $T$ derived in line~\ref{lne:T-set-fin} of Algorithm~\ref{alg:diverse-finite} can be partitioned such that they satisfy the conditions of the if statement in line~\ref{lne:if-condition} and that the Algorithm~\ref{alg:diverse-finite} finds such a partitioning and therefore the strings $y,z$ returned by the algorithm are $(1+\epsilon)$-approximate medians and $H(y,z)=D^*$. 

\begin{proof}
To prove~\autoref{lem:correct-case2}, we will first show that if there exists $y^*,z^*$ such that, $H(y^*,z^*)=D^*$, then there is a partitioning of $T$, $T_1,T_2$ such that the $\sum_{i \in T_1} \left(\freq{i}{w}-\freq{i}{\hat{w}}\right) \le \epsilon \opt$ and $\sum_{i \in T_2} \left(\freq{i}{w}-\freq{i}{\hat{w}}\right) \le \epsilon \opt$. We will then show that if there is such a partitioning of $T$, Algorithm~\ref{alg:diverse-finite} finds two strings $y,z$ such that $y,z$ are $(1+\epsilon)$-approximate medians, and $H(y,z)=D^*$. 

Let us consider $(1+\epsilon)$-approximate medians $\hat{y},\hat{z}$ and the sets $\hat{Y},\hat{Z}$ (from~\autoref{lem:optimal-string}).  By applying~\autoref{lem:opt-offset}, we derive that,
\begin{align*}
\sum_{i:\hat{y}_i \neq w_i} \left(\freq{i}{w}-\freq{i}{\hat{y}}\right) \le \epsilon \opt
\end{align*}
and
\begin{align*}
\sum_{i:\hat{z}_i \neq w_i} \left(\freq{i}{w}-\freq{i}{\hat{z}}\right) \le \epsilon \opt
\end{align*}

By~\autoref{lem:optimal-string}, $\hat{Y},\hat{Z}$ are disjoint (thus $|\hat{Y} \cup \hat{Z}|=|\hat{Y}|+|\hat{Z}|$). Let $P_1=\hat{Y} \setminus T$, $P_2=\hat{Z} \setminus T$ and $Q=T \setminus \left(\hat{Y}\cup\hat{Z} \right)$. Since $\hat{Y}, \hat{Z}$ are disjoint, $P_1, P_2$ are also disjoint. Therefore, $|\left(\hat{Y} \cup \hat{Z}\right)\setminus T| = |\hat{Y} \setminus T|+ |\hat{Z} \setminus T|=|P_1| +|P_2|$.

Note that, $D^* = H(s,r)+1 = |S|+|R|+1= |T|$
(where the last equality comes from the fact that $T$ is constructed by adding one element to $S \cup R$). By~\autoref{lem:optimal-string}, $H(\hat{y},\hat{z})=|\hat{Y}|+|\hat{Z}|$ and since $|T|=D^*=H(\hat{y},\hat{z})$, $|T|=|\hat{Y}|+|\hat{Z}|$. Therefore we can see that, $|P_1| +|P_2|=|Q|$.

Note that $T$ consists of the first $|T|$ indices of $M$ sorted in the non-decreasing order of $\freq{i}{w}-\freq{i}{\hat{w}}$. Thus
\begin{equation}
    \label{eq:sorted-T}
    \forall i\in P_1\cup P_2, \forall j\in T,; \freq{i}{w}-\freq{i}{\hat{w}} \ge \freq{j}{w}-\freq{j}{\hat{w}}.
\end{equation}

Let $Q_1, Q_2$ be two paritions of $Q$ such that $|Q_1|=|P_1|$ and $|Q_2|=|P_2|$. Let $\hat{T}_1=Q_1 \cup \left(T \cap \hat{Y}\right)$ and $\hat{T}_2=Q_2 \cup \left(T \cap \hat{Z}\right)$. Note that $Q_1,Q_2 \subseteq T$.

Using Equation~\ref{eq:sorted-T},
\begin{align*}
&\sum_{i \in \hat{T}_1} \left(\freq{i}{w}-\freq{i}{\hat{w}}\right) \\
&= \sum_{i \in Q_1 \cup \left(T \cap \hat{Y}\right)} \left(\freq{i}{w}-\freq{i}{\hat{w}}\right)\\
&=\sum_{i \in Q_1} \left(\freq{i}{w}-\freq{i}{\hat{w}}\right)+\sum_{i \in \left(T \cap \hat{Y}\right)} \left(\freq{i}{w}-\freq{i}{\hat{w}}\right)\\
&\le \sum_{i \in P_1} \left(\freq{i}{w}-\freq{i}{\hat{w}}\right)+\sum_{i \in \left(T \cap \hat{Y}\right)} \left(\freq{i}{w}-\freq{i}{\hat{w}}\right)\;\;\text{(Since $|P_1|=|Q_1|$)}\\
&= \sum_{i \in \hat{Y}} \left(\freq{i}{w}-\freq{i}{\hat{w}}\right)\\
&\le  \sum_{i \in \hat{Y}} \left(\freq{i}{w}-\freq{i}{\hat{y}}\right)\;\;\text{(From the definitions of $\hat{Y}$ and $\hat{w}$)}\\
& \le \epsilon \opt
\end{align*}
Similarly,
\begin{align*}
\sum_{i \in \hat{T}_2} \left(\freq{i}{w}-\freq{i}{\hat{w}}\right) \le \epsilon \opt
\end{align*}

Note that since $\hat{Y},\hat{Z}$ are disjoint and $Q_1,Q_2$ are disjoint, $\hat{T}_1,\hat{T}_2$ are partitions of $T$. 

Next, given the existence of $\hat{T}_1,\hat{T}_2$, we will prove that the Algorithm~\ref{alg:diverse-finite} returns two strings $y,z$ such that $y,z$ are $(1+\epsilon)$-approximate medians and $H(z,y)=D^*$.

Let $T_1, T_2$ be the partitions given by Algorithm~\ref{alg:diverse-finite}. If there exists $\hat{T}_1$ and $\hat{T}_1$ such that $\sum_{i \in \hat{T}_1} (\freq{i}{w}-\freq{i}{\hat{w}}) \le \epsilon \opt$ and $\sum_{i \in \hat{T}_2} (\freq{i}{w}-\freq{i}{\hat{w}}) \le \epsilon \opt$, then we can show that $\sum_{i \in T_1} (\freq{i}{w}-\freq{i}{\hat{w}}) \le \epsilon \opt$ and $\sum_{i \in T_2} (\freq{i}{w}-\freq{i}{\hat{w}}) \le \epsilon \opt$.  

Assume otherwise. Without loss of generality assume for $T_1$, $\sum_{i \in T_1} (\freq{i}{w}-\freq{i}{\hat{w}}) > \epsilon \opt$. Also, without loss of generality, assume $\sum_{i \in \hat{T}_1} (\freq{i}{w}-\freq{i}{\hat{w}}) \ge \sum_{i \in \hat{T}_2} (\freq{i}{w}-\freq{i}{\hat{w}})$. Then,
\begin{align*}
&\sum_{i \in T_1} (\freq{i}{w}-\freq{i}{\hat{w}})-\sum_{i \in T_2} (\freq{i}{w}-\freq{i}{\hat{w}})\\
&= 2\sum_{i \in T_1} (\freq{i}{w}-\freq{i}{\hat{w}})-\left(\sum_{i \in T_1} (\freq{i}{w}-\freq{i}{\hat{w}})+\sum_{i \in T_2} (\freq{i}{w}-\freq{i}{\hat{w}})\right) \\
&= 2\sum_{i \in T_1} (\freq{i}{w}-\freq{i}{\hat{w}})-\sum_{i \in T} (\freq{i}{w}-\freq{i}{\hat{w}})\\
& > 2\sum_{i \in \hat{T}_1} (\freq{i}{w}-\freq{i}{\hat{w}}) -\sum_{i \in T} (\freq{i}{w}-\freq{i}{\hat{w}})\\
& = 2\sum_{i \in \hat{T}_1} (\freq{i}{w}-\freq{i}{\hat{w}}) -\left(\sum_{i \in \hat{T}_1} (\freq{i}{w}-\freq{i}{\hat{w}})+\sum_{i \in \hat{T}_2} (\freq{i}{w}-\freq{i}{\hat{w}})\right)\\
&=\sum_{i \in \hat{T}_1} (\freq{i}{w}-\freq{i}{\hat{w}})-\sum_{i \in \hat{T}_2} (\freq{i}{w}-\freq{i}{\hat{w}})
\end{align*}
However, since $T_1,T_2$ are the partitions that minimize $|\sum_{i \in T_1} (\freq{i}{w}-\freq{i}{\hat{w}})-\sum_{i \in T_2} (\freq{i}{w}-\freq{i}{\hat{w}})|$, this is contradiction and therefore, $\sum_{i \in T_1} (\freq{i}{w}-\freq{i}{\hat{w}}) \le \epsilon \opt$ and $\sum_{i \in T_2} (\freq{i}{w}-\freq{i}{\hat{w}}) \le \epsilon \opt$.

Since there exist partitions of $T$, $T_1,T_2$ such that, $\sum_{i \in T_1} (\freq{i}{w}-\freq{i}{\hat{w}}) \le \epsilon \opt$ and $\sum_{i \in T_2} (\freq{i}{w}-\freq{i}{\hat{w}}) \le \epsilon \opt$, the Algorithm~\ref{alg:diverse-finite} returns string $y,z$ such that for $i \in [d]$, 
\begin{align*}
y_i = 
\begin{cases}
\hat{w}_i, & i\in T_1\\
w_i, & \text{otherwise}
\end{cases}
\end{align*}
and 
\begin{align*}
z_i = 
\begin{cases}
\hat{w}_i, & i\in T_2\\
w_i, & \text{otherwise}
\end{cases}
\end{align*}
Therefore, 
\begin{align}\label{eq:T1-def-2}
T_1 = \{i: y_i = \hat{w}_i\}=\{i: y_i \neq w_i\}
\end{align}
\begin{align}\label{eq:T2-def-2}
T_2 = \{i: z_i = \hat{w}_i\}=\{i: z_i \neq w_i\}
\end{align}
By~\autoref{lem:opt-offset},
\begin{align*}
\sum_{x \in X} H(x,y) &= \opt + \sum_{i:y_i \neq w_i} (\freq{i}{w}-\freq{i}{y})\\
&= \opt + \sum_{i\in T_1} (\freq{i}{w}-\freq{i}{y})\\
&= \opt + \sum_{i\in T_1} (\freq{i}{w}-\freq{i}{\hat{w}}) \\
&\le \opt+\epsilon \opt = (1+\epsilon)\opt
\end{align*}
Similarly, $\sum_{x \in X} H(x,z) \le (1+\epsilon) \opt$.

We now argue that $H(z,y) = D^*$. Note, by construction, $T_1 \cap T_2 = \emptyset$. Thus,
\begin{align}
H(y,z) &= \sum_{i=1}^d \ind{y_i \neq z_i}\nonumber\\
&=\sum_{i \in T_1} \ind{y_i \neq z_i}+\sum_{i \in T_2} \ind{y_i \neq z_i} +\sum_{i \in [d]\setminus (T_1\cup T_2)} \ind{y_i \neq z_i}\nonumber \\
&=\sum_{i \in T_1} \ind{y_i \neq z_i}+\sum_{i \in T_2} \ind{y_i \neq z_i}\nonumber\\
&=|T_1| + |T_2| \; \; \text{(By Eq.~\ref{eq:T1-def-2},~\ref{eq:T2-def-2}, and $T_1\cap T_2=\emptyset$)}\\
&=|T|=D^*
\end{align}
Therefore, $y,z$ are $(1+\epsilon)$-approximate medians and $H(y,z)=D^*$ as desired.
\end{proof}

\begin{proof}[Proof of~\autoref{thm:diverse-fin}]

From~\autoref{lem:correct-case1} and~\autoref{lem:correct-case2}, we get that Algorithm~\ref{alg:diverse-finite} outputs strings $y,z$ such that $y,z$ are $(1+\epsilon)$-approximate medians and $H(y,z)=D^*$. Now, all we have to do is argue that the algorithm takes polynomial time. 

To analyze the time complexity of Algorithm~\ref{alg:diverse-finite}, we can consider the main steps in the algorithm. Starting with line~\ref{lne:g-fin}, where we find $\hat{w}$, we get that this involves checking all indices $i \in [d]$ and all characters in the $i$th index of the strings in $X$, and counting the number of occurrences. We can see that using an $O(nd)$ memory to store the number of occurrences for each character (assuming constant read and write), we can calculate the $\hat{w}$ in $O(nd)$ time. Next, in the line~\ref{lne:M-sort-fin}, we also calculate calculates $\freq{i}{w}-\freq{i}{\hat{w}}$ for all $i$ and this also takes $O(nd)$ time and then sorting $M$ takes $O(d\log d)$ time. The lines from line~\ref{lne:for-st-fin} to line~\ref{lne:for-en-fin} are two simple for loops and, therefore, only take $O(d)$ time. And the lines~\ref{lne:S-string-fin} and lines~\ref{lne:R-string-fin} also take only $O(d)$ time since they simply construct strings of size $d$. Now, all we need to argue about is the time complexity of finding $\mdiff$ (Algorithm~\ref{alg:min-sum-diff}). From~\autoref{lem:subset-closest} we get that this takes $O(d \cdot \text{target})$ where the target is $\le \left \lfloor \frac{\lfloor 2\epsilon \opt \rfloor}{2}\right \rfloor \le \left \lfloor \frac{\lfloor 2\epsilon n \rfloor}{2}\right \rfloor$ (since $\sum_{i \in T} \freq{i}{w}-\freq{i}{\hat{w}} \le 2\epsilon \opt$ and since $\opt \le n$). Therefore, we can see that line~\ref{lne:fin-sum-diff} takes $O(\epsilon n d)$. Also, the construction of the strings $y,z$ takes  $O(d)$ time. Therefore, the overall time complexity is $O \left((1+\epsilon)nd+d\log d\right)$. 
\end{proof}

\subsection{Finding a Partition with Smallest Sum Difference}\label{sec:smallest-sum-diff}

In this section, we will introduce the Algorithm~\ref{alg:min-sum-diff}, which is used in the Algorithm~\ref{alg:diverse-finite} to derive the diverse medians. In Algorithm~\ref{alg:min-sum-diff}, given an index set $T$ and a set of elements $M$, we use dynamic programming and backtracking to find a partitioning of a given index set $T$ such that the sum difference, i.e., for partitions $W,U$ of $T$, $|\sum_{i \in W} M_i-\sum_{i \in U} M_i|$, is minimized. We also keep track of how each sum in our dynamic programming table is achieved and use backtracking to find the correct partition. 

\begin{algorithm}
\caption{Minimum sum difference paritioning}\label{alg:min-sum-diff}
{\bf Input:} An array $T$ of size $t$ ($t \le d$) and set $M$ which are values corresponding to $i \in T$.\\
\begin{algorithmic}[1]
\STATE Set $\text{total} = \sum_{i\in R} M_i$ and $\text{target}= \left\lfloor \frac{\text{total}}{2}\right\rfloor$
\STATE Set $DP_{\text{tab}}$ be a table of size $(t+1) \times (\text{target}+1)$ and let all $DP_{\text{tab}}[i][j]=\text{False}$\\\label{lne:dp-st}
\STATE Set $P_{\text{tab}}$ be a table of size $(t+1) \times (\text{target}+1)$ and let all $P_{\text{tab}}[i][j]=-1$\\
\STATE Set $DP_{\text{tab}}[i][0]=\text{True}$ for all $i \in [t+1]$
\FOR{$i$ from $1$ to $n$}
    \FOR{$j$ from $1$ to $\text{target}$}
        \IF{$j \ge M_{T_i}$ and $DP_{\text{tab}}[i-1][j-M_{T_i}]=\text{True}$}
            \STATE Set $DP_{\text{tab}}[i][j]=\text{True}$ and $P_{\text{tab}}[i][j]=j-M_{T_i}$
        \ELSIF{$DP_{\text{tab}}[i-1][j]=\text{True}$}
            \STATE Set $DP_{\text{tab}}[i][j]=\text{True}$ and $P_{\text{tab}}[i][j]=j$
        \ELSE
            \STATE Set $DP_{\text{tab}}[i][j]=\text{False}$
        \ENDIF
    \ENDFOR\label{lne:dp-en}
\ENDFOR
\FOR{$j$ from $\text{target}$ to $0$}\label{lne:find-target}
    \IF{$DP_{\text{tab}}[n][j]=\text{True}$}
         \STATE Set $\text{sumTarget}=j$ and break the loop.
    \ENDIF
\ENDFOR
\STATE Set $W=[\cdot]$ and $i=n$ and $j=\text{sumTarget}$
\WHILE{$j \neq 0$}
    \IF{$P_{\text{tab}}[i][j]\neq j$}
        \STATE Append $T_i$ to $W$
    \ENDIF
    \STATE Set $j = P_{\text{tab}}[i][j]$ and $i=i-1$
\ENDWHILE
\STATE Return $W,R\setminus W$
\end{algorithmic}
\end{algorithm}

Consider the Algorithm~\ref{alg:min-sum-diff}. We will show that the algorithm runs in $O(t \cdot \text{target})$ time (where $\text{target}$ is $\sum_{i\in R} M_i/2$) and finds a paritioning $(W,U)$ of $T$ such that $|\sum_{i\in W} M_i-\sum_{i\in U} M_i|$ is minimized.

\begin{lemma}\label{lem:subset-closest}
Given an array of indices $T$ and the array $M$, where $|T|=t$, the Algorithm~\ref{alg:min-sum-diff} finds two paritions $(W,U)$ of $R$ such that $|\sum_{i\in W} M_i-\sum_{i\in U} M_i|$ is minimized. The algorithm finds this in $O(t \cdot \text{target})$ and it uses $O(t \cdot \text{target})$ memory.
\end{lemma}

We can prove~\autoref{lem:subset-closest} as follows,
\begin{proof}
We will first establish the correctness of our algorithm. Note that the goal of the algorithm is to find a subset $W$ of $R$ such that $\sum_{i \in W} M_i$ is close to the $\text{target}$. We can first consider the recursive function of the dynamic programming table, $DP_{\text{tab}}[i][j]$. Assume we are at the index $i$ of $T_i$ and our goal is to find a subset of $\{R_k: k\le i\}$ such that the subset sums to $j$. Depending on the value of $M_{T_i}$ and the current target $j$, we have a couple of choices. If $j \ge M_{T_i}$, we can check if $T_i$ can be used to make $j$. For this to be possible we need $DP_{\text{tab}}[i-1][j-M_{T_i}] = \text{True}$, i.e., there is a way to make $j-M_{T_i}$ using a subset from $\{R_k: k\le i-1\}$. If this is the case, we can see that the previous target used to reach $j$ is $j-M_{T_i}$. We set $P_{\text{tab}}[i][j]=j-M_{T_i}$ which indicates we need to make $j-M_{T_i}$ using a subset of $\{R_k: k\le i-1\}$. If $j < M_{T_i}$ or $DP_{\text{tab}}[i-1][j-M_{T_i}] = \text{False}$, only way we can sum to $j$ using a subset of elements from $\{R_k: k\le i\}$ is to sum to $j$ using a subset of elements from $\{R_k: k\le i-1\}$. In this case, the previous target used to reach $j$ is $j$. We set $P_{\text{tab}}[i][j]=j$ which indicates we need to sum to $j$ using a subset of $\{R_k: k\le i-1\}$. If neither case is true, then we can see that there is no way to sum to $j$ using a subset of $\{R_k: k\le i\}$.

We can see that the dynmaic programing function, $\text{DP}(i,j)$ is, $\text{DP}(i,j) = \text{True}$ if $j-M_{T_i} \ge 0$ and $\text{DP}(i-1,j-M_{T_i})=\text{True}$, or $\text{DP}(i,j)=\text{True}$ and $\text{False}$ otherwise. We can see that lines~\ref{lne:dp-st} and \ref{lne:dp-en} capture this dynamic programming formulation. Since the algorithm only does lookups with the nested for loops, we only take $O(t\cdot\text{target})$ time to tabulate the dynamic programming table $DP_{\text{tab}}$. Using the for loop starting from line~\ref{lne:find-target}, we can get the largest value $\le \text{target}$ ($\text{sumTarget}$) such that there is a partition of $R$ that sums to $\text{sumTarget}$. Since there is a partition that sums to $\text{sumTarget}$, we can consider the while loop and argue that the while loop constructs the desired partition. In the while loop, we start with $i=n$ and $j=\text{sumTarget}$ and if $P_{\text{tab}}[i][j]=j$ we move onto $P_{\text{tab}}[i-1][j]$. If $P_{\text{tab}}[i][j]\neq j$ we know that $P_{\text{tab}}[i][j]= j-M_{T_i}$ which means we have used $T_i$ in the sum and thereofre we add it to $W$ and move onto $P_{\text{tab}}[i-1][j-M_{T_i}]$. 

Note that since $\text{sumTarget}$ is achievable, the while loop will terminate with $j=0$ and the set $W$ is a set such that $\sum_{i \in W} M_i = \text{sumTarget}$. We get that for $U= R\setminus W$,$\sum_{i \in U} M_i = \text{total}-\text{sumTarget}$. We can also see that $|\sum_{i \in W} M_i-\sum_{i \in U} M_i| =|\text{total}-2\text{sumTarget}|$ and since the $\text{sumTarget}$ is the closest to the $\text{target} = \lfloor \frac{\text{total}}{2} \rfloor$, we get that $W,U$ are the two partitions with minimum sum difference.

Note that the $DP$ in the Algorithm~\ref{alg:min-sum-diff} takes $O(t\cdot \text{target})$, the while loop takes $O(\text{sumTarget})$ time and the for loop staring from line~\ref{lne:find-target} takes $O(\text{target})$ time. So the overall time complexity is $O(t\cdot \text{target})$. We can also see that since $DP_{\text{tab}},P_{\text{tab}}$ are of size $O(t\cdot \text{target})$ memory complexity is also $O(t\cdot \text{target})$.
\end{proof}

\section{Exact Algorithm for Sum Dispersion: \texorpdfstring{$k$}{k} Hamming medians}
\label{app:sum-exact}

\sumdispersionkexact*

\begin{proof}

Let $T$ be the set of all indices $i \in [d]$ where $\exists \; e\neq w_i \in \Gamma$ such that $|\{x \in X : x_i=e\}|=\freq{i}{w}$. For each index $i \in T$, let $\Gamma_i = \{e \in \Gamma \text{ such that } |\{x \in X : x_i=e\}|=\freq{i}{w}\}$. We can construct $k$ strings as follows. 

\begin{algorithm}[!htbp]
\caption{Sum Dispersion Exact Median ($\sdelg$)}
\label{alg:sum-exact}
\begin{algorithmic}[1]
\REQUIRE Strings $X$ and parameter $k$
\STATE Set $Y$ to be $k$ copies of $w$. Let $Y = \{y^{(1)},y^{(2)},\dots,y^{(k)}\}$
\STATE Let $T$ be the set of all indices $i \in [d]$ where $\exists \; e\neq w_i \in \Gamma$ such that $|\{x \in X : x_i=e\}|=\freq{i}{w}$. 
\STATE For each index $i \in T$, let $\Gamma_i = \{e \in \Gamma \text{ such that } |\{x \in X : x_i=e\}|=\freq{i}{w}\}$. Let $\Gamma_i = \{a_i,a_2,\dots,a_{|\Gamma_i|}\}$
\FOR{$i \in T$}
\FOR{$j$ from $1$ to $|\Gamma_i|$}
\FOR{$\ell$ from $\left \lfloor \frac{k}{|\Gamma_i|} \right \rfloor\cdot (j-1)+1$ to $\left \lfloor \frac{k}{|\Gamma_i|} \right \rfloor \cdot j$} 
\STATE Set $y^{(\ell)}_i = a_j$
\ENDFOR
\ENDFOR
\FOR{$\ell$ from $\left \lfloor \frac{k}{|\Gamma_i|} \right \rfloor\cdot |\Gamma_i|+1$ to $k$}
\STATE Let $j = \ell-\left \lfloor \frac{k}{|\Gamma_i|} \right \rfloor\cdot |\Gamma_i|$ and set $y^{(\ell)}_i = a_{j}$
\ENDFOR
\ENDFOR
\RETURN $Y$
\end{algorithmic}
\end{algorithm}

We will show that the $\sdelg$ Algorithm (Algorithm~\ref{alg:sum-exact}) ends up giving $k$ strings that maximize sum dispersion. Before formally establishing the optimality of the procedure, we first observe the following:
\begin{remark}
Let $Y$ be a set of $k$ strings and let $\ell_{a,i} = \sum_{y \in Y} \ind{y_i = a}$ for all $i \in [d]$ and $a \in \Gamma$. Note that,
\begin{align*}
\textsc{sumDp}(Y) &= \frac{1}{2}\sum_{\hat{y},\tilde{y}\in Y} H(\hat{y},\tilde{y})\\
&= \frac{1}{2}\sum_{\hat{y},\tilde{y}\in Y} \sum_{i=1}^d \ind{{\hat{y}}_i \neq {\tilde{y}}_i}\\
&=\sum_{i=1}^d \frac{1}{2}\sum_{\hat{y},\tilde{y}\in Y} \ind{{\hat{y}}_i \neq {\tilde{y}}_i}\\
\end{align*}
Note that for any index $i$ 
\begin{align*}
\sum_{\hat{y},\tilde{y}\in Y} \ind{{\hat{y}}_i \neq {\tilde{y}}_i} &=\sum_{\hat{y},\tilde{y}\in Y} \sum_{a \in \Gamma} \ind{\hat{y}_i =a \text{ and } \tilde{y}_i \neq a} \\
&= \sum_{a \in \Gamma} \ell_{a,i} (k-\ell_{a,i})
\end{align*}
and therefore,
\begin{align*}
\textsc{sumDp}(Y) &= \sum_{i=1}^d \left(\frac{1}{2}\sum_{a \in \Gamma} \ell_{a,i} (k-\ell_{a,i})\right)
\end{align*}
\end{remark}
Note that for all $i \not \in T$, $\ell_{w_i,i} = k$ and for all $i \in T$ and $a \not \in \Gamma_i$, $\ell_{a,i}=0$ (since otherwise it would not be a median string). Therefore, we get,
\begin{align*}
\textsc{sumDp}(Y) &= \sum_{i \in T}\left(\frac{1}{2}\sum_{a \in \Gamma_i} \ell_{a,i} (k-\ell_{a,i})\right)\\
&= \sum_{i \in T}\frac{1}{2}\left(k^2-\sum_{a \in \Gamma_i} \ell^2_{a,i}\right)
\end{align*}
Note that therefore, maximizing the sum dispersion is the same as minimizing the $\sum_{i \in T}\sum_{a \in \Gamma_i} \ell^2_{a,i}$ which imvolves minimizing the $\sum_{a \in \Gamma_i} \ell^2_{a,i}$ (because $\ell_{a,i}$ are independent across $i$). Now consider any $i \in T$ and the optimal $\ell^*_{a,i}$ values that minimize the $\sum_{a \in \Gamma_i} \ell^2_{a,i}$. Let $\ell_{\min} = \min_{a \in \Gamma_i} \ell^*_{a,i}$ and $\ell_{\max} = \max_{a \in \Gamma_i} \ell^*_{a,i}$ (and $a_{\min} = \arg \min_{a \in \Gamma_i} \ell^*_{a,i}$ and $a_{\max} = \arg \max_{a \in \Gamma_i} \ell^*_{a,i}$). 

We first claim that $\ell_{\max} - \ell_{\min} \le 1$. Assume otherwise, i.e., $\ell_{\max} - \ell_{\min} \ge 2$. Now consider the solution where $\ell_{a_{\max},i} = \ell^*_{\max}-1$ and $\ell_{a_{\min},i} = \ell^*_{\min}+1$ (and all the rest are the same). Since the sum is invariant, this is also a feasible solution. However, note that,
\begin{align*}
\ell_{a_{\max},i}^2+\ell_{a_{\min},i}^2 &= \left(\ell_{\max}-1\right)^2+\left(\ell_{\min}+1\right)^2\\
&= {\ell^*_{\max}}^2+{\ell_{\min}}^2-2(\ell_{\max}-\ell_{\min}-1)\\ &<{\ell_{\max}}^2+{\ell_{\min}}^2
\end{align*} 
However, this is a contradiction since $\ell_{\min}, \ell_{\max}$ is part of the optimal solution. Therefore, $\ell_{\max}-\ell_{\min} \le 1$. 

Note that $\ell_{\max} = \ell_{\min} = \hat{\ell}$ is possible if and only if $|\Gamma_i|$ divides $k$ since this implies $\hat{\ell} |\Gamma_i| = k$ for an integer $\hat{\ell}$ and this is only possible when $|\Gamma_i|$ divides $k$. In this case, we have,
\[\sum_{a \in \Gamma_i} \ell^2_{a,i} = |\Gamma_i| \left(\frac{k}{|\Gamma_i|}\right)^2 = p_i \left(\left \lfloor \frac{k}{|\Gamma_i|} \right \rfloor \right)^2 + q_i \left(\left \lfloor \frac{k}{|\Gamma_i|} \right \rfloor+1\right)^2\] 
where $q_i = k (\text{mod } |\Gamma_i|) = 0$ and $p_i = |\Gamma_i|-q_i = |\Gamma_i|$. 

Now consider the case where $\ell_{\max} - \ell_{\min} = 1$. Assume we have $\hat{p}_i$ characters with $\ell_{a,i} = \ell_{\min}$ and $\hat{q}_i$ characters with $\ell_{a,i} = \ell_{\max} = \ell_{\min}+1$. We know that $\hat{p}_i + \hat{q}_i = |\Gamma_i|$ and $\hat{p}_i \cdot \ell_{\min} + \hat{q}_i \cdot (\ell_{\min}+1) = k$. Therefore, we get $(\hat{p}_i+\hat{q}_i) \cdot \ell_{\min} + \hat{q}_i = |\Gamma_i| \cdot \ell_{\min} + \hat{q}_i = k$. Since $\hat{q}_i < |\Gamma_i|$, we get that $\hat{q}_i = k (\text{mod } |\Gamma_i|)$ and therefore, $\hat{p}_i = |\Gamma_i|-\hat{q}_i$. Given $\hat{q}_i, \hat{p}_i$, we also get $\ell_{\min} = \left \lfloor \frac{k}{|\Gamma_i|} \right \rfloor$. Therefore, we get, 
\[\sum_{a \in \Gamma_i} \ell^2_{a,i} = p_i \left(\left \lfloor \frac{k}{|\Gamma_i|} \right \rfloor \right)^2 + q_i \left(\left \lfloor \frac{k}{|\Gamma_i|} \right \rfloor+1\right)^2\] 
where $q_i =  k (\text{mod } |\Gamma_i|)$ and $p_i = |\Gamma_i|-q_i$. 

\begin{figure}
    \centering
    \includegraphics[width=\linewidth]{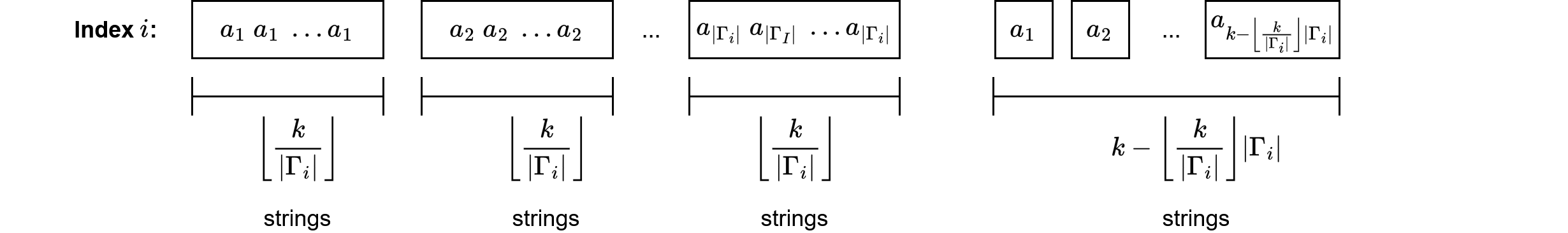}
    \caption{Let $\Gamma_i =\{a_1,\cdots,a_r\}$ be set of most frequent characters at the index $i$. Overview of the characters at index $i$ after the modifications by $\sdelg$ Algorithm (Algorithm~\ref{alg:sum-exact})}
    \label{fig:sum-exact-index}
\end{figure}

Now consider the $\sdelg$ Algorithm (Algorithm~\ref{alg:sum-exact}). Consider any index $i \in T$. Let $q_i =  k (\text{mod } |\Gamma_i|)$ and $p_i = |\Gamma_i|-q_i$. We can see that each character in $\Gamma_i$ occurs $\left \lfloor \frac{k}{|\Gamma_i|}\right \rfloor$ times within the first $\left \lfloor \frac{k}{|\Gamma_i|}\right \rfloor |\Gamma_i|$ strings of $Y$ and in the next $k-\left \lfloor \frac{k}{|\Gamma_i|}\right \rfloor |\Gamma_i|=q_i$ strings, in each string we have a unique character from $\Gamma_i$ (refer to the Figure~\ref{fig:sum-exact-index}). So we end up with $q_i$ character each appearing in $\left \lfloor \frac{k}{|\Gamma_i|}\right \rfloor |\Gamma_i|+1$ strings and $|\Gamma_i|-q_i=p_i$ character each appearing in $\left \lfloor \frac{k}{|\Gamma_i|}\right \rfloor |\Gamma_i|$ strings. Therefore, we get the desired outcome, and from our proof, this is optimal. Therefore, Algorithm~\ref{alg:sum-exact} gives the exact solution. 

Note that using $O(nd)$ memory we can calculate all $\Gamma_i$ in $O\left(nd\log \min\left(n,|\Gamma|\right)\right)$ time in the worst case. Constructing the $k$ strings outlined in the Algorithm~\ref{alg:sum-exact} takes $O(kd)$ time. Therefore, the overall time complexity is $O\left(nd\log \min\left(n,|\Gamma|\right)+kd\right)$. 
\end{proof}

\begin{remark}
Note that while the current construction does not generate distinct strings it can be easily modified to achieve a set of distinct strings by considering an ordering of the indices and then from the starting index, at each subsequent index, for any block of strings with the same character at the previous index, assigning blocks of appropirately sized characters at the current index (a form of branching) so that we cover all characters. This would ensure that as long as the size of $T$ is sufficiently large (i.e., $T \in \Omega(\log k)$, we would end up getting distinct strings.
\end{remark}

\section{A PTAS for Sum Dispersion: \texorpdfstring{$k$}{k} Approximate Hamming medians}
\label{app:sum-approx}

Given $X \subseteq \Gamma^d$, $v^*$ be the maximum possible sum dispersion for any set of $k$ $(1+\epsilon)$-approximate medians, and let $D^*$ be the diameter for the $(1+\epsilon)$-approximate medians. Then,

\sumdispersionk*

Before proving~\autoref{thm:sum-dispersion-k}, we first establish the following result. We then show that if $D^*$ is sufficiently large, this leads to a PTAS, while if $D^*$ is small, a PTAS can be obtained using the earlier work of~\cite{cevallos2015maxsum}.

\sumdislargeD*

In order to formally prove this, we will first introduce some basic notations that will be used throughout this section. Given $X \subseteq \Gamma^d$, an index $i \in [d]$ and $a \in \Gamma \setminus \{w_i\}$, we define the index cost of any string $y$ with $y_i=a$ at the index $i$ to be,  
\begin{align*}
\costind(a,i) = f_i^w-f_i^y.
\end{align*}

Now, we describe the $\sdalg$ algorithm (Algorithm~\ref{alg:density-flip-k}).

\paragraph*{Description of the algorithm.}
Initiate a set of strings $Y$ such that $y_i = w$ for all $y \in Y$ and $|Y|=k$. For each column $i$ (based on the cost) and character $a \in \Gamma \setminus \{w_i\}$, the algorithm modifies a set of strings by changing the character at index $i$ from $w_i$ to $a$. The algorithm employs a notion of "value" from each modification to decide how many strings to modify for each index $i$ and character $a \in \Gamma \setminus \{w_i\}$. 

For any index $i$ and a character $a \in \Gamma \setminus \{w_i\}$, assume we have $r-1$ strings in $Y$ with character $a$ and $\ell_{(i,a,r)}$ strings in $Y$ with character $w_i$. Now consider the case where we change one more string with character $w_i$ to $a$ (to get $r$ strings with character $a$ at index $i$ and $\ell_{(i, a,r)}-1$ strings with $w_i$). Let $Y_{a,r-1}$ be the set of candidate strings before modifying the $r$-th string, and let $Y_{a,r}$ be the set of strings after the modification. Then, we can see that,
\begin{align*}
\ell_{(i,a,r)}-r&=\frac{1}{2} \left(\sum_{b \in \Gamma \setminus \{w_i,a\}} \ell_{b,i}(k-\ell_{b,i}) (\ell_{(i,a,r)}-1)(k-(\ell_{(i,a,r)}-1)+r(k-r)\right)\\
&-\frac{1}{2} \left(\sum_{b \in \Gamma \setminus \{w_i,a\}} \ell_{b,i}(k-\ell_{b,i})+\ell_{(i,a,r)}(k-\ell_{(i,a,r)})+(r-1)(k-(r-1))\right)\\
&= \textsc{sumDp}(Y_{a,r})-\textsc{sumDp}(Y_{a,r-1})
\end{align*}
We define the \emph{density} of modifying a string at index $i$ from $w_i$ to $a$ given that there are $\ell_{(i,a,r)}$ strings with character $w_i$ and $r-1$ strings with character $a$ at index $i$,
    \[
    \rho_{i, a, r, \ell_{(i,a,r)}} = \frac{\ell_{(i,a,r)}-r}{\costind(a,i)}
    \]
Let $L = \{(i,a,r,\ell_{(i,a,r)})|\ell_{(i,a,r)} \le k, i \in [d], a \in \Gamma \setminus \{w_i\}, r \le k\}$. The algorithm first sorts $L$ based on the non-increasing order of density, $\rho_{i,a,r,\ell_{i,a,r}}$, and preprocesses sorted $L$ to remove conflicting entries. Let $L_j$ denote the first $j$ elements of $L$ (in the sorted order). We call $L_j$ a prefix of $L$. For any prefix $L_j$, the prefix has a \texttt{feasible} solution if $\cgrd$ algorithm returns a solution $Y$ such that for each $(i,a,r,\ell_{(i,a,r)}) \in L_j$, $r \le |\{y \in Y|y_i = a\}|$. The $\sdalg$ algorithm finds the longest prefix $L_{\text{thr}}$ for which the $\cgrd$ algorithm returns a \texttt{feasible} solution and returns this solution.
\begin{algorithm}[!htbp]
\caption{Sum Dispersion Approximate Median ($\sdalg$)}
\label{alg:density-flip-k}
\begin{algorithmic}[1]
\REQUIRE Input strings $X$ and parameters $k$, $\epsilon$.
\STATE Calculate $w$ and $\opt$.
\STATE Let $L \leftarrow \{(i,a,r,\ell_{(i,a,r)})|\ell_{(i,a,r)} \le k, i \in [d], a \in \Gamma \setminus \{w_i\}, r \le k\}$
\STATE Sort $L$ in non-increasing order of density $\rho_{i,a,r,\ell_{(i,a,r)}}$ for all $(i,a,r,\ell_{(i,a,r)}) \in L$
\STATE Preprocess $L$: for any index $i$ and $\ell \in [k]$, if $\exists (i,a,r,\ell), (i,b,r',\ell) \in L$ such that $\rho_{i,a,r,\ell}\ge \rho_{i,b,r',\ell}$ then remove $\rho_{i,b,r',\ell}$ from $L$.
\STATE Let $L_{\text{thr}}$ be the longest prefix of $L$ for which the $\cgrd$ Algorithm (Algorithm~\ref{alg:cost-greedy-k}) returns a \texttt{feasible} solution (on the parameters $k,\epsilon,\opt,w$ and the prefix). 
\STATE Let the strings $Y$ be the set of strings returned by using $L_{\text{thr}}$ \label{lne:rhoStar-k}
\RETURN $Y$
\end{algorithmic}
\end{algorithm}

\begin{algorithm}[!htbp]
\caption{Cost Greedy Assignment ($\cgrd$)}
\label{alg:cost-greedy-k}
\begin{algorithmic}[1]
    \REQUIRE Parameters $k$, $\epsilon$, $\opt$, string $w$ and prefix $L_j$.
    \STATE For each index $i$ and $a \in \Gamma \setminus \{w_i\}$, let $h_{a,i} = |\{r|(i,a,r, \ell_{(i,a,r)}) \in L_j\}|$. 
    \STATE Sort each index $i$ and character $a \in \Gamma \setminus \{w_i\}$, i.e. $(a,i)$ in the non-decreasing order of $\costind(a, i)$.
    \STATE Initialize a set of $k$ strings $Y$ such that $y = w$ for all $y \in Y$.
    \STATE Set \texttt{feasible} $\gets$ \texttt{true}
    \FOR{each $(a,i)$ in the sorted order} \label{lne:assign-loop-k}
        \STATE Let $\mathcal{R}_{a,i}$ be the set of strings where for all $y \in \mathcal{R}_{a,i}$, $\sum_{x \in X} H(x,y) + \costind(a,i) \le (1+\epsilon) \opt$
        \IF{$|\mathcal{R}_{a,i}| \ge h_{a,i}$}
            \STATE Pick the $h_{a,i}$ strings in $\mathcal{R}_{a,i}$ with lowest cost and set the character at index $i$ to $a$ for all $h_{a,i}$ strings.
        \ELSE
            \STATE For all strings in $\mathcal{R}_{a,i}$, set the character at index $i$ to $a$.
            \STATE \texttt{feasible} $\gets$ \texttt{false}
        \ENDIF
    \ENDFOR
\RETURN Set of strings $Y$ and $\texttt{feasible}$.
\end{algorithmic}
\end{algorithm}

We first show that the optimal diversity has to be sufficiently large. Let $v^*$ be the optimal diversity. Then,

\begin{lemma}\label{lem:opt-l-bound-k}
Let $D^*$ be the maximum Hamming distance between any two $(1+\epsilon)$-approximate median strings. Then,
\[
v^* \ge \frac{(k-1)(k+1)}{4} \cdot D^*
\]
\end{lemma}

\begin{proof}
Assume there exist two $(1+\epsilon)$-approximate median strings such that the Hamming distance between them is $D^*$. Assume the two strings are $s^*_1, s^*_2$. Then we can create a set of $k$ strings $\{s_1,s_2,\dots,s_k\}$ such that $s_i = s^*_1$ for all $i \le \left \lfloor \frac{k}{2} \right \rfloor$ and $s_i = s^*_2$ for all other $i \in [k]$. The total pairwise distance between the strings in $\{s_1,s_2,\dots,s_k\}$ is,
\[
\left \lfloor \frac{k}{2} \right \rfloor \cdot \left(k-\left \lfloor \frac{k}{2} \right \rfloor\right) D^*
\]
Note that if $k$ is even, $\left \lfloor \frac{k}{2} \right \rfloor \cdot \left(k-\left \lfloor \frac{k}{2}\right \rfloor\right) = \frac{k^2}{4} \ge \frac{(k-1)(k+1)}{4}$ and if $k$ is odd, $\left \lfloor \frac{k}{2} \right \rfloor \cdot \left(k-\left \lfloor \frac{k}{2}\right \rfloor\right) =  \frac{k-1}{2} \cdot \left(k-\frac{k-1}{2}\right) = \frac{(k-1)(k+1)}{4}$. Therefore, $v^* \ge \frac{(k-1)(k+1)}{4} D^*$
\end{proof}

Next, we will show that the prefix considered by the $\sdalg$ algorithm (Algorithm~\ref{alg:density-flip-k}) returns a solution with sufficiently large sum dispersion.

\begin{lemma}\label{lem:opt-v-bound-k}
Let $v$ be the value achieved by the $\sdalg$ algorithm (Algorithm~\ref{alg:density-flip-k}). Then:

\[
v^* \le v + (k - 1)(k + 1)
\]

\end{lemma}

\begin{proof}

Let $L_{\text{thr}}$ be the prefix for which the algorithm returns a \texttt{feasible} solution $Y$. For all $i \in [d]$ and $a \in \Gamma$, let $r_{a,i} = \sum_{y \in Y} \ind{y_i=a}$ and let $S = \{(a,i,r_{a,i}) \forall i \in [d] \text{ and } a \in \Gamma \setminus \{w_i\}\}$. 

Now consider $L_{\text{thr}+1}$ which is the prefix of $L$ such that $|L_{\text{thr}+1}| = 1+|L_{\text{thr}}|$. Run the $\cgrd$ algorithm with $L_{\text{thr}+1}$ as the input. Let $Y'$ be the output of the $\cgrd$ algorithm. For all $i \in [d]$ and $a \in \Gamma$, let $r'_{a,i} = \sum_{y' \in Y'} \ind{y'_i=a}$ and let $S' = \{(a,i,r'_{a,i}) \forall i \in [d] \text{ and } a \in \Gamma \setminus \{w_i\}\}$. 

Note that by~\autoref{lem:opt-offset}, for any $y \in Y$
\[\sum_{x \in X} H(x,y) = \opt + \sum_{i: y_i \neq w_i} (f_i^w-f_i^y)\]
and therefore,
\begin{align*}
\sum_{y \in Y} \sum_{x \in X} H(x,y) &= k \cdot \opt+\sum_{y \in Y} \sum_{i: y_i \neq w_i} (f_i^w-f_i^y)\\
&=k \cdot \opt+\sum_{i \in [d]}\sum_{y \in Y| y_i \neq w_i} (f_i^w-f_i^y)
\end{align*}
Note that for any $i$ such that $y_i = a$ (i.e., $y_i \neq w_i$), the value of $(f_i^w-f_i^y) =  \costind(a,i)$. Also note that for each $i$, there are $r_{a,i}$ strings in $Y$ such that $y_i = a$. Therefore,
\begin{align*}
\sum_{y \in Y} \sum_{x \in X} H(x,y) =k \cdot \opt+\sum_{i \in [d]}\sum_{a \in \Gamma \setminus \{w_i\}} r_{a,i} \costind(a,i)
\end{align*}
A similar argument shows that,
\begin{align*}
\sum_{y' \in Y'} \sum_{x \in X} H(x,y') =k \cdot \opt+\sum_{i \in [d]} \sum_{a \in \Gamma \setminus \{w_i\}} r'_{a,i} \costind(a,i)
\end{align*}

Since the total maximum allowed median cost budget over the $k$ strings is $k \cdot (1+\epsilon) \opt$, we will first define the total residual capacity of the outputs $Y$ and $Y'$ as follows. Let $U$ be the total residual capacity of $Y$. Then,
\begin{align*}
U &= (1+\epsilon)k\cdot \opt - \left(k\cdot \opt+\sum_{i \in [d]} \sum_{a \in \Gamma \setminus \{w_i\}} r_{a,i} \costind(a,i)\right)\\
& = \epsilon k \opt-\sum_{i \in [d]} \sum_{a \in \Gamma \setminus \{w_i\}} r_{a,i} \costind(a,i)
\end{align*}
Similarly, let $U'$ be the total residual capacity of $Y'$. Then,
\begin{align*}
U' &= (1+\epsilon)k\cdot \opt - \left(k\cdot \opt+\sum_{i \in [d]} \sum_{a \in \Gamma \setminus \{w_i\}} r'_{a,i} \costind(a,i)\right)\\
& = \epsilon k \opt-\sum_{i \in [d]} \sum_{a \in \Gamma \setminus \{w_i\}} r'_{a,i} \costind(a,i)
\end{align*}

We will now first relate the two residual capacities $U$ and $U'$ and use that to show that the sum dispersion of the solution $Y$ is sufficiently large. 

Consider $(\hat{i},\hat{a},\hat{r},\ell_{(\hat{i},\hat{a},\hat{r})}) \in L_{\text{thr}+1} \setminus L_{\text{thr}}$ and $(\hat{a}, \hat{i}, r'_{\hat{a}, \hat{i}}) \in S'$. Then, there are only two posibile scenarios for $r'_{\hat{a}, \hat{i}}$ and $\hat{r}$, i.e. $r'_{\hat{a}, \hat{i}} = \hat{r}$ or $r'_{\hat{a}, \hat{i}} = \hat{r}-1$. Note that if $r'_{\hat{a}, \hat{i}} = \hat{r}-1$, then $(\hat{a}, \hat{i}, r'_{\hat{a}, \hat{i}}) \in S$ and therefore, $U'=U$. Otherwise,
\begin{align*}
U' &= k \cdot \epsilon \opt - \sum_{i \in [d]} \sum_{a \in \Gamma \setminus \{w_i\}} r'_{a,i} \costind(a,i)\\
&= k \cdot \epsilon \opt - \sum_{i \in [d] \setminus \{\hat{i}\}} \sum_{a \in \Gamma \setminus \{w_i\}} r'_{a,i} \costind(a,i)\\ 
&- (r'_{\hat{a},\hat{i}}-1) \costind(\hat{a},\hat{i})-\costind(\hat{a},\hat{i})\\
\end{align*}
Note that since $L_{\text{thr}+1}$ gives an \texttt{infeasible} solution, $r'_{a,i} \le r_{a,i}$ for $i \in [d] \setminus \{\hat{i}\}, a \in \Gamma \setminus \{w_i\}$ and $r_{\hat{a},\hat{i}} = (r'_{\hat{a},\hat{i}}-1)$. Therefore, $k \cdot \epsilon \opt - \sum_{i \in [d] \setminus \{\hat{i}\}} r'_{a,i} \costind(a,i) - (r'_{\hat{a},\hat{i}}-1) \costind(a,\hat{i}) \ge k \cdot \epsilon \opt -\sum_{i \in [d]} r_{a,i} \costind(a,i) = U$. Therefore $U' \ge U-\costind(a,\hat{i})$. This implies that $U \le U'+\costind(\hat{a},\hat{i})$ always holds true.

Note that since $Y'$ is an \texttt{infeasible} solution, $\exists (i,a,r,\ell_{(i,a,r)}) \in L_{\text{thr}+1}$ such that $r'_{a,i} < r$ and let $\tilde{i},\tilde{a}$ be the index and character combination corresponding to the highest $\costind(\tilde{a},\tilde{i})$ of such density values. Note that since we cannot increase the $r'_{\tilde{a},\tilde{i}}$, this implies $U' \le k\cdot \costind(\tilde{a},\tilde{i})$. Therefore, 
\begin{align*}
U \le k \cdot \costind(\tilde{a},\tilde{i}) + \costind(\hat{a},\hat{i})
\end{align*}
Note that by the definition of the prefix, $\rho_{\hat{i},\hat{a},\hat{r},\ell_{(\hat{i},\hat{a},\hat{r})}} \le \rho_{\tilde{i},\tilde{a},\tilde{r},\ell_{(\tilde{i},\tilde{a},\tilde{r})}}$. 

Therefore, if the $\costind(\tilde{a},\tilde{i}) > \costind(\hat{a},\hat{i})$, 
\begin{align*}
U \rho_{\hat{i},\hat{a},\hat{r},\ell_{(\hat{i},\hat{a},\hat{r})}} &\le U \rho_{\tilde{i},\tilde{a},\tilde{r},\ell_{(\tilde{i},\tilde{a},\tilde{r})}}\\
&\le (k+1) \cdot \costind(\tilde{a},\tilde{i}) \rho_{\tilde{i},\tilde{a},\tilde{r},\ell_{(\tilde{i},\tilde{a},\tilde{r})}}\\
& \le (k+1)(\ell_{(\tilde{i},\tilde{a},\tilde{r})}-\tilde{r})) \le (k+1)(k-1)
\end{align*}
and if $\costind(\tilde{a},\tilde{i}) \le \costind(\hat{a},\hat{i})$,
\begin{align*}
U \rho_{\hat{i},\hat{a},\hat{r},\ell_{(\hat{i},\hat{a},\hat{r})}} &\le (k+1) \cdot \costind(\hat{a},\hat{i}) \rho_{\hat{i},\hat{a},\hat{r},\ell_{(\hat{i},\hat{a},\hat{r})}}\\
& \le (k+1)(\ell_{(\hat{i},\hat{a},\hat{r})}-\hat{r}) \le (k+1)(k-1)
\end{align*}
Therefore, $U \rho_{\hat{i},\hat{a},\hat{r},\ell_{(\hat{i},\hat{a},\hat{r})}} \le (k+1)(k-1)$.

\begin{remark}\label{rem:all-ell}
Consider the $P_i=\{\ell|\exists \rho_{i,a,r,\ell} \in L_{\text{thr}}\}$. Let $r_{w_i,i}=\sum_{y \in Y} \ind{y_i=w_i}$. Then we can see that for any $\hat{\ell} \in \{r_{w_i,i}+1, \dots, k-1,k\}$, $\hat{\ell} \in P_i$ and $|P_i| = k-r_{w_i,i}$. This is because for any $\ell \le k$, we always keep the occurence with higher $\rho_{i,a,r,\ell}$ in the preprocessing and also for any $\tilde{\ell} < \ell$, if $\rho_{i,\tilde{a},\tilde{r},\tilde{\ell}} \in L_{\text{thr}}$ then $\exists \hat{a},\hat{r} \text{ such that} \rho_{i,\hat{a},\hat{r},\ell} \in L_{\text{thr}}$ (otherwise since $\rho_{i,\tilde{a},\tilde{r},\ell} > \rho_{i,\tilde{a},\tilde{r},\tilde{\ell}}$ preprocessing would remove $\rho_{i,\tilde{a},\tilde{r},\tilde{\ell}}$).
\end{remark}

Given Remark~\ref{rem:all-ell} and the bounds on $U \rho_{\hat{i},\hat{a},\hat{r},\ell_{(\hat{i},\hat{a},\hat{r})}}$, we can now establish the final bounds on the sum dispersion of $Y$. Let $v$ be the sum dispersion value of $Y$ and let $v^*$ be the sum dispersion value of the optimal solution (denoted by $Y^*)$. Let $r^*_{a,i} = |\{y^* \in Y^*| y^*_i = a\}|$. Then, we can see that by the definition of density,
\begin{align*}
&\sum_{i \in [d]} \sum_{a\in \Gamma \setminus \{w_i\}}\sum_{r=1}^{r_{a,i}} \costind(a,i) \rho_{i,a,r,\ell_{(i,a,r)}} \\
&=\sum_{i \in [d]} \sum_{a\in \Gamma \setminus \{w_i\}}\sum_{r=1}^{r_{a,i}} (\ell_{(i,a,r)}-r) \\
&=\sum_{i \in [d]} \sum_{a\in \Gamma \setminus \{w_i\}}\sum_{r=1}^{r_{a,i}} \frac{1}{2} \left((\ell_{(i,a,r)}-1)(k-(\ell_{(i,a,r)}-1)+r(k-r)\right) \\
&-\sum_{i \in [d]} \sum_{a\in \Gamma \setminus \{w_i\}}\sum_{r=1}^{r_{a,i}} \frac{1}{2} \left(\ell_{(i,a,r)}(k-\ell_{(i,a,r)})+(r-1)(k-(r-1))\right) \\
&= \sum_{i \in [d]} \sum_{a\in \Gamma \setminus \{w_i\}}\frac{1}{2} r_{a,i}\left(k-r_{a,i}\right)\\
&+ \sum_{i \in [d]} \sum_{a\in \Gamma \setminus \{w_i\}} \frac{1}{2}\left((\ell_{(i,a,r)}-1)(k-(\ell_{(i,a,r)}-1)-\ell_{(i,a,r)}(k-\ell_{(i,a,r)})\right)\\
&= \sum_{i \in [d]} \sum_{a\in \Gamma}\frac{1}{2} r_{a,i}\left(k-r_{a,i}\right)\\
& = v
\end{align*}
where $\sum_{a\in \Gamma \setminus \{w_i\}} \frac{1}{2}\left((\ell_{(i,a,r)}-1)(k-(\ell_{(i,a,r)}-1)-\ell_{(i,a,r)}(k-\ell_{(i,a,r)})\right) = \frac{1}{2}r_{w_i,i}(k-r_{w_i,i})$ (where $r_{w_i,i} = \{y\in Y|y_i = w_i\}$) since for any $k \ge \ell_{(i,a,r)} \ge r_{w_i,i}$ there is a character $a$ that has $\ell_{(i,a,r)}$ value in the sum (Remark~\ref{rem:all-ell}) and the sum is telescopic.

\begin{claim}\label{clm:sum-opt-rho}
Let $v^*$ be the optimal sum dispersion. For any $a \in \Gamma$ $r^*_{a,i}$ be the number of strings with $a$ at index $i$ in the optimal solution and let $r_{a,i}$ be the number of strings with $a$ at index $i$ in the solution $Y$ from $L_{\text{thr}}$. Then, we have an 
\begin{align*}
v^* &\le v+\left(\sum_{i \in [d]}\left(\sum_{\substack{a \in \Gamma \setminus \{w_i\}\\ r^*_{a,i}>r_{a,i}}} \sum_{r=r_{a,i}+1}^{r^*_{a,i}} \costind(a,i) -\sum_{\substack{a \in \Gamma \setminus \{w_i\}\\ r^*_{a,i}<r_{a,i}}} \sum_{r=r^*_{a,i}+1}^{r_{a,i}} \costind(a,i)\right)\right)\rho_{\hat{i},\hat{a},\hat{r},\ell_{(\hat{i},\hat{a},\hat{r})}}
\end{align*}
\end{claim}
\begin{proof}

Let $v^*_i = \sum_{a \in \Gamma} \frac{1}{2}r^*_{a,i}(k-r^*_{a,i})$. We can see that,
\[v^* = \sum_{i \in [d]} v^*_i\]
and 
\begin{align*}
v^*_i &= \sum_{a \in \Gamma} \frac{1}{2} r^*_{a,i}(k-r^*_{a,i})\\
& = \sum_{a \in \Gamma \setminus \{w_i\}} \sum_{r=1}^{r^*_{a,i}} \frac{1}{2} \left(r(k-r)-(r-1)(k-(r-1))\right)+\frac{1}{2} \left(r^*_{w_i,i}(k-r^*_{w_i,i})\right)\\
& = \sum_{a \in \Gamma \setminus \{w_i\}} \sum_{r=1}^{r^*_{a,i}} \frac{1}{2} \left(k-2r+1\right)+\frac{1}{2} \left(r^*_{w_i,i}(k-r^*_{w_i,i})\right)\\
\end{align*}
We can reorder the sums to get the following, 
\begin{align*}
v^*_i & = \sum_{a \in \Gamma \setminus \{w_i\}} \sum_{r=1}^{r_{a,i}} \frac{1}{2} \left(k-2r+1\right)+\sum_{\substack{a \in \Gamma \setminus \{w_i\}\\r_{a,i}<r^*_{a,i}}} \sum_{r=r_{a,i}+1}^{r^*_{a,i}} \frac{1}{2} \left(k-2r+1\right)\\
&-\sum_{\substack{a \in \Gamma \setminus \{w_i\}\\r_{a,i}>r^*_{a,i}}} \sum_{r=r^*_{a,i}+1}^{r_{a,i}} \frac{1}{2} \left(k-2r+1\right)+\frac{1}{2} \left(r^*_{w_i,i}(k-r^*_{w_i,i})\right)\\
\end{align*}

\begin{remark}\label{rem:ordering}
Let $a,b \in \Gamma \setminus \{w_i\}$ be two characters such that $\costind(a,i) \le \costind(b,i)$, then we can assume there exists an optimal solution with $r^*_{a,i} \ge r^*_{b,i}$ for all such pairs. Assume otherwise, i.e. $r_1=r^*_{b,i} > r^*_{a,i}=r_2$. Then take any $r_1-r_2$ strings with $b$ in the index $i$ and change them to $a$. Since $\costind(a,i) \le \costind(b,i)$ it still gives a feasible solution, and this gives a new solution with $r^*_{a,i} = r_1$ and $r^*_{b,i} = r_2$, which gives the exact same sum dispersion value as before. Therefore, our ordering assumption is valid. 
\end{remark}
Let $a\in \Gamma \setminus \{w_i\}$ be such that $r^*_{a,i} < r_{a,i}$ and let $M_i=\{\ell_{(i,a,r)}\in L_{\text{thr}}| r^*_{a,i} < r\le  r_{a,i}\}$. Note that since $\sum_{a \in \Gamma \setminus \{w_i\}} r_{a,i}+r_{w_i,i} = k =\sum_{a \in \Gamma \setminus \{w_i\}} r^*_{a,i}+r^*_{w_i,i}$, we can see that, 
\[(r_{w_i,i}-r^*_{w_i,i})+\sum_{\substack{a \in \Gamma \setminus \{w_i\}\\r^*_{a,i} < r_{a,i}}} (r_{a,i}-r^*_{a,i})=\sum_{\substack{a \in \Gamma \setminus \{w_i\}\\r^*_{a,i} > r_{a,i}}} (r^*_{a,i}-r_{a,i})\]
Assuming $r^*_{w_i,i} < r_{w_i,i}$, we can create the set $M'_i = M_i \cup \{\ell| r_{w_i,i} \ge \ell > r^*_{w_i,i}\}$. Note that while we assume $r^*_{w_i,i} < r_{w_i,i}$ for ease of argument this can be easily removed by setting $M'_i$ to be a subset of $M_i$ of size $\sum_{\substack{a \in \Gamma \setminus \{w_i\}\\r^*_{a,i} > r_{a,i}}} (r^*_{a,i}-r_{a,i})$ when $r^*_{w_i,i} \ge r_{w_i,i}$. Now, consider $a\in \Gamma \setminus \{w_i\}$ where $r^*_{a,i}> r_{a,i}$ and order them in the non-increasing order of $\costind(a,i)$, and for all $r \in \{r_{a,i}+1,\dots,r^*_{a,i}\}$ assign the values in $M'_i$ in the sorted order. Let these be indicated by $\ell^*_{(i,a,r)}$. Then, we can see that,
\begin{align*}
v^*_i & = \sum_{a \in \Gamma \setminus \{w_i\}} \sum_{r=1}^{r_{a,i}} \frac{1}{2} \left(k-2r+1-(k-2\ell_{(i,a,r)}+1)\right)\\
&+\sum_{\substack{a \in \Gamma \setminus \{w_i\}\\r_{a,i}<r^*_{a,i}}} \sum_{r=r_{a,i}+1}^{r^*_{a,i}} \frac{1}{2} \left(k-2r+1-(k-2\ell^*_{(i,a,r)}+1)\right)\\
&-\sum_{\substack{a \in \Gamma \setminus \{w_i\}\\r_{a,i}>r^*_{a,i}}} \sum_{r=r^*_{a,i}+1}^{r_{a,i}} \frac{1}{2} \left(k-2r+1-(k-2\ell_{(i,a,r)}+1)\right)+\frac{1}{2} \left(r^*_{w_i,i}(k-r^*_{w_i,i})\right)\\
& + \sum_{a \in \Gamma \setminus \{w_i\}} \sum_{r=1}^{r_{a,i}} \frac{1}{2} (k-2\ell_{(i,a,r)}+1)+\sum_{\substack{a \in \Gamma \setminus \{w_i\}\\r_{a,i}<r^*_{a,i}}} \sum_{r=r_{a,i}+1}^{r^*_{a,i}} \frac{1}{2} (k-2\ell^*_{(i,a,r)}+1)\\
&-\sum_{\substack{a \in \Gamma \setminus \{w_i\}\\r_{a,i}>r^*_{a,i}}} \sum_{r=r^*_{a,i}+1}^{r_{a,i}} \frac{1}{2} (k-2\ell_{(i,a,r)}+1)
\end{align*}
Note that,
\begin{align*}
&\sum_{a \in \Gamma \setminus \{w_i\}} \sum_{r=1}^{r_{a,i}} \frac{1}{2} (k-2\ell_{(i,a,r)}+1)+\sum_{\substack{a \in \Gamma \setminus \{w_i\}\\r_{a,i}<r^*_{a,i}}} \sum_{r=r_{a,i}+1}^{r^*_{a,i}} \frac{1}{2} (k-2\ell^*_{(i,a,r)}+1)\\
&-\sum_{\substack{a \in \Gamma \setminus \{w_i\}\\r_{a,i}>r^*_{a,i}}} \sum_{r=r^*_{a,i}+1}^{r_{a,i}} \frac{1}{2} (k-2\ell_{(i,a,r)}+1)\\
&=\sum_{a \in \Gamma \setminus \{w_i\}} \sum_{r=1}^{r^*_{a,i}} \frac{1}{2} (k-2\ell_{(i,a,r)}+1)=\sum_{r=r^*_{w_i,i}+1}^k \frac{1}{2} (k-2r+1)
\end{align*}
where the last equality comes from the construction of $M'_i$. Therefore,
\begin{align*}
v^*_i & = \sum_{a \in \Gamma \setminus \{w_i\}} \sum_{r=1}^{r_{a,i}} \left(\ell_{(i,a,r)}-r\right)+\sum_{\substack{a \in \Gamma \setminus \{w_i\}\\r_{a,i}<r^*_{a,i}}} \sum_{r=r_{a,i}+1}^{r^*_{a,i}} \left(\ell_{(i,a,r)}-r\right)\\
&-\sum_{\substack{a \in \Gamma \setminus \{w_i\}\\r_{a,i}>r^*_{a,i}}} \sum_{r=r^*_{a,i}+1}^{r_{a,i}} \left(\ell_{(i,a,r)}-r\right)+\frac{1}{2} \left(r^*_{w_i,i}(k-r^*_{w_i,i})\right)+\sum_{r = r^*_{w_i,i}+1}^k \frac{1}{2}(k-2r+1)\\
& = \sum_{a \in \Gamma \setminus \{w_i\}} \sum_{r=1}^{r_{a,i}} \left(\ell_{(i,a,r)}-r\right)+\sum_{\substack{a \in \Gamma \setminus \{w_i\}\\r_{a,i}<r^*_{a,i}}} \sum_{r=r_{a,i}+1}^{r^*_{a,i}} \left(\ell^*_{(i,a,r)}-r\right)\\
&-\sum_{\substack{a \in \Gamma \setminus \{w_i\}\\r_{a,i}>r^*_{a,i}}} \sum_{r=r^*_{a,i}+1}^{r_{a,i}} \left(\ell_{(i,a,r)}-r\right)
\end{align*}
which gives,
\begin{align*}
v^* &= \sum_{i \in [d]}\sum_{a \in \Gamma \setminus \{w_i\}} \sum_{r=1}^{r_{a,i}} \costind(a,i) \rho_{i,a,r,\ell_{(i,a,r)}}\\
&+\sum_{i \in [d]}\sum_{\substack{a \in \Gamma \setminus \{w_i\}\\ r^*_{a,i}>r_{a,i}}} \sum_{r=r_{a,i}+1}^{r^*_{a,i}} \costind(a,i) \rho_{i,a,r,\ell^*_{(i,a,r)}} \\
&-\sum_{i \in [d]}\sum_{\substack{a \in \Gamma \setminus \{w_i\}\\ r^*_{a,i}<r_{a,i}}} \sum_{r=r^*_{a,i}+1}^{r_{a,i}} \costind(a,i) \rho_{i,a,r,\ell_{(i,a,r)}}\\
&= v+\sum_{i \in [d]}\sum_{\substack{a \in \Gamma \setminus \{w_i\}\\ r^*_{a,i}>r_{a,i}}} \sum_{r=r_{a,i}+1}^{r^*_{a,i}} \costind(a,i) \rho_{i,a,r,\ell^*_{(i,a,r)}} \\
&-\sum_{i \in [d]}\sum_{\substack{a \in \Gamma \setminus \{w_i\}\\ r^*_{a,i}<r_{a,i}}} \sum_{r=r^*_{a,i}+1}^{r_{a,i}} \costind(a,i) \rho_{i,a,r,\ell_{(i,a,r)}}
\end{align*}

We observe the following about the given sum.   
\begin{enumerate}
\item \textbf{Pairing (via $M_i,M'_i$).} Every assignment 
where $r_{d,i} \ge \tilde r>r^*_{d,i}$ (which we will call deficit) can be matched to an assignment 
$r_{c,i}<\hat r \le r^*_{c,i}$ (which we call surplus) through the construction of $M_i,M'_i$ where they have $\ell_{(i,d,\tilde r)} = \ell^*_{(i,c,\hat r)}$, so that surplus and deficit terms appear in disjoint pairs.  

\item \textbf{Ordering (Remark~\ref{rem:ordering}).} If $(c,\hat r)$ is a surplus
and $(d,\tilde r)$ is a deficit in the same pair, then $\costind(c,i)\le \costind(d,i)$.  

\item \textbf{Monotonicity of $\rho$.} For a fixed $\ell$, if $(c,\hat r)$ is a surplus
and $(d,\tilde r)$ is a deficit in the same pair, since the algorithm
uses $\rho$ values in non-increasing order, so $\rho_{i,d,\tilde r,\ell}\ge \rho_{i,c,\hat r,\ell}$.  
\end{enumerate}

Consider now the contribution of one such pair with common index $\ell$:  
\begin{align*}
&\costind(c,i)\,\rho_{i,c,\hat r,\ell}
-\costind(d,i)\,\rho_{i,d,\tilde r,\ell} \\
&\le \costind(c,i)\,\rho_{i,d,\tilde r,\ell}
-\costind(d,i)\,\rho_{i,d,\tilde r,\ell} \\
&= (\costind(c,i)-\costind(d,i))\,
\rho_{i,d,\tilde r,\ell}.
\end{align*}
Here, the inequality uses $\rho_{i,d,\tilde r,\ell}\ge \rho_{i,c,\hat r,\ell}$, 
and the equality is just from factoring. Since 
$\costind(c,i)-\costind(d,i)\le 0$, this contribution is non-positive. Also note that since $\rho_{i,d,\tilde r,\ell} \in L_{\text{thr}}$, we get, $\rho_{i,d,\tilde r,\ell} \ge \rho_{\hat{i},\hat{a},\hat{r},\ell_{(\hat{i},\hat{a},\hat{r})}}$ which gives,
\begin{align*}
&\costind(c,i)\,\rho_{i,c,\hat r,\ell}
-\costind(d,i)\,\rho_{i,d,\tilde r,\ell} \\
&\le (\costind(c,i)-\costind(d,i))\,
\rho_{\hat{i},\hat{a},\hat{r},\ell_{(\hat{i},\hat{a},\hat{r})}}.
\end{align*}

For the unpaired terms in $M'_i\setminus M_i$ note that the corresponding $\rho$ values are less than or equal to $\rho_{\hat{i},\hat{a},\hat{r},\ell_{(\hat{i},\hat{a},\hat{r})}}$ by definiton (since they come from surplus terms) and therefore, 
\[\costind(c,i)\,\rho_{i,c,\hat r,\ell} \le \costind(c,i)\,\rho_{\hat{i},\hat{a},\hat{r},\ell_{(\hat{i},\hat{a},\hat{r})}}.\] 
Therefore, the sum,
\[\sum_{\substack{a \in \Gamma \setminus \{w_i\}\\ r^*_{a,i}>r_{a,i}}} \sum_{r=r_{a,i}+1}^{r^*_{a,i}} \costind(a,i) \rho_{i,a,r,\ell^*_{(i,a,r)}} -\sum_{\substack{a \in \Gamma \setminus \{w_i\}\\ r^*_{a,i}<r_{a,i}}} \sum_{r=r^*_{a,i}+1}^{r_{a,i}} \costind(a,i) \rho_{i,a,r,\ell_{(i,a,r)}}\] 
is at most sum of $(\costind(c,i)-\costind(d,i))\,
\rho_{\hat{i},\hat{a},\hat{r},\ell_{(\hat{i},\hat{a},\hat{r})}}$ over the paired indices and the sum of $\costind(c,i)\,
\rho_{\hat{i},\hat{a},\hat{r},\ell_{(\hat{i},\hat{a},\hat{r})}}$ over the additional surplus and since $\rho_{\hat{i},\hat{a},\hat{r},\ell_{(\hat{i},\hat{a},\hat{r})}}$ is a common factor this gives us,
\begin{align*}
&\sum_{\substack{a \in \Gamma \setminus \{w_i\}\\ r^*_{a,i}>r_{a,i}}} \sum_{r=r_{a,i}+1}^{r^*_{a,i}} \costind(a,i) \rho_{i,a,r,\ell^*_{(i,a,r)}} -\sum_{\substack{a \in \Gamma \setminus \{w_i\}\\ r^*_{a,i}<r_{a,i}}} \sum_{r=r^*_{a,i}+1}^{r_{a,i}} \costind(a,i) \rho_{i,a,r,\ell_{(i,a,r)}}\\
& \le \left(\sum_{\substack{a \in \Gamma \setminus \{w_i\}\\ r^*_{a,i}>r_{a,i}}} \sum_{r=r_{a,i}+1}^{r^*_{a,i}} \costind(a,i)  -\sum_{\substack{a \in \Gamma \setminus \{w_i\}\\ r^*_{a,i}<r_{a,i}}} \sum_{r=r^*_{a,i}+1}^{r_{a,i}} \costind(a,i) \right)\rho_{\hat{i},\hat{a},\hat{r},\ell_{(\hat{i},\hat{a},\hat{r})}}
\end{align*}
{\bf Note:} If $r^*_{w_i,i} > r_{w_i,i}$, we would have more deficit than surplus, but we can still use the pairing argument and then use the fact that for all unpaired deficit entries $(d,\tilde r)$, 
\[-\costind(d,i)\,\rho_{i,d,\tilde r,\ell} \le -\costind(d,i)\,\rho_{\hat{i},\hat{a},\hat{r},\ell_{(\hat{i},\hat{a},\hat{r})}}.\] 
to show that the bound is still satisfied.

Summing over all indices, we therefore obtain
\[
v^*\le v +
\left(\sum_{i \in [d]}\left(\sum_{\substack{a \in \Gamma \setminus \{w_i\}\\ r^*_{a,i}>r_{a,i}}} \sum_{r=r_{a,i}+1}^{r^*_{a,i}} \costind(a,i)  -\sum_{\substack{a \in \Gamma \setminus \{w_i\}\\ r^*_{a,i}<r_{a,i}}} \sum_{r=r^*_{a,i}+1}^{r_{a,i}} \costind(a,i) \right)\right)\rho_{\hat{i},\hat{a},\hat{r},\ell_{(\hat{i},\hat{a},\hat{r})}},
\]
which is precisely the desired inequality.
\end{proof}

Note that the Claim~\ref{clm:sum-opt-rho}, gives,
\begin{align*}
v^* \le v+\left(\sum_{i \in [d]}\left(\sum_{\substack{a \in \Gamma \setminus \{w_i\}\\ r^*_{a,i}>r_{a,i}}} \sum_{r=r_{a,i}+1}^{r^*_{a,i}} \costind(a,i) -\sum_{\substack{a \in \Gamma \setminus \{w_i\}\\ r^*_{a,i}<r_{a,i}}} \sum_{r=r^*_{a,i}+1}^{r_{a,i}} \costind(a,i)\right)\right)\rho_{\hat{i},\hat{a},\hat{r},\ell_{(\hat{i},\hat{a},\hat{r})}}
\end{align*}
Note that $\sum_{i \in [d]} \sum_{a \in \Gamma \setminus \{w_i\}} r^*_{a,i} \costind(a,i) \le (1+\epsilon)k \opt = U+\sum_{i \in [d]} \sum_{a \in \Gamma \setminus \{w_i\}} r_{a,i} \costind(a,i)$. Therefore,
\begin{align*}
&\sum_{i \in [d]}\sum_{\substack{a \in \Gamma \setminus \{w_i\}\\ r^*_{a,i}>r_{a,i}}} \sum_{r=r_{a,i}+1}^{r^*_{a,i}} \costind(a,i) -\sum_{i \in [d]}\sum_{\substack{a \in \Gamma \setminus \{w_i\}\\ r^*_{a,i}<r_{a,i}}} \sum_{r=r^*_{a,i}+1}^{r_{a,i}} \costind(a,i) \\
&=  \sum_{i \in [d]} \sum_{a \in \Gamma \setminus \{w_i\}} r^*_{a,i} \costind(a,i)-\sum_{i \in [d]} \sum_{a \in \Gamma \setminus \{w_i\}} r_{a,i} \costind(a,i)\\
& \le U
\end{align*}
Therefore,
\begin{align*}
v^* & \le v+U \cdot \rho_{\hat{i},\hat{a},\hat{r},\ell_{(\hat{i},\hat{a},\hat{r})}}\\
& \le v+(k+1)(k-1)
\end{align*}
as desired.
\end{proof}

Now, we prove~\autoref{thm:sum-dispersion-largeD}. 

\begin{proof}[Proof of~\autoref{thm:sum-dispersion-largeD}]
\autoref{lem:opt-v-bound-k} implies
\[
\frac{v}{v^*} \ge \frac{v^*-(k-1)(k+1)}{v^*} = 1 - \frac{(k - 1)(k + 1)}{v^*}
\]
and~\autoref{lem:opt-l-bound-k} implies
\[
v^* \ge \frac{(k-1)(k+1)}{4} D^* \Rightarrow \frac{(k - 1)(k + 1)}{v^*} \le \frac{4(k-1)(k + 1)}{(k-1)(k+1) D^*} 
\]
and therefore,
\[
\frac{(k - 1)(k + 1)}{v^*} \le \frac{4}{D^*} 
\]
Combining these, we get,
\[
v \ge \left(1-\frac{4}{D^*}\right)v^*
\]

We can see that in this case, the following runtime guarantees hold. Note that using $O(nd)$ memory, we can calculate $w$ and also calculate all $\costind(a,i)$ for $i \in [d]$ and $a \in \Gamma$ in $O(nd|\Gamma|)$ time in the worst case. Since there are $O(dk^2 |\Gamma|)$ entries in $L$, sorting $L$ takes $O(dk^2 |\Gamma|\log \left(dk^2 |\Gamma|\right))$ time and preprocessing takes $O(d^2k^4 |\Gamma|^2)$ and finding $L_{\text{thr}}$ in $\sdalg$ algorithm takes $O(\log dk^2 |\Gamma|)$ searches (using binary search) and for each search, the loop takes $O(dk^3 |\Gamma|)$ time. Finding the final solution takes $O(dk^3 |\Gamma|)$ time. Therefore, overall time complexity is $O(nd|\Gamma|+d^2k^4 |\Gamma|^2)$.
\end{proof}

With~\autoref{thm:sum-dispersion-largeD} established, we are ready to prove~\autoref{thm:sum-dispersion-k}.

\begin{proof}[Proof of~\autoref{thm:sum-dispersion-k}]
Let $\delta>0$ be some parameter. We will first consider the case when $D^*$ is large, i.e., $D^* \ge \frac{4}{\delta}$. Then as a direct implication of~\autoref{thm:sum-dispersion-largeD} we get, $v \ge (1-\delta) v^*$.

Now consider the case when $D^* < \frac{4}{\delta}$. We will show that there exists a PTAS that gives $v \ge (1-\delta)v^*$. We can see that in this case, for any $(1+\epsilon)$ approximate median string $s$, it can differ from $w$ in at most $D^*$ indices. Therefore, we can see that there can be at most $N = |\Gamma|^{\frac{4}{\delta}}d^{\frac{4}{\delta}}$ such strings. We claim that in this case, since we have a bound (that is, polynomial in $d,|\Gamma|$) on the number of strings, there is a PTAS for this instance. To see this, we will utilize the work of~\cite{cevallos2015maxsum}, which shows that if we consider the max sum dispersion of a set of points under a ``negative-type'' metric and matroid constraints, the problem admits a PTAS through a rounding of a quadratic program. We will first show that Hamming distance is a ``negative-type'' metric. A ``negative-type'' metric is defined as follows: Let $S = \{s^{(1)},s^{(2)},\dots,s^{(n)}\}$ be a given set of points and $d(.,.)$ be a metric. $d(.,.)$ is called a ``negative-type'' metric if for any $\{b_1,b_2,\dots,b_n\}$ where $\sum_{i=1}^n b_i = 0$, $\sum_{i,j} b_ib_j d(s^{(i)},s^{(j)}) \le 0$.

In order to show that Hamming distance is a ``negative-type'' metric, we can consider the following mapping of the strings in $\Gamma^d$ to vectors in $\mathbf{R}^{|\Gamma|d}$. We will use this mapping to show that the Hamming distance in the strings is equivalent to the $\ell_2$ distance in the mapped vectors. Let $s$ be any string. Assuming $\Gamma$ has some lexicographical ordering (if not, we can assign some ordering), for each $a \in \Gamma$, we can define a unique basis vector $e_a$ where $e_a$ has $1$ in the index of $a$ and $0$ everywhere else. Then, we define the mapping for the string $s$ as,
\[\phi(s) = \frac{1}{\sqrt{2}} \begin{bmatrix}e_{s_1} & e_{s_2} & \dots & e_{s_i} & \dots & e_{s_d}\end{bmatrix}\]
We will now show that for any two strings $s,\hat{s}$, $H(s,\hat{s}) = \|\phi(s)-\phi(\hat{s})\|^2$. Note that, 
\begin{align*}
\|\phi(s)-\phi(\hat{s})\|^2 = \frac{1}{2} \sum_{i=1}^d \|e_{s_i}-e_{\hat{s}_i}\|^2 =  \sum_{i=1}^d \frac{1}{2}\|e_{s_i}-e_{\hat{s}_i}\|^2
\end{align*}
It is easy to see that if $s_i = \hat{s}_i$ then $\frac{1}{2}\|e_{s_i}-e_{\hat{s}_i}\|^2 = 0$ and otherwise it is $\frac{1}{2}\|e_{s_i}-e_{\hat{s}_i}\|^2 = \frac{1}{2}\cdot 2=1$. Therefore, 
\begin{align*}
\|\phi(s)-\phi(\hat{s})\|^2 =  \sum_{i=1}^d \ind{s_i \neq \hat{s}_i} = H(s,\hat{s})
\end{align*}
Therefore,
\begin{align*}
&\sum_{i=1}^n \sum_{j=1}^n b_i b_j H(s^{(i)},s^{(j)}) \\
&= \sum_{i=1}^n \sum_{j=1}^n b_i b_j \|\phi(s^{(i)})-\phi(s^{(j)})\|^2 \\
&= \sum_{i=1}^n \sum_{j=1}^n b_i b_j \left(\langle \phi(s^{(i)}),\phi(s^{(i)})\rangle+\langle \phi(s^{(j)}),\phi(s^{(j)})\rangle-2\langle \phi(s^{(i)}),\phi(s^{(j)})\rangle\right) \\
&= \sum_{i=1}^n \sum_{j=1}^n b_i b_j \langle \phi(s^{(i)}),\phi(s^{(i)})\rangle+ \sum_{i=1}^n \sum_{j=1}^n b_i b_j \langle \phi(s^{(j)}),\phi(s^{(j)})\rangle-2\sum_{i=1}^n \sum_{j=1}^n b_i b_j \langle \phi(s^{(i)}),\phi(s^{(j)})\rangle \\
&= \sum_{j=1}^n b_j \left(\sum_{i=1}^n b_i  \langle \phi(s^{(i)}),\phi(s^{(i)})\rangle\right)+ \sum_{i=1}^n b_i \left(\sum_{j=1}^n b_j \langle \phi(s^{(j)}),\phi(s^{(j)})\rangle\right)-2\sum_{i=1}^n \sum_{j=1}^n b_i b_j \langle \phi(s^{(i)}),\phi(s^{(j)})\rangle \\
&= -2\sum_{i=1}^n \sum_{j=1}^n b_i b_j \langle \phi(s^{(i)}),\phi(s^{(j)})\rangle \quad \text{(Since $\sum_{i=1}^n b_i = 0$)}\\
&=-2\langle\sum_{i=1}^n \phi(s^{(i)}),\sum_{j=1}^n \phi(s^{(j)})\rangle = -2\|\sum_{i=1}^n \phi(s^{(i)})\|^2 \le 0
\end{align*}
Therefore, Hamming distance is a ``negative-type'' metric. Given this, we can consider the quadratic programming framework of~\cite{cevallos2015maxsum}. Let $\hat{X}$ be the potential $(1+\epsilon)$ approximate medians (we know that $|\hat{X}| \le N$). Let $M \in \mathbf{R}^{|\hat{X}| \times |\hat{X}|}$ be the distance matrix for $\hat{X}$. Then, we consider the quadratic program,
\begin{align*}
&\text{maximize } x^\top M x\\
&\text{subject to }\\
& \sum_{i=1}^{|\hat{X}|} x_i = k\\
& x_i \in \{0,1\} \text{ for all } i \in [|\hat{X}|]
\end{align*}
where $x_i$ acts as an indicator to indicate whether the $i$th string in $\hat{X}$ has been selected. Note that the only constraint we have is $\sum_{i=1}^{|\hat{X}|} x_i = k$ and we can represent this using a matroid $(E,I)$ where $E = \hat{X}$ and $I = \{S \subseteq \hat{X}|\;|S| \le k\}$. Therefore, we get a matroid constraint, and our metric is a ``negative-type'' metric. Therefore, we can apply the PTAS of~\cite{cevallos2015maxsum} to find a solution such that the sum dispersion $v \ge (1-\delta)v^*$. Since this algorithm is polynomial on $N,d,k$, we get that it is polynomial on $n,d,k,|\Gamma|$ as well (since $N$ is polynomial in $d,|\Gamma|$). This completes the proof of~\autoref{thm:sum-dispersion-k}.
\end{proof}

\section{Approximation Algorithm for Min Dispersion: \texorpdfstring{$k$}{k} Hamming Medians} 
\label{app:min-dispersion-exact}

\exactkon*

In the rest of this section, we will establish~\autoref{thm:exact-kon}. First, we derive a dynamic programming–based algorithm that exactly solves the minimum dispersion problem. Formally,

\begin{lemma}\label{lem:min-dispersion-DP}
Given a set of strings $X$ and a parameter $k$, there exists an algorithm that finds the min dispersion in $O(nd\log \min\left(n,|\Gamma|\right)+|\Gamma|^k d^{{\frac{k(k-1)}{2}}+1})$ time.
\end{lemma}
\begin{proof}
We first define the following dynamic program: 

For each index $i \in [d]$, let $\Gamma_i = \{e \in \Gamma : |{x \in X : x_i = e}| = \freq{i}{w}\}$. Note that the sets $\Gamma_i$ can be computed in $O(nd\log \min\left(n,|\Gamma|\right))$ time in the worst case by using $O(nd)$ memory. Define a dynamic program $\textsc{MinDispersion-DP}[d_{1,2}, d_{1,3}, \dots, d_{k-1,k}, \ell]$ such that,
\[
\textsc{MinDispersion-DP}[d_{1,2}, d_{1,3}, \dots, d_{k-1,k}, \ell] = \text{True}
\]
if there exist $k$ strings $s_1,\dots,s_k \in \Gamma_1 \times \cdots \times \Gamma_\ell$ such that $H(s_i,s_j)=d_{i,j}$ for all pairs. The DP state space has size $d^{\frac{k(k-1)}{2}} \cdot d$ (since each $d_{i,j}$ can take values up to $d$ and since $\ell$ is bounded by $d$). Suppose we want to compute 
\[
\textsc{MinDispersion-DP}(d_{1,2},\dots,d_{k-1,k}, \ell+1).
\]
For this to be true, there must exist a set of $k$ strings $\{s_1,\ldots,s_k\}$ in $\Gamma_1 \times \cdots \times \Gamma_\ell$ and an assignment of characters $\{a_1,\dots,a_k\}$ ($a_i \in \Gamma_{\ell+1}$ for all $i\in [k]$) such that
\[
\textsc{MinDispersion-DP}(d_{1,2} - \ind{a_1 \neq a_2}, \dots, d_{i,j} - \ind{a_i \neq a_j}, \dots, d_{k-1,k} - \ind{a_{k-1} \neq a_k}, \ell)
\]
is marked True. Since there are $|\Gamma_{\ell+1}|^k \le |\Gamma|^k$ possible assignments $\{a_1,\dots,a_k\}$, each state update requires at most $|\Gamma|^k$ operations giving overall runtime $O(|\Gamma|^k d^{{\frac{k(k-1)}{2}}+1})$. To extract the solution, we simply check all states of the form $(d_{1,2},\ldots,d_{k-1,k},d)$, which takes $O(d^{\frac{k(k-1)}{2}})$ time. 
\end{proof}

Next, we consider the case where $D^*$ is large enough, and show that one can obtain a $(1-\delta)$-approximation to the minimum dispersion.
\begin{lemma}\label{lem:uniform-exact}
Let $D^* \ge \frac{4}{\delta^2} \left(2\log k+1\right)$ where $\delta>0$. Then there exists an algorithm that, with probability at least $1-\eta$, outputs a set of Hamming median strings $S$ such that $\textsc{minDp}(S) \ge (1-\delta)\sum_{i \in [d]} \frac{|\Gamma_i|-1}{|\Gamma_i|}$ in $O(nd\log \min \left(n,|\Gamma|\right)+\left(kd|\Gamma|+k^2d\right) \log \frac{1}{\eta})$ time, where for all $i \in [d]$, $\Gamma_i = \{e \in \Gamma : |{x \in X : x_i = e}| = \freq{i}{w}\}$.
\end{lemma}

\begin{proof}
Consider the following set of strings $S$ such that $|S| = k$: For any $s \in S$, for any $i \in [d]$, $s_i$ is uniformly sampled from $\Gamma_i$.  Note that for any $s \in S$, $s$ is a Hamming median. This comes from a direct application of~\autoref{lem:MCC-opt} since for any $i \in [d]$, $f_i^w = f_i^s$ by definition. 

Let $T \subseteq [d]$ be the set of indices where $|\Gamma_i| \ge 2$. Note that for all $i \not \in T$, $\Gamma_i = \{w_i\}$ and $|T| = D^*$. 

Consider any two strings $\hat{s},\tilde{s} \in S$. For all $i \in [d]$, let $z_i=1$ if $\hat{s}_i \neq \tilde{s}_i$ and $0$ otherwise. Then, $Pr(z_i = 1) = 1-\frac{1}{|\Gamma_i|}$. Furthermore, $H\left(\hat{s},\tilde{s}\right) = \sum_{i\in [d]} z_i$. Let $\mu = \exv\left(H\left(\hat{s},\tilde{s}\right)\right)$. Note that, $\mu = \exv \left(  \sum_{i\in [d]} z_i \right) = \sum_{i \in [d]} \exv \left(z_i\right) = \sum_{i \in [d]} \left(1-\frac{1}{|\Gamma_i|}\right)$. Also note that since for all $i \not \in T$, $|\Gamma_i| = 1$  and  for all $i \in T$, $|\Gamma_i|\ge 2$, we get $\mu = \sum_{i \in [d]} \left(1-\frac{1}{|\Gamma_i|}\right) = \sum_{i \in T} \left(1-\frac{1}{|\Gamma_i|}\right) \ge \sum_{i \in T} \frac{1}{2} = \frac{|T|}{2} = \frac{D^*}{2}$. Then, 
\begin{align*}
Pr\left(H\left(\hat{s},\tilde{s}\right) \le (1-\delta)\mu\right) &\le e^{-\frac{\delta^2 \mu}{2}} \quad \text{(by Chernoff bounds)}\\
& \le e^{-\frac{\delta^2 D^*}{4}} \quad \text{(Since $\mu \ge \frac{D^*}{2}$)} \\
& = \frac{1}{2k^2}
\end{align*}
Note that since $|S|=k$, there are at most $k^2$ pairs of strings $\hat{s},\tilde{s}$ in $S$. Therefore, by union bound,
\begin{align*}
Pr\left(\exists \hat{s},\tilde{s} \in S \text{ s.t. } H\left(\hat{s},\tilde{s}\right) \le (1-\delta)\mu\right) \le k^2 \frac{1}{2k^2} = \frac{1}{2}
\end{align*}

We can execute this process $N=\log \frac{1}{\eta}$ times and then select the output that maximizes the min dispersion. Let $\{S^{(1)},S^{(2)},\dots,S^{(N)}\}$ represent the $N$ outputs obtained by repeating the sampling process $N$ times, and let $\hat{S} \in \{S^{(1)},S^{(2)},\dots,S^{(N)}\}$ denote the output with the maximum min dispersion. Then, we observe that
\begin{align*}
Pr\left(\forall i \in [N] \exists \hat{s},\tilde{s} \in S^{(i)} \text{ s.t. } H\left(\hat{s},\tilde{s}\right) \le (1-\delta)\mu\right) \le \frac{1}{2^N} = \eta
\end{align*}
and consequently, with probability at least $1-\eta$, there exists at least one $S^{(i)}$ such that $\forall \hat{s},\tilde{s} \in S^{(i)}$, $H\left(\hat{s},\tilde{s}\right) \ge (1-\delta)\sum_{i \in [d]} \frac{|\Gamma_i|-1}{|\Gamma_i|}$. By definition, since $\hat{S}$ is the output achieving the maximum min dispersion, this implies $\forall \hat{s},\tilde{s} \in \hat{S}$, $H\left(\hat{s},\tilde{s}\right) \ge (1-\delta)\sum_{i \in [d]} \frac{|\Gamma_i|-1}{|\Gamma_i|}$, and therefore $\hat{S}$ is a solution.

Note that we can use $O(nd)$ memory and find all the $\Gamma_i$ in $O(nd\log \min \left(n,|\Gamma|\right))$ time. Since the algorithm involves generating $k$ randomly sampled strings, the time complexity for generating $k$ strings is bounded by $O(kd|\Gamma|)$. Therefore, overall time complexity for sampling and generation is $O\left(nd\log \min \left(n,|\Gamma|\right)+kd|\Gamma|\right)$. Since we repeat this process for $\log \frac{1}{\eta}$ steps, we get a total runtime of $O\left(nd\log \min \left(n,|\Gamma|\right)+kd|\Gamma| \log \frac{1}{\eta} \right)$. Since for each $S^{(i)}$ we need to calculate the min dispersion and we need to select the maximum, we get additional $O(k^2d \log\frac{1}{\eta})$. Therefore, the overall time complexity is $O(nd\log \min \left(n,|\Gamma|\right)+\left(kd|\Gamma|+k^2d\right) \log \frac{1}{\eta})$. 
\end{proof}

Next, we consider the scenario in which $D^*$ is small, and demonstrate that a solution achieving a $1/2$-approximation to the minimum dispersion can be obtained.
\begin{lemma}\label{lem:exact-small}
Let $D^* < \frac{4}{\delta^2} \left(2\log k+1\right)$ where $\delta>0$. Then there exists an algorithm that outputs a set of $k$ Hamming median strings $S$ such that $\textsc{minDp}(S) \ge \frac{1}{2}t^*$ in $O(nd\log \min\left(n,|\Gamma|\right)+|\Gamma|^{\frac{4}{\delta^2}}\cdot k^{2+\frac{8}{\delta^2} \log |\Gamma|})$ time.
\end{lemma}

\begin{proof}
Note that $D^* < \frac{4}{\delta^2} \left(2\log k+1\right)$ implies that the total number of possible median strings is bounded by 
\begin{align*}
|\Gamma|^{D^*} \le |\Gamma|^{\left(\frac{4}{\delta^2} \left(2\log k+1\right)\right)} = |\Gamma|^{\frac{4}{\delta^2}}\cdot k^{\frac{8}{\delta^2} \log |\Gamma|}
\end{align*} 
Note that the sets $\Gamma_i$ can be computed in $O(nd\log \min\left(n,|\Gamma|\right))$ time in the worst case by using $O(nd)$ memory and we can enumerate all candidate median strings in $O\left(|\Gamma|^{\frac{4}{\delta^2}}\cdot k^{\frac{8}{\delta^2} \log |\Gamma|} \right)$ time. Since Hamming distance satisfies triangle inequality, the greedy algorithm of~\cite{Ravi1994Heuristic} guarantees a set of $k$ strings with minimum dispersion at least $\frac{1}{2}t^*$, and its runtime is $O\left(k \cdot k \cdot |\Gamma|^{\frac{4}{\delta^2}}\cdot k^{\frac{8}{\delta^2} \log |\Gamma|}\right) = O \left(|\Gamma|^{\frac{4}{\delta^2}}\cdot k^{2+\frac{8}{\delta^2} \log |\Gamma|}\right)$. Thus, the overall time complexity in this case is $O\left(nd\log \min\left(n,|\Gamma|\right)+|\Gamma|^{\frac{4}{\delta^2}}\cdot k^{2+\frac{8}{\delta^2} \log |\Gamma|}\right)$.
\end{proof}

Finally, given~\autoref{lem:min-dispersion-DP},~\autoref{lem:uniform-exact}, and~\autoref{lem:exact-small}, we will now prove the~\autoref{thm:exact-kon}.

\begin{proof}[Proof of~\autoref{thm:exact-kon}] 

Let $k \le \frac{1}{\delta}$. Then,~\autoref{lem:min-dispersion-DP} directly implies we can calculate the exact solution to the min dispersion problem in $O(|\Gamma|^{\frac{1}{\delta}} d^{\frac{1}{2\delta^2}})$ time.

Now consider the case when $D^* < \frac{4}{\delta^2} \left(2\log k+1\right)$ and $k > \frac{1}{\delta}$. Then~\autoref{lem:exact-small} results in an algorithm with the desired guarantees.

Finally, when $D^* \ge \frac{4}{\delta^2} \left(2\log k+1\right)$ and $k > \frac{1}{\delta}$, from~\autoref{lem:uniform-exact}, we get a set of $k$ Hamming median strings $S$ such that $\textsc{minDp}(S) \ge (1-\delta) \sum_{i\in [d]} \frac{|\Gamma_{i}|-1}{|\Gamma_{i}|}$. Let $t^*$ be the optimal min dispersion. If $t^* \le \sum_{i\in [d]} \frac{|\Gamma_{i}|-1}{|\Gamma_{i}|}$ we are done. Now assume $t^* > \sum_{i\in [d]} \frac{|\Gamma_{i}|-1}{|\Gamma_{i}|}$. Then, the Generalized Plotkin Bound (\autoref{lem:plotkin-thm}) ensures that 
\[k \le \frac{t^*}{t^*-\sum_{i\in [d]} \frac{|\Gamma_{i}|-1}{|\Gamma_{i}|}}.\] 
which implies, 
\[t^* \le \left(1+\frac{1}{k-1}\right) \sum_{i\in [d]} \frac{|\Gamma_{i}|-1}{|\Gamma_{i}|}.\] 
Therefore,
\[\textsc{minDp}(S) \ge \frac{1-\delta}{1+\frac{1}{k-1}} t^* = \left(1-\frac{1}{k}\right)(1-\delta)t^* \ge (1-2\delta)t^*.\] 
Hence this yields a $(1-2\delta)$-approximation in $O\left(nd\log \min\left(n,|\Gamma|\right)+kd|\Gamma|+k^2d\log \frac{1}{\eta}\right)$ time.
\end{proof}

\section{Bi-criteria Approximation for Min Dispersion: \texorpdfstring{$k$}{k} Approximate Hamming Medians} 
\label{app:min-dispersion-approx}

\mindisapproxk*

In the rest of this section, we will establish~\autoref{thm:mindis-approx-k}. First, we derive a dynamic programming–based algorithm that exactly solves the minimum dispersion problem for approximate medians. Formally,

\begin{lemma}\label{lem:min-dispersion-approx-DP}
Given a set of strings $X$ and a parameter $k$, there exists an algorithm that finds the min dispersion for $k$ $(1+\epsilon)$-approximate medians in $O((1+\epsilon)^k|\Gamma|^k n^k d^{\frac{k(k+1)}{2}+1})$ time.
\end{lemma}
\begin{proof}
We first define the following dynamic program: 

Define a dynamic program $\textsc{MinDispersion-DP}[d_{1,2}, d_{1,3}, \dots, d_{k-1,k}, c_1,\dots,c_k, \ell]$ such that,
\[
\textsc{MinDispersion-DP}[d_{1,2}, d_{1,3}, \dots, d_{k-1,k}, c_1,\dots,c_k, \ell] = \text{True}
\]
if there exist $k$ strings $s_1,\dots,s_k \in \Gamma^{\ell}$ such that $H(s_i,s_j)=d_{i,j}$ for all pairs and $\sum_{x \in X} H(x,s_i) = c_i$. Note that all $c_i$ values are non-negative integers and bounded by $(1+\epsilon)\opt$. The DP state space has size $d^{\frac{k(k-1)}{2}} \cdot d ((1+\epsilon)nd)^k$ (since each $d_{i,j}$ can take values up to $d$ and since $\ell$ is bounded by $d$ and since $c_i \le (1+\epsilon)\opt \le (1+\epsilon)nd$). Suppose we want to compute 
\[
\textsc{MinDispersion-DP}(d_{1,2},\dots,d_{i,j},\dots,d_{k-1,k},c_1,\dots,c_k, \ell+1).
\]
For this to be true, there must exist a set of $k$ strings $\{s_1,\ldots,s_k\}$ in $\Gamma^{\ell}$ and an assignment of characters $\{a_1,\dots,a_k\}$ ($a_i \in \Gamma$ for all $i\in [k]$) such that
\[
\textsc{MinDispersion-DP}(d'_{1,2} , \dots, d'_{i,j}, \dots, d'_{k-1,k},c'_1,\dots,c'_k, \ell)
\]
is marked True, where $d'_{i,j} = d_{i,j}- \ind{a_i \neq a_j}$ and $c'_i = c_i-(n-|\{x\in X|x_{\ell+1}=a_i\}|)$. Since there are $|\Gamma|^k$ possible assignments $\{a_1,\dots,a_k\}$, each state update requires at most $|\Gamma|^k$ operations giving overall runtime $O(|\Gamma|^k d^{\frac{k(k-1)}{2}} \cdot d ((1+\epsilon)nd)^k)$. To extract the solution, we simply check all states of the form $(d_{1,2},\ldots,d_{k-1,k},c_1,\dots,c_k,d)$, which takes $O(d^{\frac{k(k-1)}{2}} \cdot ((1+\epsilon)nd)^k)$ time. 
\end{proof}

Next, we consider the scenario in which $D^*$ is small (constant), and demonstrate that a solution achieving a $1/2$-approximation to the minimum dispersion can be obtained.

\begin{lemma}\label{lem:approx-small}
Let $D^* \le \frac{4}{\delta^2}$. Then there exists an algorithm that outputs a set of $k$ $(1+\epsilon)$ approximate Hamming median strings $S$ such that $\textsc{minDp}(S) \ge \frac{1}{2}t^*$ in $O \left(k^2 |\Gamma|^{\frac{4}{\delta^2}}d^{\frac{4}{\delta^2}}+nd \cdot |\Gamma|^{\frac{4}{\delta^2}}d^{\frac{4}{\delta^2}}\right)$ time.
\end{lemma}
\begin{proof}
Now consider the case when $D^* \le \frac{4}{\delta^2}$. We can see that in this case, for any $(1+\epsilon)$ approximate median string $s$, it can differ from $w$ in at most $D^*$ indices. Therefore, we can see that there can be at most $N = |\Gamma|^{\frac{4}{\delta^2}}d^{\frac{4}{\delta^2}}$ such strings. Let $\hat{X}$ be the potential $(1+\epsilon)$ approximate medians (we know that $|\hat{X}| \le N$). We can calculate $\hat{X}$ in $O(nd \cdot |\Gamma|^{\frac{4}{\delta^2}}d^{\frac{4}{\delta^2}})$ time.  Since Hamming distance satisfies triangle inequality, the greedy algorithm of~\cite{Ravi1994Heuristic} guarantees a set of $k$ strings with minimum dispersion at least $\frac{1}{2}t^*$, and its runtime is $O\left(k \cdot k \cdot |\Gamma|^{\frac{4}{\delta^2}}d^{\frac{4}{\delta^2}}+nd \cdot |\Gamma|^{\frac{4}{\delta^2}}d^{\frac{4}{\delta^2}}\right) = O \left(k^2 |\Gamma|^{\frac{4}{\delta^2}}d^{\frac{4}{\delta^2}}+nd \cdot |\Gamma|^{\frac{4}{\delta^2}}d^{\frac{4}{\delta^2}}\right)$ time.
\end{proof}

Next, we consider the case where $D^*$ is large enough, and show that one can obtain a set of $k$ $(1+2\epsilon)$-approximate medians with $\frac{(1-\delta)}{2}t^*$ minimum dispersion.

\begin{lemma}\label{lem:uniform-approx}
Let $D^* \ge \frac{4}{\delta^2} \left(2\log k+1\right)$ where $\delta,\eta>0$. Then there exists an algorithm that, with probability at least $1-\eta$, outputs a set of $k$ $(1+2\epsilon)$-approximate Hamming median strings $S$ such that $\textsc{minDp}(S) \ge (1-\delta)\frac{D^*}{2}$.
\end{lemma}

\begin{proof}
Let $x^*,y^*$ be two $(1+\epsilon)$-approximate Hamming medians such that $D^* = H(x^*,y^*)$ such that $T_{x^*}=\{i \in [d]|x^*_i \neq w_i\}$ and $T_{y^*}=\{i \in [d]|y^*_i \neq w_i\}$ are disjoint (which can be calculated using the algorithm introduced in~\autoref{app:diam-approx}). Let $T=T_{x^*} \cup T_{y^*}$. Note that $|T|=D^*$.

Consider the following set of strings $S$ such that $|S| \le k$: For any $s \in S$, for any $i \in T_{x^*}$, $s_i = x^*_i$ with probability $\frac{1}{2}$ and $s_i=w_i$ with probability $\frac{1}{2}$, for any $i \in T_{y^*}$, $s_i = y^*_i$ with probability $\frac{1}{2}$ and $s_i=w_i$ with probability $\frac{1}{2}$ and for any $i \not \in T$, $s_i = w_i$. 

Note that for any $s \in S$, $s$ is an $(1+2\epsilon)$-approximate Hamming median. This comes from a direct application of~\autoref{lem:opt-offset} since $\sum_{x \in X} H(x,s)=\opt+\sum_{i \in T} (f_i^w-f_i^s) = \opt+\sum_{i \in T_{x^*}} (f_i^w-f_i^s)+\sum_{i \in T{y^*}} (f_i^w-f_i^s) \le \opt+\epsilon \opt+\epsilon \opt = (1+2\epsilon)\opt$. 

Consider any two strings $\hat{s},\tilde{s} \in S$. Note that,
\begin{align*}
H\left(\hat{s},\tilde{s}\right) &= \sum_{i \in [d]} \ind{\hat{s}_i \neq \tilde{s}_i}\\
&= \sum_{i \in T} \ind{\hat{s}_i \neq \tilde{s}_i}+ \sum_{i \not \in T} \ind{\hat{s}_i \neq \tilde{s}_i}\\
&= \sum_{i \in T} \ind{\hat{s}_i \neq \tilde{s}_i} \quad \text{(Since $\hat{s}_i=\tilde{s}_i=w_i$ for all $i\not\in T$)}
\end{align*}
For all $i \in T$, let $z_i=1$ if $\hat{s}_i \neq \tilde{s}_i$ and $0$ otherwise. Then, $Pr(z_i = 1) = \frac{1}{2}$. Furthermore, $H\left(\hat{s},\tilde{s}\right) = \sum_{i\in T} z_i$. Let $\mu = \exv\left(H\left(\hat{s},\tilde{s}\right)\right)$. Then, $\mu = \exv \left(  \sum_{i\in T} z_i \right) = \sum_{i \in T} \exv \left(z_i\right) = \frac{|T|}{2} = \frac{D^*}{2}$. Then, 
\begin{align*}
Pr\left(H\left(\hat{s},\tilde{s}\right) \le (1-\delta)\mu\right) &\le e^{-\frac{\delta^2 \mu}{2}} \quad \text{(by Chernoff bounds)}\\
& \le e^{-\frac{\delta^2 |T|}{4}} \quad \text{(Since $\mu \ge \frac{|T|}{2}$)} \\
& = \frac{1}{2k^2}
\end{align*}
Note that since $|S|=k$, there are at most $k^2$ pairs of strings $\hat{s},\tilde{s}$ in $S$. Therefore, by union bound,
\begin{align*}
Pr\left(\exists \hat{s},\tilde{s} \in S \text{ s.t. } H\left(\hat{s},\tilde{s}\right) \le (1-\delta)\mu\right) \le k^2 \frac{1}{2k^2} = \frac{1}{2}
\end{align*}

We can perform this process $N=\log \frac{1}{\eta}$ times and then select the output that maximizes the min dispersion. Let $\{S^{(1)},S^{(2)},\dots,S^{(N)}\}$ denote the $N$ outputs obtained by repeating the sampling process $N$ times, and let $\hat{S} \in \{S^{(1)},S^{(2)},\dots,S^{(N)}\}$ be the output achieving the maximum min dispersion. Then, we have
\begin{align*}
Pr\left(\forall i \in [N] \exists \hat{s},\tilde{s} \in S^{(i)} \text{ s.t. } H\left(\hat{s},\tilde{s}\right) \le (1-\delta)\mu\right) \le \frac{1}{2^N} = \eta
\end{align*}
and thus, with probability at least $1-\eta$, there exists at least one $S^{(i)}$ such that $\forall \hat{s},\tilde{s} \in S^{(i)}$, $H\left(\hat{s},\tilde{s}\right) \ge (1-\delta)\frac{D^*}{2}$. By definition, since $\hat{S}$ is the output with the maximum min dispersion, it follows that $\forall \hat{s},\tilde{s} \in \hat{S}$, $H\left(\hat{s},\tilde{s}\right) \ge (1-\delta)\frac{D^*}{2}$, and therefore $\hat{S}$ achieves the desired bound (with probability at least $1-\eta$).

From~\autoref{thm:diverse-fin}, we get that finding $x^*, y^*$ takes $O \left((1+\epsilon)nd+d\log d\right)$ time.  Since the algorithm simply involves generating $k$ randomly sampled strings, generating $k$ strings takes $O(kd)$ time. And running the sampling process $\log \frac{1}{\eta}$ times and calculating min dispersion for each solution takes $O(kd \log \frac{1}{\eta}+k^2d \log \frac{1}{\eta})$. Therefore, overall time complexity is $O \left((1+\epsilon)nd+d\log d+k^2d\log \frac{1}{\eta}\right)$. 
\end{proof}

Uniform sampling yields a bi-criteria approximation for the min dispersion problem, but incurs an additional cost of $\epsilon \opt$. A more desirable outcome is for this overhead to depend only on $\delta$, not on $\epsilon$. We show that when $t^*$ is sufficiently large, a bi-criteria approximation exists whose additional cost depends solely on $\delta$, while achieving a $\frac{1-\delta}{2}$-approximation to the minimum dispersion.
\begin{lemma}\label{thm:diverse-fin-k}
If $t^* \ge \frac{8+4\delta}{\delta} \sqrt{d} \left(2\log k+2\right)$, where $\delta>0$, there exists a randomized algorithm that, given any $X \subseteq \Gamma^d$ and parameter $\epsilon \ge 0$, with probability at least $1-\eta$, outputs a set of $k$ distinct $(1+\epsilon+\delta)$-median strings with min dispersion at least $ \frac{1-\delta}{2} t^*$, and runs in $O\left(nd\log \min\left(n,|\Gamma|\right)+(k^9 \cdot d^3+nkd+k^2 d)\log  \frac{1}{\eta} \right)$ time.
\end{lemma}
\begin{proof} 
We start by defining an integer linear program for the $\mdehm$ problem. Then, we consider its linear programming relaxation and derive a fractional solution. Using dependent rounding ideas from~\cite{Gandhi2006Dependent}, we round this fractional solution to an integral solution that gives a $(1-\delta)/2$-approximation for the min dispersion.

Let us first describe the ILP for the $\mdehm$ problem. Given $X \subseteq \Gamma^d$ and a string $w = \mcc(X)$, we define auxiliary strings ${\hat{w}^{(j)} \mid j \in [k]}$. For each index $i$, $\hat{w}_i^{(j)}$ is the $j$th most frequent character at position $i$ among the strings in $X$. More precisely: if $\freq{i}{w} < n$, then set $\hat{w}^{(j)}_i = \arg\max_{e \in \Gamma \setminus \{w_i,\hat{w}^{(2)}_i,\dots,\hat{w}^{(j-1)}_i\}} |{x \in X : x_i = e}|$, breaking ties arbitrarily. Otherwise, set $\hat{w}^{(j)}_i = e$ for some arbitrary $e \in \Gamma \setminus {w_i}$. We also define a weight parameter $c_{ij} = f_i^w - f_i^{\hat{w}^{(j)}}$.

Let $r, \hat{r} \in [k]$ be the indices of strings $s_r$ and $s_{\hat{r}}$. Define $u_{rij}$ as a variable such that $u_{rij} = 1$ if and only if string $s_r$ has $\hat{w}^{(j)}_i$ at position $i$. Similarly, define $z_{r\hat{r}ij}$ as a variable such that $z_{r\hat{r}ij} = 1$ if and only if exactly one of $s_r$ or $s_{\hat{r}}$ has $\hat{w}^{(j)}_i$ at position $i$.

The ILP constraints are as follows: Constraints~\ref{eq:opt-ub} and~\ref{eq:opt-lb} capture the cost of the solution and ensure that any solution string remains an approximate median. Constraint~\ref{eq:x-sum} ensures that for each string $r$ and index $i$, only one character is assigned at position $i$. Constraints~\ref{eq:z-sum}, \ref{eq:z-lb-r-rp}, \ref{eq:z-lb-rp-r}, and \ref{eq:z-minus} capture the contribution to the min dispersion from the selected strings, while Constraint~\ref{eq:z-div} ensures that the ILP maximizes the min dispersion.

\begin{align}
&\text{Maximize } t \nonumber\\
&\text{subject to the constraints} \nonumber\\
&\sum_{(i,j)} u_{rij} c_{ij} \le (1+\epsilon)\opt \; \forall r\ \in [k]\label{eq:opt-ub}\\
&\sum_{(i,j)} u_{rij} c_{ij} \ge \opt \; \forall r\ \in [k]\label{eq:opt-lb}\\
& \sum_{j=1}^k u_{rij} = 1 \; \forall r \in [k],i \in [d]\label{eq:x-sum}\\
& z_{r\hat{r}ij} \le u_{rij}+ u_{\hat{r}ij} \; \forall r,\hat{r},j \in [k],i \in [d]\label{eq:z-sum}\\
& z_{r\hat{r}ij} \ge u_{\hat{r}ij} - u_{rij} \; \forall r,\hat{r},j \in [k],i \in [d] \label{eq:z-lb-r-rp}\\
& z_{r\hat{r}ij} \ge u_{rij} - u_{\hat{r}ij} \; \forall r,\hat{r},j \in [k],i \in [d] \label{eq:z-lb-rp-r}\\
& z_{r\hat{r}ij} \le 2-u_{rij} - u_{\hat{r}ij} \; \forall r,\hat{r},j \in [k],i \in [d] \label{eq:z-minus}\\
& \sum_{(i,j)} z_{r\hat{r}ij} \ge 2t \; \forall  r,\hat{r} \in [k] \nonumber\\
& u_{rij}, z_{r\hat{r}ij} \in \{0,1\} \; \forall r,\hat{r},j \in [k],i \in [d] \label{eq:z-div}
\end{align}
where $t$ is a variable that captures the min dispersion. Let us call the above ILP as $\ilpApp$. 

Consider a solution to $\ilpApp$ is denoted by $u_{rij}$, where $r,j \in [k], i \in [d]$.
Let $\{s_1,s_2,\dots,s_k\}$ be the corresponding $k$ strings defined as follows. 
\[{s_r}_i = \hat{w}^{(j)}_i \text{ where } u_{rij}=1 \; \forall r,j \in [k], i \in [d]\] where ${s_r}_i$ is the character at the index $i$ of the string $s_r$. 
Recall, by constraint~\ref{eq:x-sum}, we have $\sum_{j=1}^k u_{rij} = 1$ implies that there is exactly one non-zero  $u_{rij}$. Given $\ilpApp$ and the strings corresponding to the $\{u_{rij}|i \in [d], j\in [k]\}$ for $r \in [k]$,i.e. $\{s_1,s_2,\dots,s_k\}$, we can see that each string $s_r$ is a $(1+\epsilon)$-approximation median string and the solution maximizes the min dispersion. Formally,
\begin{claim}\label{clm:ilp-opt}
Let $\{s_1,s_2,\dots,s_k\}$ be the $k$-strings derived using the $\ilpApp$, then,
\begin{enumerate}[(I).]
\item $\opt \le \sum_{x \in X} H(x,s_r) \le (1+\epsilon) \opt$ for all $r \in [k]$, and 
\item $\min_{r,\hat{r}} H(s_r,s_{\hat{r}}) = t^*$
\end{enumerate}
\end{claim}

We now consider the relaxation of $\ilpApp$ by letting $u_{rij} \in [0,1]$ and $z_{r\hat{r}ij} \in [0,1]$. We can show that using the solution to the relaxed linear program, and rounding the linear programming solution gives us a simple $(1-\delta)/2$-approximation to the solution of the $\ilpApp$ and therefore a $(1-\delta)/2$-approximation to the min dispersion.

Let $\{\tilde u_{rij}| r,j \in [k], i \in [d]\}$ be the solution to the relaxed LP. For each $r \in [k]$, we can see that $\{\tilde u_{rij}| j \in [k], i \in [d]\}$ forms a set of weights for a bipartite graph $(A,B,E)$ where $A = [d], B= [k]$ and $E = \{(i,j) \in [d] \times [k]\}$. Note that $\forall i\in [d]$, $\sum_{j=1}^k \tilde u_{rij} = 1$. 

Given that $\{\tilde u_{rij}| j \in [k], i \in [d]\}$ forms a bipartite graph $(A,B,E)$, the work of~\cite{Gandhi2006Dependent} gives a dependent rounding framework, that returns a randomized rounded solution $\{\hat{u}_{rij}| j \in [k], i \in [d]\}$ such that $\hat{u}_{rij} \in \{0,1\} \; \forall i \in [d],j\in [k]$,$\sum_{j=1}^k \hat{u}_{rij}=1 \; \forall i \in [d]$ and $Pr\left(\hat{u}_{rij}\right) = \tilde u_{rij}$ and runs in time $O\left((|A|+|B|)|E|\right)=O\left((d+k)dk\right)$. Therefore for each $r \in [k]$ we can get a rounded integer solution $\{\hat{u}_{rij}| j \in [k], i \in [d]\}$ such that $\sum_{j=1}^k \hat{u}_{rij}=1 \; \forall i \in [d]$ and $Pr\left(\hat{u}_{rij}\right) = \tilde u_{rij}$ in $O\left(d^2k^2+dk^3\right)$ time. 

We will show that the rounded solution $\{\hat{u}_{rij}| r,j \in [k], i \in [d]\}$ gives a $(1-\delta)/2$-approximation to the maximum diversity. Given $\{\hat{u}_{rij}| r,j \in [k], i \in [d]\}$ we first calculate the corresponding $\hat{z}_{r\hat{r}ij}$ values using the constraints~\ref{eq:z-sum}, \ref{eq:z-lb-r-rp}, \ref{eq:z-lb-rp-r} and~\ref{eq:z-minus}.

$\forall r,\hat{r},i,j$, $Pr(\hat{z}_{r\hat{r}ij} = 1)$ after rounding is the probability of $\left(\hat{u}_{rij}=1 \text{ and } \hat{u}_{\hat{r}ij}=0\right)$ or $\left(\hat{u}_{rij}=0 \text{ and } \hat{u}_{\hat{r}ij}=1\right)$ and the probability of this is $(1-\tilde u_{rij})\tilde{u}_{\hat{r}ij}+(1-\tilde{u}_{\hat{r}ij})\tilde{u}_{rij} = \tilde{u}_{rij}+\tilde{u}_{\hat{r}ij}-2\tilde{u}_{rij}\tilde{u}_{\hat{r}ij}$. Therefore,
\[E\left(\hat{z}_{r\hat{r}ij}\right) = \tilde{u}_{rij}+\tilde{u}_{\hat{r}ij}-2\tilde{u}_{rij}\tilde{u}_{\hat{r}ij}\]
Using this, we can show that $E\left(\hat{z}_{r\hat{r}ij}\right) \ge \frac{1}{2}\tilde{z}_{r\hat{r}ij}$ where $\tilde{z}_{r\hat{r}ij}$ is the corresponding value from the relaxed solution. In order to do this, we will consider two cases.

\noindent {\bf Case 1:} When $\tilde{u}_{rij}+\tilde{u}_{\hat{r}ij} \le 1$.

Note that $\tilde{u}_{rij}+\tilde{u}_{\hat{r}ij} \le 1 \implies \left(\tilde{u}_{rij}+\tilde{u}_{\hat{r}ij}\right)^2 \le \tilde{u}_{rij}+\tilde{u}_{\hat{r}ij}$. A simple calculation shows us,
\begin{align*}
E\left(\hat{z}_{r\hat{r}ij}\right) &= \tilde{u}_{rij}+\tilde{u}_{\hat{r}ij}-2\tilde{u}_{rij}\tilde{u}_{\hat{r}ij}\\
& = \frac{1}{2} \left(\tilde{u}_{rij}+\tilde{u}_{\hat{r}ij}\right)+\frac{1}{2} \left(\tilde{u}_{rij}+\tilde{u}_{\hat{r}ij}-4\tilde{u}_{rij}\tilde{u}_{\hat{r}ij}\right) \\
& \ge \frac{1}{2} \left(\tilde{u}_{rij}+\tilde{u}_{\hat{r}ij}\right)+\frac{1}{2} \left(\left(\tilde{u}_{rij}+\tilde{u}_{\hat{r}ij}\right)^2-4\tilde{u}_{rij}\tilde{u}_{\hat{r}ij}\right) \\
& \ge \frac{1}{2} \left(\tilde{u}_{rij}+\tilde{u}_{\hat{r}ij}\right)+\frac{1}{2} \left(\tilde{u}_{rij}-\tilde{u}_{\hat{r}ij}\right)^2 \\
& \ge \frac{1}{2} \left(\tilde{u}_{rij}+\tilde{u}_{\hat{r}ij}\right) \ge \frac{1}{2}\tilde{z}_{r\hat{r}ij.}
\end{align*}
where the last inequality follows from the constraint~\ref{eq:z-sum}. Note that this is tight when $\tilde{u}_{rij}=\tilde{u}_{\hat{r}ij}=\frac{1}{2}$.

\noindent {\bf Case 2:} When $\tilde{u}_{rij}+\tilde{u}_{\hat{r}ij} > 1$.

We can see that by a direct application of AM-GM inequality, $\left(u_{rij}+u_{\hat{r}ij}\right)^2 \ge 4u_{rij}\cdot u_{\hat{r}ij}$. Using this, we get that,
\begin{align*}
E\left(\hat{z}_{r\hat{r}ij}\right) &= \tilde{u}_{rij}+\tilde{u}_{\hat{r}ij}-2\tilde{u}_{rij}\tilde{u}_{\hat{r}ij}\\
& = \frac{1}{2} \left(2\tilde{u}_{rij}+2\tilde{u}_{\hat{r}ij}-4\tilde{u}_{rij}\tilde{u}_{\hat{r}ij}\right) \\
& \ge \frac{1}{2} \left(2\left(\tilde{u}_{rij}+\tilde{u}_{\hat{r}ij}\right)-\left(\tilde{u}_{rij}+\tilde{u}_{\hat{r}ij}\right)^2\right) \\
& \ge \frac{1}{2} \left(\tilde{u}_{rij}+\tilde{u}_{\hat{r}ij}\right)\left(2-\left(\tilde{u}_{rij}+\tilde{u}_{\hat{r}ij}\right)\right) \\
& \ge \frac{1}{2} \left(2-\left(\tilde{u}_{rij}+\tilde{u}_{\hat{r}ij}\right)\right) \ge \frac{1}{2}\tilde{z}_{r\hat{r}ij}\\
\end{align*}

Therefore, we establish that in all cases we end up with $E\left(\hat{z}_{r\hat{r}ij}\right) \ge \frac{1}{2}\tilde{z}_{r\hat{r}ij.}$ Therefore, using linearity of expectation, for any $r,\hat{r} \in [k]$, $E \left(\sum_{i,j}\hat{z}_{r\hat{r}ij}\right) \ge \frac{1}{2} \left(\sum_{i,j}\tilde{z}_{r\hat{r}ij}\right)$.

Let $t^*$ be the solution to $\ilpApp$. Note that since $\left(\sum_{i,j}\tilde{z}_{r\hat{r}ij}\right) \ge 2t^*$ we have, $E \left(\sum_{i,j}\hat{z}_{r\hat{r}ij}\right) \ge \frac{1}{2} \left(\sum_{i,j}\tilde{z}_{r\hat{r}ij}\right) \ge \frac{1}{2} 2t^*=t^*$. Let $\hat{Z}_{r\hat{r}i} = \sum_{j=1}^k \hat{z}_{r\hat{r}ij}$ and $\hat{Z}_{r\hat{r}} = \sum_{i=1}^d \sum_{j=1}^k \hat{z}_{r\hat{r}ij} = \sum_{i=1}^d \hat{Z}_{r\hat{r}i}$. We can see that by definition $\hat{z}_{r\hat{r}ij} \le \hat{u}_{rij}+\hat{u}_{\hat{r}ij}$ and therefore, $\hat{Z}_{r\hat{r}i} \le \sum_{j=1}^k \left( \hat{u}_{rij}+\hat{u}_{\hat{r}ij}\right)$. Since from dependent rounding, $\sum_{j=1}^k \hat{u}_{rij} = 1$ and $\sum_{j=1}^k \hat{u}_{\hat{r}ij} = 1$, we get $0 \le \hat{Z}_{r\hat{r}i} \le 2$. Note that $\{\hat{Z}_{r\hat{r}i}|i\in [d]\}$ are $d$ independent variables and therefore $\hat{Z}_{r\hat{r}}$ is a sum of independent variables. Therefore, using Hoeffding's inequality, 
\begin{align*}
Pr \left(\hat{Z}_{r\hat{r}} \le \frac{2}{(2+\delta)} t^*\right) &= Pr \left(\hat{Z}_{r\hat{r}} \le t^*-\frac{\delta}{(2+\delta)} t^*\right)\\ 
& \le Pr \left(\hat{Z}_{r\hat{r}} \le E(\hat{Z}_{r\hat{r}})-\frac{\delta}{(2+\delta)} t^*\right)\\
& \le 2e^{-\frac{\frac{\delta^2}{(2+\delta)^2} {t^*}^2}{\sum_{i=1}^d 2^2}}\\
& \le 2e^{-\frac{\frac{\delta^2}{(2+\delta)^2} {t^*}^2}{4d}}.
\end{align*}

Let $t^* \ge \frac{8+4\delta}{\delta}\sqrt{d} (2 \log k +2) > \frac{4+2\delta}{\delta}\sqrt{d} (2 \log k +2)$. Then, $Pr \left(\hat{Z}_{r\hat{r}} \le \frac{2}{(2+\delta)} t^*\right) < \frac{2}{k^2} \frac{1}{4}$. Therefore,
\begin{align*}
Pr\left(\exists r,\hat{r} \text{ s.t. } \hat{Z}_{r\hat{r}} \le \frac{2}{(2+\delta)} t^*\right) &\le \sum_{r,\hat{r}}  Pr \left(\hat{Z}_{r\hat{r}} \le \frac{2}{(2+\delta)} t^*\right)\\
& < \sum_{r,\hat{r}} \frac{2}{k^2} \frac{1}{4} = \frac{k(k-1)}{2} \frac{2}{k^2} \frac{1}{4} \\
& < \frac{1}{4}.
\end{align*}
Therefore, with probability at least $3/4$ we get that given $t^* \ge \frac{8+4\delta}{\delta}\sqrt{d} (2 \log k +2)$ the rounding gives us a solution such that $\hat{Z}_{r\hat{r}} \ge \frac{2}{(2+\delta)} t^*$ for any $r,\hat r \in [k]$. Let $t$ be the objective value for the rounded solution. Then, we can see that $t = \frac{1}{2}\min_{r,\hat{r} \in [k]}\hat{Z}_{r\hat{r}}\ge \frac{1}{(2+\delta)} t^*$.

Before moving on to establishing the approximation ratio for the median objective, we will first establish the following relationship between the optimal value of the optimization objective and the min dispersion value.
\begin{claim}\label{clm:opt-hamming}
Let $t^*$ be the optimal min dispersion and let $x^*,y^*$ be two strings such that, $\sum_{x \in X} H(x,x^*) \le (1+\epsilon)\opt$ and $\sum_{x \in X} H(x,y^*) \le (1+\epsilon)\opt$ and $H(x^*,y^*)=t^*$. Then,
\[t^* \le \frac{4(1+\epsilon)\opt}{n}.\]
\end{claim}

Next, we will show that the strings output using the rounded solution are $(1+\epsilon+\delta)$-approximate medians. Let $X_{ri} = \sum_{j=1}^k \hat{u}_{rij} c_{ij} \le n \sum_{j=1}^k \hat{u}_{rij} = n$ (since $c_{ij} \le n$) and $X_r = \sum_{i=1}^d \sum_{j=1}^k \hat{u}_{rij} c_{ij} = \sum_{i=1}^d X_{ri}$. Note that $(1+\epsilon) \opt \ge E(X_r) \ge \opt$. Since $\{X_{ri}|i\in[d]\}$ are independent, $X_r = \sum_{i=1}^d X_{ri}$ is a sum of independent variables. Therefore, using Hoeffding's inequality, we can see that,
\begin{align*}
Pr \left(X_r \ge (1+\epsilon+\delta) \opt\right) & \le Pr \left(X_r \ge E(X_r)+\delta \opt\right) \\
& \le e^{-\frac{\delta^2 \opt^2}{d n^2}}\\
& \le e^{-\frac{\delta^2 \frac{n^2 {t^*}^2}{16(1+\epsilon)^2}}{d n^2}} \quad\text{(By Claim~\ref{clm:opt-hamming})}\\
&= e^{-\frac{\delta^2 {t^*}^2}{d {16(1+\epsilon)^2}}}\\
& \le \frac{1}{4k^2}.
\end{align*}
Therefore, we can see that,
\begin{align*}
Pr \left(\left(X_r \le (1+\epsilon+\delta) \opt\right)\right) \ge 1-\frac{1}{4k^2}.
\end{align*}
Therefore,
\begin{align*}
& Pr \left(\forall r \in [k] \;\left(X_r \le (1+\epsilon+\delta) \opt\right)\right) \\
& = \prod_{r=1}^k Pr \left( \left(X_r \le (1+\epsilon+\delta) \opt\right)\right) \\ 
& \ge \left(1-\frac{1}{4k^2}\right)^k \\
& \ge 1-\frac{1}{4k}\\
& \ge \frac{3}{4}.
\end{align*}

Therefore, with probability $\ge \frac{1}{2}$, we have $k$ strings $\{s_1,s_2,\dots,s_k\}$ such that $\sum_{x \in X} H(x,s_r) \le (1+\epsilon+\delta) \opt$ for all $r \in [k]$ and $\min_{r,\hat{r}} H(s_r,s_{\hat{r}}) \ge \frac{1}{(2+\delta)} t^* \ge \frac{1-\delta}{2} t^*$. 

We can perform this procedure $N = \log \frac{1}{\eta}$ times and select the output that both yields $(1+\epsilon+\delta)$-approximate median strings and maximizes the minimum dispersion. Denote the $N$ outputs from these repeated samplings as $\{S^{(1)}, S^{(2)}, \dots, S^{(N)}\}$, and let $\hat{S} \in \{S^{(1)}, S^{(2)}, \dots, S^{(N)}\}$ be the set of $(1+\epsilon+\delta)$-approximate median strings with the largest minimum dispersion. Then we have
\begin{align*}
&Pr\left(\forall i \in [N], \exists \hat{s} \in S^{(i)} \text{ s.t. } \sum_{x \in X} H(x,\hat{s}) > (1+\epsilon+\delta)\opt \text{ or } \hat{s}, \tilde{s} \in S^{(i)} \text{ s.t. } H(\hat{s}, \tilde{s}) < \frac{t^*}{2+\delta}\right) \\
&\le \frac{1}{2^N} = \eta,
\end{align*}
which means that with probability at least $1-\eta$, there exists some $S^{(i)}$ such that for all $\hat{s}, \tilde{s} \in S^{(i)}$, we have $H(\hat{s}, \tilde{s}) \ge \frac{t^*}{2+\delta}$, and for all $\hat{s} \in S^{(i)}$, $\sum_{x \in X} H(x,\hat{s}) \le (1+\epsilon+\delta)\opt$. By definition, since $\hat{S}$ is the chosen set where every $\hat{s} \in \hat{S}$ satisfies $\sum_{x \in X} H(x,\hat{s}) \le (1+\epsilon+\delta)\opt$ and the minimum dispersion is maximized, it follows that for all $\hat{s}, \tilde{s} \in \hat{S}$, $H(\hat{s}, \tilde{s}) \ge \frac{t^*}{2+\delta} \ge \frac{1-\delta}{2}t^*$. Hence, $\hat{S}$ is indeed a valid solution.

Note that the $\ilpApp$ has $O(k^3d)$ variables and $O(k^3d)$ constraints. Therefore, the relaxed linear program can be solved in $O(k^9d^3)$ time. Once we have the relaxed solution, the rounding takes $O(k^3d+k^2d^2)$ time. Furthermore, calculating $c_{ij}$ values in defining the $\ilpApp$ takes $O\left(nd\log \min\left(n,|\Gamma|\right)+kd\right)$ time (to find the frequency of characters in $X$ and then to calculate $c_{ij}$ in the sorted order and find the cost for top $k$ most frequent characters). Therefore, overall time complexity for solving the $\ilpApp$ is $O(nd\log \min\left(n,|\Gamma|\right)+k^9 \cdot d^3)$. Note that since we run this process $\log \frac{1}{\eta}$ times we get $O\left(nd\log \min\left(n,|\Gamma|\right)+k^9 \cdot d^3\log  \frac{1}{\eta} \right)$ and since for each solution, we need to check the cost and the dispersion, we also have additional $O((nkd+k^2d)\log \frac{1}{\eta})$ time on finding the solution. Therefore, overall time complexity is $O\left(nd\log \min\left(n,|\Gamma|\right)+(k^9 \cdot d^3+nkd+k^2 d)\log  \frac{1}{\eta} \right)$.
\end{proof}

We include the proofs of Claim~\ref{clm:ilp-opt} and Claim~\ref{clm:opt-hamming} below for completeness.

\begin{proof}[Proof of Claim~\ref{clm:ilp-opt}]
We will first show that the solution to the $\ilpApp$ leads to a set of $(1+\epsilon)$-approximate medians. Consider a string $s_r$ where $r \in [k]$. Let $j(r,i) = j$ where $u_{rij}=1$, for some $i,j,r$. Note that 
\begin{align*}
\sum_{i=1}^d \sum_{j=1}^k u_{rij} c_{ij} &= \sum_{i=1}^d c_{ij(r,i)} \\
& = \sum_{i=1}^d (f_i^w-f_i^{\hat{w}^{j(r,i)}}) \quad \text{(Since ${s_r}_i = \hat{w}^{j(r,i)}_i$)}\\
& = \sum_{i=1}^d (f_i^w-f_i^{s_r})\\
& = \sum_{x \in X} H(x,s_r).
\end{align*}
Therefore, by constraints ~\ref{eq:opt-ub} and ~\ref{eq:opt-lb}, we have
\[\opt \le \sum_{x \in X} H(x,s_r) \le (1+\epsilon)\opt.\]

Next, we show that $\min_{r,\hat{r}} H(s_r,s_{\hat{r}}) = t^*$.  In order to do this, we first show that for any $r,\hat{r} \in [k]$, $\sum_{i=1}^d \sum_{j=1}^k z_{r\hat{r}ij} = 2H(s_r,s_{\hat{r}})$.
Consider $z_{r\hat{r}ij}$, $u_{rij}$ and $u_{\hat{r}ij}$. Constraints~\ref{eq:z-sum},\ref{eq:z-lb-r-rp},\ref{eq:z-lb-rp-r} and~\ref{eq:z-minus} imply $z_{r\hat{r}ij} = 1$ if and only if $u_{rij} = 1$ and $u_{\hat{r}ij}=0$ (or $u_{rij} = 0$ and $u_{\hat{r}ij}=1$). 

Let $i \in [d]$ be an index such that there exists $j \in [k]$ such that $z_{r\hat{r}ij} = 1$.
We claim that there exists exactly one other $\hat{j} \neq j$ such that $z_{r\hat{r}i\hat{j}} = 1$. Without loss of generality, assume $u_{rij} = 1$ and $u_{\hat{r}ij}=0$. 
Since $\sum_{j=1}^k u_{\hat{r}ij} = 1$ and $u_{\hat{r}ij}=0$ there exists exactly one $\hat{j} \neq j$ such that $u_{\hat{r}i\hat{j}}=1$.
Also, $u_{ri\hat{j}} = 0$ as $\sum_{j=1}^k u_{rij} = 1$. Hence, we obtain $z_{r\hat{r}i\hat{j}} = 1$. 
Note that $u_{rij'} = 0$ and $u_{\hat{r}ij'} = 0$, for any $j' \in [k]\setminus \{j,\hat{j}\}$, which implies that $z_{r\hat{r}ij'} = 0$. 
Therefore, $\sum_{j=1}^k z_{r\hat{r}ij}=2$. Let $T = \{i\in [d]|\exists j \text{ s.t. } z_{r\hat{r}ij} = 1\}$. Then,
\begin{align*}
\sum_{i=1}^d \sum_{j=1}^k z_{r\hat{r}ij} & = \sum_{i \in T} \sum_{j=1}^k z_{r\hat{r}ij}+ \sum_{i \in [d]\setminus T} \sum_{j=1}^k z_{r\hat{r}ij}\\
& = \sum_{i \in T} \sum_{j=1}^k z_{r\hat{r}ij} \\
& = \sum_{i \in T} 2.
\end{align*}
Let $\hat{T} = \{i \in [d]|{s_r}_i \neq {s_{\hat{r}}}_i\}$.
We claim that $T = \hat{T}$.
Let $i \in T$, then there exists $j,\hat{j}$ ($j \neq \hat{j}$) such that $u_{rij}=1$ and $u_{\hat{r}i\hat{j}}=1$.
That is, ${s_r}_i =\hat{w}^{(j)}_i \neq {s_{\hat{r}}}_i=\hat{w}^{(j)}_i$ thus $i \in \hat{T}$. 
Let $i \in \hat{T}$.
By the definition of $\hat{T}$, ${s_r}_i \neq {s_{\hat{r}}}_i$.
If ${s_r}_i = \hat{w}_i^{(j)}$ then $u_{rij} = 1$ and $u_{\hat{r}ij}=0$.
Therefore $z_{r\hat{r}ij}=1$ which implies $i \in T$. Therefore, $T = \hat{T}$.
\begin{align*}
\sum_{i=1}^d \sum_{j=1}^k z_{r\hat{r}ij} &= \sum_{i \in T} 2\\
&= \sum_{i \in \hat{T}} 2 \\
&= 2H(s_r,s_{\hat{r}}).
\end{align*}
Therefore, for all $r ,\hat{r} \in [k]$ $H(s_r,s_{\hat{r}}) = \frac{1}{2}\sum_{i=1}^d \sum_{j=1}^k z_{r\hat{r}ij} \ge t$ and since the objective is to maximize $t$, this implies $t = \min_{r,\hat{r} \in [k]} H(s_r,s_{\hat{r}})$. 

Assume there exist a set of $k$ strings $\{\hat{s}_1,\hat{s}_2,\dots,\hat{s}_k\}$ such that they are are $(1+\epsilon)$-approximation medians and $\min_{r,\hat{r} \in [k]} H(\hat{s}_r,\hat{s}_{\hat{r}}) = t^* > t$. Let $\{\hat{u}_{rij}|r,j \in [k], i \in [d]\}$ be a set of variables such that $\hat{u}_{rij} = 1$ if $\hat{s}_r$ has character $w^{(j)}_i$ in its $i$th index (and $0$ otherwise). Note that since the strings are $(1+\epsilon)$-approximation medians, $\sum_{(i,j)} \hat{u}_{rij} c_{ij} \le (1+\epsilon)\opt \; \forall r\ \in [k]$
and $\sum_{(i,j)} \hat{u}_{rij} c_{ij} \ge \opt \; \forall r\ \in [k]$. Thus, there exists is a feasible solution using $\{\hat{u}_{rij}|r,j \in [k], i \in [d]\}$, such that the corresponding assignmnet of $\hat{z}_{r\hat{r}ij}$ gives a min dispersion $t^*$. This contradicts the optimality of the solution for $\ilpApp$. Hence, $\min_{r,\hat{r} \in [k]} H(s_r,s_{\hat{r}}) \ge t^*$. Since $t^*$ is the maximum possible min dispersion, $\min_{r,\hat{r} \in [k]} H(s_r,s_{\hat{r}}) = t^*$ 
\end{proof}

\begin{proof}[Proof of Claim~\ref{clm:opt-hamming}]
Consider any string $s$ such that, $\sum_{x \in X} H(x,s) \le (1+\epsilon)\opt$. Let $\ell$ be the number of indices where $s_i \neq w_i$ ($i \in [d]$) and let $S_{\ell}$ be the set of such indices. We can see that by definition,
\[\sum_{x \in X} H(x,s) = \sum_{i=1}^d \left(n-f^s_i\right).\]
Note that for any $i \in [d]$ where $s_i \neq w_i$, $f^s_i \le \frac{n}{2}$ (becuase $w_i$ is the most frequent character so $f^w_i \ge f^s_i$ and also $f^w_i+f^s_i \le n$). Therefore, we can see that,
\begin{align*}
\sum_{x \in X} H(x,s) &= \sum_{i=1}^d \left(n-f^s_i\right)\\
& \ge  \sum_{i \in S_{\ell}} \left(n-f^s_i\right)\\
& \ge  \sum_{i \in S_{\ell}} \frac{n}{2} = \frac{n \cdot \ell}{2}.
\end{align*}
Therefore, we get that,
\[(1+\epsilon)\opt \ge \frac{n \cdot \ell}{2} \implies \frac{2(1+\epsilon)\opt}{n} \ge \ell.\]

Consider any two strings $x^*,y^*$ such that, $\sum_{x \in X} H(x,x^*) \le (1+\epsilon)\opt$ and $\sum_{x \in X} H(x,y^*) \le (1+\epsilon)\opt$ and $H(x^*,y^*)=t^*$. Let $S_{x^*} = \{x^*_i \neq w_i|i \in [d]\}$ and $S_{y^*} = \{y^*_i \neq w_i|i \in [d]\}$. We can see that,
\begin{align*}
H(x^*,y^*) &= \sum_{i=1}^d \ind{x^*_i \neq y^*_i} \\
& = \sum_{i \in S_{x^*}\setminus S_{y^*}} 1 + \sum_{i \in S_{y^*}\setminus S_{x^*}} 1 +\sum_{i \in S_{x^*}\cap S_{y^*}} \ind{x^*_i \neq y^*_i} \\
& \le \sum_{i \in S_{x^*}} 1 + \sum_{i \in S_{y^*}} 1= |S_{x^*}|+|S_{y^*}|\\
& \le 2\ell \\
& \le \frac{4(1+\epsilon)\opt}{n}.
\end{align*}
\end{proof}

Finally, combining Lemmas~\ref{lem:min-dispersion-approx-DP},~\ref{lem:approx-small}, \ref{lem:uniform-approx}, and \ref{thm:diverse-fin-k}, we prove~\autoref{thm:mindis-approx-k}.

\begin{proof}[Proof of~\autoref{thm:mindis-approx-k}]
Let $k \le \frac{1}{\delta}$. Then,~\autoref{lem:min-dispersion-approx-DP} directly implies we can calculate the exact solution to the min dispersion problem in $O((1+\epsilon)^{\frac{1}{\delta}}|\Gamma|^\frac{1}{\delta} n^{\frac{1}{\delta}} d^{\frac{2}{\delta^2}})$ time. For the rest of the proof, we consider the case when $k > \frac{1}{\delta}$.

Note that when $D^* \le \frac{4}{\delta^2}$ as a direct implication of~\autoref{lem:approx-small}, we get a solution $S$ such that $\textsc{minDp}(S) \ge \frac{1}{2}t^*$.

Let $D^* \ge \frac{4}{\delta^2} \left(2\log k+1\right)$. ~\autoref{lem:uniform-approx} gives a set $S$ of $k$ median strings such that $\textsc{minDp}(S) \ge \left(1-\delta\right) \frac{D^*}{2}$. Note that $D^* \ge t^*$ (otherwise there exist two strings such that the distance is $t^*$ and therefore $D^*$ would not be the diameter). Therefore, this implies,  $\textsc{minDp}(S) \ge \left(1-\delta\right) \frac{t^*}{2}$. 

The case for $t^* \ge \frac{8+4\delta}{\delta} \sqrt{d} \left(2\log k + 2\right)$ follows directly from~\autoref{thm:diverse-fin-k}.
\end{proof}

\section{Generalized Plotkin Bound}
\label{sec:plotkin}

In this section, we prove a generalized version of the Plotkin bound~\cite{guruswami2012essential}. 

Consider the following setting: Let 
$\Gamma_1,\Gamma_2,\dots,\Gamma_d$ be a set of $d$ alphabets and let $C \subseteq \left[|\Gamma_1|\right] \times \left[|\Gamma_2|\right] \times \dots \times \left[|\Gamma_d|\right]$. In the following lemma, we extend the standard Plotkin bound to the new setting.

\begin{lemma}[Generalized Plotkin Bound]\label{lem:plotkin-thm}
Let $\Gamma_1,\Gamma_2,\dots,\Gamma_d$ be a set of $d$ alphabets and let $C \subseteq \left[|\Gamma_1|\right] \times \left[|\Gamma_2|\right] \times \dots \times \left[|\Gamma_d|\right]$. Let $t = \min_{c_1,c_2 \in C} H(c_1,c_2)$, then:
\begin{enumerate}
\item If $t = \sum_{\ell=1}^d \frac{|\Gamma_{\ell}|-1}{|\Gamma_{\ell}|}$, then $|C| \le 2\sum_{\ell=1}^d |\Gamma_{\ell}|$;
\item If $t > \sum_{\ell=1}^d \frac{|\Gamma_{\ell}|-1}{|\Gamma_{\ell}|}$, then $|C| \le \frac{t}{t-\sum_{\ell=1}^d \frac{|\Gamma_{\ell}|-1}{|\Gamma_{\ell}|}}$.
\end{enumerate}
\end{lemma}

In order to prove~\autoref{lem:plotkin-thm}, we follow a similar analysis to the proof of the Plotkin bound in~\cite{guruswami2012essential}. We will first restate the geometric lemma from Guruswami et. al.~\cite{guruswami2012essential}.

\begin{lemma} [Geometric Lemma from~\cite{guruswami2012essential}]\label{lem:geometric}
Let $v_1,v_2,\dots,v_m \in \ral^N$ be non-zero vectors.
\begin{enumerate}
\item If $\langle v_i,v_j \rangle \le 0$ for all $i \neq j$, then $m \le 2N$;
\item Let $v_i$ be unit vectors for $1\le i \le m$. Further, if $\langle v_i,v_j \rangle \le -\epsilon <0$ for all $i \neq j$, then $m \le 1+\frac{1}{\epsilon}$.
\end{enumerate}
\end{lemma}

Next, we establish a modified mapping lemma, which extends the mapping lemma from Guruswami et. al.~\cite{guruswami2012essential} to our extended setting.

\begin{lemma} [Modified Mapping Lemma]\label{lem:map-lem}
For every $\left[|\Gamma_1|\right] \times \left[|\Gamma_2|\right] \times \dots \times \left[|\Gamma_d|\right]$, there exists a function $f: \left[|\Gamma_1|\right] \times \left[|\Gamma_2|\right] \times \dots \times \left[|\Gamma_d|\right] \rightarrow \ral^{\sum_{\ell=1}^d |\Gamma_{\ell}|}$ such that for every $c_1,c_2 \in \left[|\Gamma_1|\right] \times \left[|\Gamma_2|\right] \times \dots \times \left[|\Gamma_d|\right]$, we have,
\begin{align*}
\langle f(c_1), f(c_2) \rangle = 1-\frac{H(c_1,c_2)}{\sum_{\ell=1}^d \frac{|\Gamma_{\ell}|-1}{|\Gamma_{\ell}|}}.
\end{align*}
Consequently, we get:
\begin{enumerate}
\item For every $c \in \left[|\Gamma_1|\right] \times \left[|\Gamma_2|\right] \times \dots \times \left[|\Gamma_d|\right]$, $\|f(c)\|=1$;
\item If $H(c_1,c_2)\ge t$ then we have 
\[\langle f(c_1), f(c_2) \rangle \le 1-\frac{t}{\sum_{\ell=1}^d \frac{|\Gamma_{\ell}|-1}{|\Gamma_{\ell}|}}.\]
\end{enumerate}
\end{lemma}

\begin{proof}
Consider any $\ell \in [d]$. We define a map $\phi^{(\ell)}:\left[|\Gamma_{\ell}|\right] \rightarrow \ral^{|\Gamma_{\ell}|}$ as follows:
Let $e^{\ell}_i$ denote the unit vector along the $i$th direction in $\ral^{|\Gamma_{\ell}|}$, i.e.,
\[e^{\ell}_i = \begin{bmatrix}0,0,\dots,\underbrace{1}_{\text{$i$th position}},\dots,0,0\end{bmatrix}\]
and let $\hat{e}^{\ell} = \frac{1}{|\Gamma_{\ell}|} \sum_{i=1}^{|\Gamma_{\ell}|} e^{\ell}_i = \begin{bmatrix} \frac{1}{|\Gamma_{\ell}|},\frac{1}{|\Gamma_{\ell}|},\dots,\frac{1}{|\Gamma_{\ell}|},\frac{1}{|\Gamma_{\ell}|}\end{bmatrix}$. Note that $\langle \hat{e}^{\ell} , e^{\ell}_i\rangle = \frac{1}{|\Gamma_{\ell}|}$ for all $i \in \left[|\Gamma_{\ell}|\right]$ and $\langle \hat{e}^{\ell} , \hat{e}^{\ell}\rangle = \frac{1}{|\Gamma_{\ell}|}$. Also, note that for all $i,j \in \left[|\Gamma_{\ell}|\right]$, $\langle e^{\ell}_j , e^{\ell}_i\rangle = 1$ if $i =j$ and $0$ otherwise. 

We define $\phi^{(\ell)}$ to be $\phi^{(\ell)}(i) = \hat{e}^{\ell}-e^{\ell}_i$. For any $i,j \in \left[|\Gamma_{\ell}|\right]$,
\begin{align*}
\langle \phi^{(\ell)}(j), \phi^{(\ell)}(i)\rangle &= \langle \hat{e}^{\ell}-e^{\ell}_j, \hat{e}^{\ell}-e^{\ell}_i\rangle \\
&= \langle \hat{e}^{\ell},\hat{e}^{\ell}\rangle-\langle \hat{e}^{\ell},e^{\ell}_j\rangle-\langle e^{\ell}_j,\hat{e}^{\ell}\rangle+\langle e^{\ell}_j,e^{\ell}_i\rangle\\
&=\langle e^{\ell}_j,e^{\ell}_i\rangle-\frac{1}{|\Gamma_{\ell}|}.
\end{align*}
Therefore, for every $i \in \left[|\Gamma_{\ell}|\right]$,
\begin{align*}
\|\phi^{(\ell)}(i)\|^2=\langle e^{\ell}_i,e^{\ell}_i\rangle-\frac{1}{|\Gamma_{\ell}|}=1-\frac{1}{|\Gamma_{\ell}|}
\end{align*}
and for every $i,j \in \left[|\Gamma_{\ell}|\right]$ such that $i \neq j$,
\begin{align*}
\langle \phi^{(\ell)}(j), \phi^{(\ell)}(i)\rangle =-\frac{1}{|\Gamma_{\ell}|}.
\end{align*}
We can now define the final map $f:\left[|\Gamma_1|\right] \times \left[|\Gamma_2|\right] \times \dots \times \left[|\Gamma_d|\right]\rightarrow \ral^{\sum_{\ell=1}^d |\Gamma_{\ell}|}$. For every $c= \begin{bmatrix}c_1,c_2,\dots,c_d\end{bmatrix} \in \left[|\Gamma_1|\right] \times \left[|\Gamma_2|\right] \times \dots \times \left[|\Gamma_d|\right]$, define,
\begin{align*}
f(c) = \sqrt{\frac{1}{\sum_{\ell=1}^d \frac{|\Gamma_{\ell}|-1}{|\Gamma_{\ell}|}}}\begin{bmatrix}\phi^{(1)}(c_1),\phi^{(2)}(c_2),\dots,\phi^{(d)}(c_d)\end{bmatrix}.
\end{align*}
Note that for any $c\in \left[|\Gamma_1|\right] \times \left[|\Gamma_2|\right] \times \dots \times \left[|\Gamma_d|\right]$,
\begin{align*}
\|f(c)\|^2 &= \frac{1}{\sum_{\ell=1}^d \frac{|\Gamma_{\ell}|-1}{|\Gamma_{\ell}|}} \sum_{\ell=1}^d \|\phi^{(\ell)}(c_{\ell})\|^2\\
&= \frac{1}{\sum_{\ell=1}^d \frac{|\Gamma_{\ell}|-1}{|\Gamma_{\ell}|}} \sum_{\ell=1}^d \frac{|\Gamma_{\ell}|-1}{|\Gamma_{\ell}|}=1
\end{align*}
and for any $\hat{c} = \begin{bmatrix}\hat{c}_1,\hat{c}_2,\dots,\hat{c}_d\end{bmatrix}$ and $\tilde{c} = \begin{bmatrix}\tilde{c}_1,\tilde{c}_2,\dots,\tilde{c}_d\end{bmatrix}$,
\begin{align*}
&\langle f(\hat{c}), f(\tilde{c})\rangle \\
&= \sum_{\ell=1}^d \langle f(\hat{c}_{\ell}),f(\tilde{c}_{\ell}) \rangle\\
&= \frac{1}{\sum_{\ell=1}^d \frac{|\Gamma_{\ell}|-1}{|\Gamma_{\ell}|}}\left(\sum_{\ell:\hat{c}_{\ell} \neq \tilde{c}_{\ell}} \langle \phi^{(\ell)}(\hat{c}_{\ell}),\phi^{(\ell)}(\tilde{c}_{\ell}) \rangle+\sum_{\ell:\hat{c}_{\ell} = \tilde{c}_{\ell}} \langle \phi^{(\ell)}(\hat{c}_{\ell}),\phi^{(\ell)}(\tilde{c}_{\ell}) \rangle\right)\\
&= \frac{1}{\sum_{\ell=1}^d \frac{|\Gamma_{\ell}|-1}{|\Gamma_{\ell}|}}\left(\sum_{\ell:\hat{c}_{\ell} \neq \tilde{c}_{\ell}} -\frac{1}{|\Gamma_{\ell}|}+\sum_{\ell:\hat{c}_{\ell} = \tilde{c}_{\ell}} \left(1-\frac{1}{|\Gamma_{\ell}|}\right)\right)\\
&= \frac{1}{\sum_{\ell=1}^d \frac{|\Gamma_{\ell}|-1}{|\Gamma_{\ell}|}}\left(\left(d-H(\hat{c},\tilde{c})\right)+ \sum_{\ell=1}^d -\frac{1}{|\Gamma_{\ell}|}\right)\\
&= \frac{1}{\sum_{\ell=1}^d \frac{|\Gamma_{\ell}|-1}{|\Gamma_{\ell}|}}\left(\sum_{\ell=1}^d \frac{|\Gamma_{\ell}|-1}{|\Gamma_{\ell}|}-H(\hat{c},\tilde{c})\right)\\
&= 1-\frac{H(\hat{c},\tilde{c})}{\sum_{\ell=1}^d \frac{|\Gamma_{\ell}|-1}{|\Gamma_{\ell}|}}\\
\end{align*}
as desired.
\end{proof}

Given~\autoref{lem:map-lem} and~\autoref{lem:geometric}, we can now complete the proof of~\autoref{lem:plotkin-thm}.

\begin{proof}[Proof of~\autoref{lem:plotkin-thm}]
Let $C = \{c_1,c_2,\dots,c_m\}$ be a code from $\left[|\Gamma_1|\right] \times \left[|\Gamma_2|\right] \times \dots \times \left[|\Gamma_d|\right]$ of length $d$ and distance $t$.
Let $f: \left[|\Gamma_1|\right] \times \left[|\Gamma_2|\right] \times \dots \times \left[|\Gamma_d|\right] \rightarrow \ral^{\sum_{\ell=1}^d |\Gamma_{\ell}|}$ be the function from~\autoref{lem:map-lem}. Then for all $i$ we have that $f(c_i)$ is a unit length vector in $\ral^{\sum_{\ell=1}^d |\Gamma_{\ell}|}$. Furthermore for all $i \neq j$, we have,
\begin{align*}
\langle f(c_i), f(c_j)\rangle \le 1-\frac{t}{\sum_{\ell=1}^d \frac{|\Gamma_{\ell}|-1}{|\Gamma_{\ell}|}}.
\end{align*}

If $t = \sum_{\ell=1}^d \frac{|\Gamma_{\ell}|-1}{|\Gamma_{\ell}|}$, then for all $i \neq j$,
\[\langle f(c_i), f(c_j)\rangle \le 0.\]
So by~\autoref{lem:geometric}, we get $m \le 2\sum_{\ell=1}^d |\Gamma_{\ell}|$.

If $t > \sum_{\ell=1}^d \frac{|\Gamma_{\ell}|-1}{|\Gamma_{\ell}|}$, then for all $i \neq j$,
\begin{align*}
\langle f(c_i), f(c_j)\rangle \le 1-\frac{t}{\sum_{\ell=1}^d \frac{|\Gamma_{\ell}|-1}{|\Gamma_{\ell}|}}.
\end{align*}
Let $\epsilon = \frac{t}{\sum_{\ell=1}^d \frac{|\Gamma_{\ell}|-1}{|\Gamma_{\ell}|}}-1 = \frac{t-\sum_{\ell=1}^d \frac{|\Gamma_{\ell}|-1}{|\Gamma_{\ell}|}}{\sum_{\ell=1}^d \frac{|\Gamma_{\ell}|-1}{|\Gamma_{\ell}|}} > 0$. Then by~\autoref{lem:geometric}, we get
\begin{align*}
m \le 1+\frac{1}{\epsilon} = 1+\frac{\sum_{\ell=1}^d \frac{|\Gamma_{\ell}|-1}{|\Gamma_{\ell}|}}{t-\sum_{\ell=1}^d \frac{|\Gamma_{\ell}|-1}{|\Gamma_{\ell}|}} = \frac{t}{t-\sum_{\ell=1}^d \frac{|\Gamma_{\ell}|-1}{|\Gamma_{\ell}|}}
\end{align*}
as desired.
\end{proof}